\documentclass[lettersize,onecolumn]{IEEEtran}

\usepackage{tabularx}

\usepackage{tabularx}
\usepackage{graphicx}
\usepackage{multicol}
\usepackage{hyperref}
\usepackage{algorithmic}
\usepackage{algorithm}
\usepackage{array}
\usepackage[caption=false,font=normalsize,labelfont=sf,textfont=sf]{subfig}
\usepackage{textcomp}
\usepackage{stfloats}
\usepackage{url}
\usepackage{verbatim}
\usepackage{graphicx}

\usepackage[square,numbers,sort&compress]{natbib}
\usepackage{graphicx}

\usepackage{float}
\usepackage{amsmath}
\usepackage{hyperref} 
\usepackage{amsthm}
\usepackage{amsmath}
\usepackage{mathrsfs}
\usepackage{amssymb} 
\usepackage{bbm,dsfont}   
\usepackage{enumerate} 

\usepackage[table]{xcolor}
\usepackage[font=small,labelfont=bf]{caption}
\usepackage{multirow}
\usepackage{tikz}
\usetikzlibrary{calc}
\usepackage{pifont}%
\newcommand{\parbreak}{ \vspace{0.3 em}}
\usepackage{physics}

\usepackage{xcolor}

\usepackage{tabularx}

\usepackage{tabularx}

\newcounter{protocol}
\renewcommand{\theprotocol}{\arabic{protocol}}

\newenvironment{protocol}[1]
  {\par\addvspace{\topsep}
   \noindent
   \tabularx{\linewidth}{@{} X @{}}
    \hline
    \refstepcounter{protocol}\textbf{Protocol \theprotocol} #1 \\
    \hline}
  { \\
    \hline
   \endtabularx
   \par\addvspace{\topsep}}

\newcommand{\sbline}{\\[.5\normalbaselineskip]}
\usepackage{graphicx}
\usepackage{subfig}
\usepackage{hyperref} 
\hypersetup{
    colorlinks,
    linkcolor={blue!80!black},%
    citecolor={green!30!black},
    urlcolor={blue!80!black}
}

\newtheorem{theorem}{Theorem}
\theoremstyle{remark}	\newtheorem{lemma}[theorem]{Lemma}
\theoremstyle{remark}	\newtheorem{corollary}[theorem]{Corollary}
\theoremstyle{remark}	\newtheorem{proposition}[theorem]{Proposition}
\theoremstyle{remark} \newtheorem{definition}{Definition}

\theoremstyle{remark} \newtheorem{remark}{Remark}
\theoremstyle{remark} 

\theoremstyle{remark}	\newtheorem*{discussion*}{Discussion}

\begin{document}

\title{Network Oblivious Transfer via Noisy Channels: Limits and Capacities}

\author{Hadi Aghaee$^\dagger$, Bahareh Akhbari$^*$,~\IEEEmembership{Member,~IEEE}, Christian Deppe$^\dagger$,~\IEEEmembership{Senior Member,~IEEE}\\
\emph{$^\dagger$Institute for Communications Technology, Technische Universität Braunschweig, Braunschweig, Germany\\
$^*$Faculty of Electrical Engineering, K. N. Toosi University of Technology, Tehran, Iran}\\
Email: (hadi.aghaee, christian.deppe)@tu-bs.de, akhbari@kntu.ac.ir

\thanks{This paper was presented in part at the 2025 IEEE International Symposium on Information Theory (ISIT 2025, \cite{AghaeeDeppe2025}), and was presented in 5th Workshop on Enabling Security, Trust, and Privacy in 6G Wireless Systems at the 2025 IEEE Global Communications Conference (GLOBECOM 2025, \cite{Hadi-limit}).}
}



\maketitle

\begin{abstract}
In this paper, we study the information-theoretic limits of oblivious transfer via noisy channels. We also investigate oblivious transfer over a noisy multiple-access channel with two non-colluding senders and a single receiver. The channel is modeled through correlations among the parties, who may be honest-but-curious or, in the case of the receiver, potentially malicious. We first revisit the information-theoretic limits of two-party oblivious transfer and then extend these results to the multiple-access setting. For honest-but-curious participants, we introduce a multiparty protocol that reduces a general multiple access channel to a suitable correlation model. In scenarios with a malicious receiver, we characterize an achievable oblivious transfer rate region.
\end{abstract}

\begin{IEEEkeywords}
Oblivious transfer, Multiple access channel, Bounds for OT capacity, Limits of OT.
\end{IEEEkeywords}

\section{Introduction}
\IEEEPARstart{O}{blivious} Transfer (OT), a fundamental primitive in secure multiparty computation (MPC), plays a central role in the design of cryptographic protocols. OT is complete in the sense that, given access to an OT protocol between two parties, any function can be securely computed between them. More precisely, if OT is assumed to be a trusted primitive available between two parties, then it suffices to achieve general secure computation, even in the presence of malicious adversaries \cite{Gold}. The most basic form of OT is the $1$-out-of-$2$ OT, where a sender (Alice) holds two distinct messages and a receiver (Bob) selects one of them to receive, without revealing his choice to Alice. At the same time, Bob learns nothing about the unchosen message. A natural extension is the $1$-out-of-$m$ OT, in which Alice holds 
$m$ messages (e.g., bit strings), and Bob retrieves exactly one, without gaining information about the others or revealing his selection to Alice.

\parbreak OT was first introduced by Rabin \cite{Rabin1}. In Rabin's form, Alice sends a message to Bob with the probability of $\frac12$ while she remains oblivious to whether or not Bob received the message. This model is called the "\emph{Erasure Channel}" with the erasure probability equal $\frac12$. After that, a basic OT protocol was introduced by Even, Goldreich, and Lempel (EGL) \cite{EGL}. It is well-known that achieving multi-party security (as a basic model) in noise-free communication is impossible \cite{Kilian1}. It has been shown that achieving two-party secure communication over a noiseless channel is possible by randomness sharing \cite{Rudolf1, Winter1}. Shared randomness consists of random variables known to all communicating parties but independent of the message being transmitted \cite{Rudolf2}.

\parbreak Up to now, some primary channels have been studied for OT purposes obtained from noise. The binary symmetric channel (BSC) has made a more outstanding contribution \cite{Crepeau1, Crepeau2, Stebila}. However, all the cryptogates and channels that can be used for obtaining OT are characterized by Kilian in the case of passive adversaries \cite{Kilian2}. It should be mentioned that most of the research works in this field have been done from the perspective of the basics of cryptography, and there are a few sources that study the problem of OT from the information theory point of view. However, the OT capacity of noisy channels is generally unknown. It is known to be non-negligible if the players (sometimes we call senders and receivers as players/parties) are committed to the protocol and implement it faithfully, not turning away from additional information (\emph{honest-but-curious players}). Still, in the case of fully malicious players (active adversaries), non-zero rates have not ever been achieved \cite{Crepeau3}.

\parbreak As the first step, Nascimento and Winter study the OT capacity of noisy correlations \cite{Winter3}, wherein they characterized which noisy channels and distributions are useful for obtaining OT. In \cite{Winter1}, they showed that for honest-but-curious players, the OT capacity of noisy resources is positive by achieving a lower bound that coincides with the upper bound of \cite{Rudolf1}. The OT capacity of the binary erasure channel (BEC) is studied in \cite{Imai}, in which the OT capacity is $C_{\text{OT}}=\frac12$ with erasure probability $\frac12$ that is a property of the channel/system model in the case of honest-but-curious players and a lower bound is calculated in the case of fully malicious players. Ahlswede and Csiszar achieved a lower and upper bound on the OT capacity of noisy channels \cite{Rudolf1}. The upper bound is general and valid for every noisy channel with honest-but-curious players, while the lower bound is just valid for a special reduced version of a DMC, wherein the channel outputs are separable into two distinct sets: fully erased bits and fully received bits. It can be easily seen that the upper bound in \cite{Rudolf1} and the lower bound in \cite{Winter1} for a special reduced version coincide. An improved upper bound compared to \cite{Rudolf1} is proved in \cite{Rao} based on a monotonicity property of the tension region in the channel model. Also, symmetric private information retrieval (SPIR) has been studied under the terminology OT, with a noisy channel
between the server and the client to achieve information-theoretic security, e.g., \cite{amir1,amir2, amir3}. 

\parbreak The limits of OT are another important bottleneck in secure two-party computation, so that is addressed in various cryptographic-based papers \cite{Clev, Blum}. Some limitations relate to the capability of cheating parties, and some of them relate to the power of noise over the communication channels \cite{Damgard}. The information-theoretical limits of OT between two parties over a shared channel between three nodes are first addressed in \cite{Hadi-limit} from the same authors. In this paper, we have a more general contribution. As a motivation, we would like to draw attention to the following statement by Kurt Gödel \cite{Godel}:

\parbreak
\begin{quote}
\small
“Any consistent formal system within which a certain amount of arithmetic can be carried out is incomplete; there are statements of arithmetic which are true, but not provable within the system.”\\
\hfill — K.~Gödel, \textit{On Formally Undecidable Propositions of Principia Mathematica and Related Systems}, 1931.
\end{quote}

\parbreak As Gödel showed that no formal system can be both complete and consistent, one can say no cryptographic primitive, even as complete as OT, can escape the limits imposed by the underlying communication models and assumptions.

\parbreak In this work, we also aim to study bounds for the OT capacity of the two-user Discrete Memoryless Multiple Access Channel (DM-MAC) as one of the primary network models from the perspective of network information theory. The MAC refers to a communication scenario where multiple users send information over a shared channel to a receiver. This setup is typical in communication systems, such as cellular networks, where several devices need to communicate with a base station or an access point. We consider the following system model: Two senders, both send two independent messages (two strings) over a noisy channel to a receiver. The receiver then has to choose only one string from each sender, and the senders are assumed to be legitimate relative to each other (\emph{non-colluding senders}). This means that there is no criterion of secrecy between them. 

\parbreak This paper is organized as follows: Some seminal definitions are presented in Section \ref{Sec: Per}. Section \ref{Sec: related} is dedicated to related works and known results. The system model and main results are presented in Sections \ref{Sec: model} and \ref{Sec: OT-MAC}, respectively. We provide some examples in Section \ref{Sec: examples}, and a brief discussion in Section \ref{Sec: conc}.
\section{Preliminaries}\label{Sec: Per}
We use the well-known notation of information theory in addition to the following notations:
We use capital letters (e.g., $X$) to denote random variables, with the specific alphabet $\mathcal X$ defined by the context in which $X$ is used. Lowercase letters (e.g., $x$) represent realizations of the corresponding random variables. Bold uppercase letters (e.g., $\mathbf{X}$) denote random $n$-tuples.

\subsection*{1. Notation for Tuples}
\begin{itemize}
    
    \item Suppose $\mathcal{A}\subset \mathbb{N}$. Then $(\mathcal{A})$ represents the tuple formed by arranging the elements of $\mathcal{A}$ in increasing order:
    \[
    (\mathcal{A}) = (a_i |\, a_i \in A, \, i = 1, 2, \dots, |\mathcal{A}|), \quad \text{with } a_i < a_{i+1} \text{ for } i \geq 1
    \]
    \textit{Example}: If $\mathcal{A} = \{1, 3, 2, 9, 4\}$, then $(\mathcal{A}) = (1, 2, 3, 4, 9)$.
    
    \item For two sets $\mathcal{A} \subset \{1, 2, \dots, k\}$, and $\mathbf{X}$, we have:
    \[
    \mathbf{X}|_{\mathcal{A}} = \mathbf{X}|_{(\mathcal{A})} = \big\{ x_i|\,i\in \mathcal{A}\big\}.
    \]
    \textit{Example}: If $\mathbf{X} = (a, b, c, d, e, f, g)$ and $\mathcal{A} = \{6, 3, 1\}$, then $\mathbf{X}|_{\mathcal{A}} = (a, c, f)$.
    \item When a member $i$ is removed from a set $\mathcal{F}$, we denote the case by: 
    \[
    \mathcal{F}\setminus\{i\}.
    \]
    \textit{Example:} Given $\mathcal{F}=\{a, b, c, d,e\}$, we have: $ \{a, b, d,e\}=\mathcal{F}\setminus\{c\}$.
\end{itemize}

\subsection*{2. Markov Chains}
Random variables $X, Y, Z$ form a Markov chain $X - Y - Z$ when $X$ and $Z$ are conditionally independent given $Y$. That is, if $X \in \mathcal{X}$, $Y \in \mathcal{Y}$, and $Z \in \mathcal{Z}$, then $X - Y - Z$ implies:
\[
\forall x \in \mathcal{X}, \, \forall y \in \mathcal{Y}, \, \forall z \in \mathcal{Z}: \quad P_{X,Z|Y}(x,z|y) = P_{X|Y}(x|y) \cdot P_{Z|Y}(z|y)
\]
\subsection*{3. Erasure Count Function}
Given a sequence $y \in \{0, 1, e\}^n$, where $e$ indicates an erasure. We denote the erasure count function by:
\begin{align*}
    \Delta(y^n) &= |\{i \in \{1, 2, \dots, n\} : y_i = e \}|, \\
    \overline{\Delta}(y^n) &= |\{i \in \{1, 2, \dots, n\} : y_i \neq e \}|,
\end{align*}
where $y_i$ is a realization of $Y$.

\subsection*{4. Information Theoretic Definitions}

The min-entropy of a discrete random variable $X $ is 
\[
H_\infty(X) = \min_x \log \left( \frac{1}{P_X(x)} \right).
\]

Its conditional version is 
\[
H_\infty(X|Y) = \min_y H_\infty(X|Y=y).
\]

The zero-entropy and its conditional version are defined as 
\[
H_0(X) = \log |\{ x \in \mathcal{X} : P_X(x) > 0 \}|,
\]
and 
\[
H_0(X|Y) = \max_y H_0(X|Y=y).
\]

The statistical distance over two probability distributions $P_X $ and $P_Y $, defined over the same domain $\mathcal{X} $, is 
\[
\| P_X - P_Y \| = \frac{1}{2} \sum_{x \in \mathcal{X}} |P_X(x) - P_Y(x)|.
\]

For $\epsilon \geq 0 $, the $\epsilon $-smooth min entropy is 
\[
H_\infty^\epsilon(X) = \max_{X' : \| P_{X'} - P_X \| \leq \epsilon} H_\infty(X').
\]

Similarly,
\[
H_\infty^\epsilon(X|Y) = \max_{X'Y' : \| P_{X'Y'} - P_{XY} \| \leq \epsilon} H_\infty(X'|Y').
\]

Let $P_{UVW}$ be a probability distribution over $\mathcal{U}\times\mathcal{V}\times\mathcal{W}$ For any $\epsilon > 0 $ and $\epsilon' > 0$ it holds that \cite{Holenstein}: 

\begin{equation}\label{eq: holenstein-1}
    H_\infty^{\epsilon+\epsilon'}(U|V,W) \geq H_\infty(U|W) + H_\infty^\epsilon(V|U,W) - H_0(V|W) - \log(\frac{1}{\epsilon'}).
\end{equation}

\parbreak Also, $H^{\epsilon}_{\infty}(U,V|W)$ can be bounded from below and above as \cite{Holenstein}: 

\begin{equation}\label{eq: holenstein-2}
    H_\infty^{\epsilon+\epsilon'}(U|V,W) + H_{0}(V|W) + \log(\frac{1}{\epsilon'})\leq H^{\epsilon}_{\infty}(U,V|W)\leq H_{\infty}(U|W) + H^{\epsilon}_{\infty}(V|U,W).
\end{equation}

Combining \eqref{eq: holenstein-1} and \eqref{eq: holenstein-2}, concludes: 
\begin{equation}\label{eq: holenstein-1 + holenstein-1}
    H^{\epsilon}_{\infty}(U,V|W)\geq H_\infty(U|W) + H_\infty^\epsilon(V|U,W). 
\end{equation}
\begin{lemma}\label{lemm: min-entropy vs smooth}
    For any random variable such $U$ and $V$, we have: 
        \[H^\epsilon_\infty(U|V)-\log (\frac{1}{\epsilon})\leq H_\infty(U|V)\leq H_\infty^\epsilon(U|V).
        \]
\end{lemma}
\begin{proof}
    In Appendix \ref{App: proof-lemm: min-entropy vs smooth}.
\end{proof}

\begin{definition}\label{Renyi entropy}
    Given a random variable $X$ with alphabet $\mathcal{X}$ and probability distribution $p_X$, the \textit{Rényi entropy of order two} of a random variable $X$ is given by:
    \[
    H_2 (X) = \log_2\left(\frac{1}{P_c(X)}\right).
    \]
    where the \textit{collision probability} $P_c(X)$ is the probability that two independent trials of $X$ produce the same outcome. It is defined as:
    \[
    P_c(X) = \sum_{x \in \mathcal{X}} p_X(x)^2.
    \]
    For a given event $\mathcal{E}$, the conditional distribution $p_{X|\mathcal{E}}$ is employed to define the \textit{conditional collision probability} $P_c(X|\mathcal{E})$ and the \textit{conditional Rényi entropy of order 2}, $H_2(X|\mathcal{E})$.
\end{definition}

\begin{lemma}\cite[
Corollary 2.12]{Holenstein}\label{lemm: Holenstein}
    Let $P_{X^nY^n}$ be independent and identically distributed (i.i.d.) according to $P_{XY}$ over the alphabet $\mathcal{X}^n \times \mathcal{Y}^n$. For any $\epsilon > 0$, we have
\[
H_\infty^\epsilon(X^n | Y^n) \geq nH(X | Y) - 4\sqrt{n \log(1/\epsilon)}\log|\mathcal{X}|.
\]
\end{lemma}
\begin{proof}
    In Appendix \ref{App: proof-lemm: Holstein}.
\end{proof}
\begin{definition}\label{Hash}
A function $h: \mathcal{R} \times \mathcal{X} \to \{0,1\}^n$ is a \textit{two-universal hash function} \cite{Carter} if, for any $x_0 \neq x_1 \in \mathcal{X}$ and for $R$ uniformly distributed over $\mathcal{R}$, it holds that
\begin{equation}\label{eq: hashs}
    \Pr(h(R, x_0) = h(R, x_1)) \leq 2^{-n}.
\end{equation}

Similarly, given two independent hash functions $h_1: \mathcal{R} \times \mathcal{X} \to \{0,1\}^n$ and $h_2: \mathcal{T} \times \mathcal{Y} \to \{0,1\}^m$, for any $x_0 \neq x_1 \in \mathcal{X}$, $y_0 \neq y_1 \in \mathcal{Y}$, and for $R, T$ uniformly distributed over $\mathcal{R}$ and $\mathcal{T}$, respectively, it holds that
\begin{equation}\label{eq: hashs-double}
    \Pr(h_1(R, x_0) = h_1(R, x_1)\cap h_2(T, y_0) = h_2(T, y_1)) \leq 2^{-(n+m)}.
\end{equation}

 An example of a two-universal class is the set of all linear mappings from $\{0,1\}^n$ to $\{0,1\}^r$. 
 \end{definition}
   
A random variable $X$ over $\mathcal{X}$ is said to be $\epsilon$-close to uniform with respect to $Z$ over $\mathcal{Z}$ if
\[
\|P_{XZ} - (P_U \times P_Z)\| \leq \epsilon,
\]
where $U$ is uniformly distributed over $\mathcal{X}$.

\begin{lemma} \cite{Winter1, Jurg} (Distributed leftover hash lemma)\label{DLHL}
    Let $\epsilon > 0$, $\epsilon' \geq 0$, and let $g_i: \mathcal{T}_i \times \mathcal{X}_i \to \{0, 1\}^{n_i}$ for $1 \leq i \leq m$ be two-universal hash functions. Assume random variables $X_i$ over $\mathcal{X}_i$, $1 \leq i \leq m$, where for any subset $\mathcal{S} \subseteq \{1, 2, \dots, m\}$, and $\mathbf{X}|_\mathcal{S} = X_{\mathcal{S}(1)}, X_{\mathcal{S}(2)}, \dots, X_{\mathcal{S}(|\mathcal{S}|)}$, we have
\[
H_\infty^{\epsilon'}(\mathbf{X}|_\mathcal{S} | Z) \geq \sum_{i \in \mathcal{S}} n_i + 2 \log(1/\epsilon),
\]
where $T_1, \dots, T_m$ are uniformly distributed over $\mathcal{T}_1 \times \dots \times \mathcal{T}_m$, and are independent of $X_1, \dots, X_m$, and $Z$. Then, it holds that the tuple $(g_1(T_1, X_1), \dots, g_m(T_m, X_m))$ is $(2^m \epsilon / 2 + 2^m \epsilon')$-close to uniform with respect to $T_1, \dots, T_m, Z$.

\parbreak We also briefly note that following from \cite[Th. 17.3.3]{Cover}, it directly follows from the previously defined $X_i$, $T_i$, $g_i(T_i, X_i)$, and $Z$ that
\[
I(g_i(T_i, X_i); T_i Z) \leq -\epsilon'' \log \frac{\epsilon''}{2^{n_i} |\mathcal{Z}| |\mathcal{T}_i|},
\]
where $\epsilon'' = 2^m \epsilon / 2 + 2^m \epsilon'$ and $I(g_i(T_i, X_i); T_i Z)$ is defined as
\[
I(g_i(T_i, X_i); T_i Z) = H(g_i(T_i, X_i)) - H(g_i(T_i, X_i) | T_i Z),
\]
which $I$ and $H$ represent the Shannon mutual information and the Shannon entropy, respectively.

\end{lemma}

\parbreak Consider a Discrete Memoryless Channel (DMC) with a transition matrix 
$W=\lbrace W(y|x),x \in\mathcal{X}, y \in\mathcal{Y} \rbrace$. There are two assumptions:
\begin{enumerate}
    \item \emph{Free Resources}: Alice and Bob have unlimited computing power, independent local randomness, and access to a noiseless public communication channel for unlimited rounds.
    \item \emph{Honest-but-Curious Model}: Both parties follow the protocol honestly but may use all available information to infer what they should remain ignorant about.
\end{enumerate} 

\emph{The general two-party protocol (Figure \ref{fig: general}):}
\begin{itemize}
    \item \emph{Initial Views:} Alice and Bob start with initial knowledge or views $U'$ and $V'$, respectively.
    \item \emph{Random Experiments:} Alice generates random variable $M$, and Bob generates random variable $N$ independently of each other and $(U',V')$.
    
    \item \emph{Message Exchange:} Alice sends Bob a message $C_1$ as a function of $U'$ and $M$. Bob responds with $C_2$, a function of $V'$, $N$ and $C_1$.

    \item \emph{Alternating Messages:} In subsequent rounds, they alternately send messages $C_3, C_4,\dots, C_{2t}$, which are functions of their instantaneous views.
    \item \emph{Final Views:} At the end of the protocol, Alice's view $U$ is $(U',M,\mathbf{C})$ and Bob's view $V$ is $(V',N,\mathbf{C})$, where $\mathbf{C}=C_1,\dots, C_{2t}$.

\end{itemize}
\begin{figure*}[tb]
\includegraphics[scale=1,trim={-2.65cm 19.2cm 1cm 1.5cm},clip]{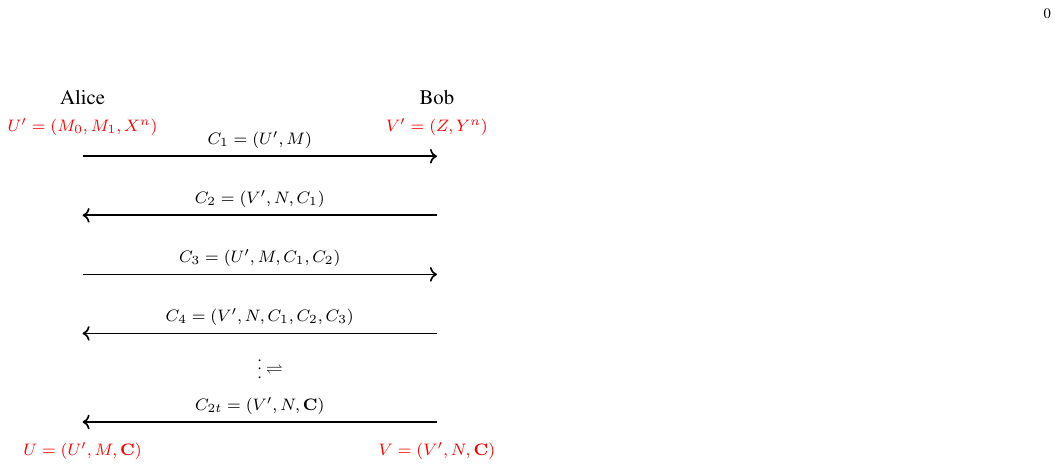} %

\caption{The general two-party protocol between Alice and Bob from the perspective of source model.}
\label{fig: general}
\end{figure*}
There are two models: The channel model and the source model.
In the source model, Alice's initial view is $U' = (M_0, M_1, X^n)$ and Bob's initial view is $V' = (Z, Y^n)$ where $(M_0, M_1)$ are binary strings and  uniformly distributed on $\lbrace0,1\rbrace^k$, and $Z\in \lbrace 0,1\rbrace$ is a binary bit. 

\parbreak In the channel model, Alice starts with her initial view $U' = (M_0,M_1)$ and Bob with his initial view $V' = Z$. In this case, Alice and Bob may perform any noisy protocol with $n$ access to the DMC with their initial views, where $M_0$, $M_1$ and $Z$ are independent, and $M_0$, $M_1$ are uniformly distributed on $\{0,1\}^k$.  

\begin{lemma}\cite[Corollary 4]{Bennet}\label{entropy hash}
Let $p_{X,Y}$ be any probability distribution, where $X \in \mathcal{X}$, $Y \in \mathcal{Y}$, and $y$ is a specific realization of $Y$. Assume that $H_2(X|Y = y) \geq c$ for some constant $c$. Let $\mathcal{K}$ be a universal class of functions mapping $X$ to $\{0,1\}^l$, and let $\kappa$ be sampled uniformly from $\mathcal{K}$. Then:
\begin{align*}
    H(\kappa(X)|\kappa, Y = y) & \geq l - \log(1 + 2^{l-c})\\
    & \geq l - \frac{2^{l-c}}{\ln 2}.
\end{align*}
\end{lemma} 

\begin{lemma}\cite[Lemma 3]{Rudolf1}\label{lemma: Rudolf}
    Let $X$, $Y$, and $Z$ be random variables defined on the finite sets $\mathcal{X}$, $\mathcal{Y}$, and $\mathcal{Z}$, respectively. For any $z_1, z_2 \in \mathcal{Z}$ with $p \triangleq \Pr[Z = z_1] > 0$ and $q \triangleq \Pr[Z = z_2] > 0$, the following inequality holds:
\[
| H(X | Y, Z = z_1) - H(X | Y, Z = z_2) | \leq 1 + 3 \log |\mathcal{X}| \sqrt{\frac{(p + q) \ln 2}{2pq} I(X,Y; Z)}.
\]
\end{lemma}

\begin{lemma}\label{lemm: test unit}
    Consider a DM-MAC with two senders and one receiver defined by transition matrix $W(y|x_1,x_2)$. For pair words $(x_1^n, x_2^n)$ and $(\tilde{x}_1^n, \tilde{x}_2^n)$ with Hamming distances $d_H(x_1^n,\tilde{x}_1^n)\geq\delta n$, $d_H(x_2^n,\tilde{x}_2^n)\geq\delta n$, such that
    \[
    \forall x_1\in\mathcal{X}_1, x_2\in\mathcal{X}_2, \forall \mathcal{P} \,\text{p.d. with}\, \mathcal{P}(x_1,x_2) = 0,  \norm{W_{x_1,x_2}-\sum_{\tilde{x}_1,\tilde{x}_2}p(\tilde{x}_1,\tilde{x}_2)W_{\tilde{x}_1,\tilde{x}_2}}\geq\eta,
    \]
    one has, with $\epsilon = \frac{\delta^4 \eta^2}{2 \lvert \mathcal{X}_1 \rvert^ 2 \lvert \mathcal{X}_2 \rvert^ 2 \lvert \mathcal{Z} \rvert}$

    \[
    W^n_{\tilde{x}_1,\tilde{x}_2}(\mathcal{T}^n_{W,\epsilon}(x_1^n,x_2^n))\leq 2 \exp (-\frac{n\epsilon^4}{2}),
    \]
    where $\mathcal{T}^n_{W,\epsilon}(x_1^n,x_2^n)$ is the set of joint typical sequences and $W_{x_1^n x_2^n} = W_{x_{1,1}x_{2,1}} \otimes W_{x_{1,2}x_{2,2}} \otimes \ldots \otimes W_{x_{1,n}x_{2,n}}$.
\end{lemma}

\begin{proof}
    In Appendix \ref{App: proof-lemm: test unit}.
\end{proof}
\begin{definition}(Reduction)
    Throughout this paper, we consider two types of reduction: (i) \emph{Cryptographic reduction:} A cryptographic reduction can be seen as a black-box that reduces one transfer mechanism $Q$ to $R$. If we reduce $Q$ to any noisy transfer, we assume the existence of a black-box that takes a bit from Alice, complements it with a certain probability, and sends the result to Bob. We also use some computational cryptographic reductions in the paper so that the cheating capability of a malicious party can be bounded within a certain range, mostly known as \emph{slightly unfair behavior}. (ii) \emph{Structural reduction:} A structural reduction means that we bound a general interaction between all parties in a communication system or a network to a weaker to limited task.  
\end{definition}
\section{Related Works and Known Results}\label{Sec: related}
\subsection{Limits of secure two-party computations}
\subsubsection{Yao's Millionaires' Problem}
One of the foundational problems that introduced the concept of secure two-party computation is Yao’s Millionaires’ Problem. In this thought experiment, proposed by Andrew Yao in 1982, two millionaires wish to determine who is richer without revealing their actual wealth to each other. Formally, each party holds a private input—Alice has $x$ and Bob has $y$—and they aim to compute the function $f(x, y) = (x > y)$ without disclosing any other information about $x$ or $y$. To solve this problem, Yao introduced the garbled circuit technique, where one party constructs a Boolean circuit representing the function and encrypts (or “garbles”) its components. The other party evaluates the garbled circuit using encrypted inputs obtained via oblivious transfer, ultimately learning only the output. This protocol demonstrated the feasibility of general secure computation between two parties and laid the groundwork for modern cryptographic protocols in this domain \cite{Yao}.

\parbreak
\subsubsection{Blum's Fair Coin Tossing}

Blum’s seminal fair coin tossing (FCT) protocol \cite{Blum} represents one of the earliest cryptographic solutions for enabling two distrustful parties to agree on a random bit in a fair manner. Consider a scenario where Alice and Bob must decide who gets to use their shared home office during a scheduling conflict. Since neither party trusts the other to flip a coin honestly, they employ a cryptographic mechanism to simulate a fair coin toss. Each party independently selects a random value and commits to it using a secure hash function, effectively locking in their choice without revealing it. Once both commitments are exchanged (e.g., over a public channel), they disclose their original values. The XOR or sum of the two values determines the winner: for example, if the result is even, Alice wins; if odd, Bob does. Crucially, the commitment phase prevents either party from altering their input based on the other’s choice, thereby ensuring fairness. If one party refuses to open their commitment or submits an invalid reveal, they automatically lose the toss, which discourages dishonest behavior. This mechanism exemplifies rational fairness in distributed systems, where the optimal strategy for both participants is to follow the protocol faithfully.

\parbreak
Although this game-theoretic interpretation of fairness aligns naturally with rational behavior, it has received comparatively little attention in the mainstream literature on secure MPC \cite{Yao, Yao86, Oded}. Instead, the field has gravitated toward a more rigorous and adversarially robust notion known as \emph{unbiasability}. Under this definition, the protocol must guarantee that no coalition of malicious participants can skew the outcome of the computation (in particular, a coin toss) regardless of their strategy. Blum’s original coin-tossing scheme does not meet this stronger requirement; while it prevents a player from biasing the outcome in their favor, it does not fully eliminate the possibility of adversarial influence. The concept of unbiasability has been extensively investigated in cryptographic research, where it is well-established that achieving it requires an honest majority among participants \cite{Oded, Rabin2}. Notably, Cleve’s impossibility result \cite{Clev} proves that in the presence of a dishonest minority comprising half or more of the parties, unbiasable coin tossing becomes fundamentally unachievable.

\parbreak
\subsubsection{Cleve’s Impossibility Result for Fair Coin Tossing}

Consider the task of two parties, Alice and Bob, jointly generating a uniformly random bit $b \in \{0,1\}$ without relying on a trusted third party. The desired properties of such a protocol are: (i) if both participants behave honestly, they agree on the same random bit $b$; and (ii) if one participant is dishonest, they cannot significantly bias the outcome in their favor. The challenge lies in designing a protocol that upholds these guarantees under adversarial conditions. Cleve’s seminal impossibility result demonstrates that in the two-party setting, achieving strong fairness is impossible. The core intuition is based on the observation that in any interactive protocol, one party can abort at an advantageous point in the execution. For example, a malicious Bob may monitor the progress of the protocol, and at each step, estimate the conditional probability distribution of the final output. If he determines that aborting the protocol at a particular round yields a more favorable distribution—e.g., increasing the likelihood that the output is a certain bit—he can halt execution strategically. This potential for selective aborting allows a dishonest party to introduce bias, even if only indirectly. Cleve formalized this insight by proving that in any two-party protocol for coin tossing, one party can skew the distribution of the final bit beyond what would be possible in an ideal, unbiased setting. Thus, perfect fairness, where no party can gain an advantage through deviation is fundamentally unachievable when one participant may behave maliciously.

\parbreak
\emph{Cleve's FCT protocol:} Bob begins the protocol by generating $r$ pairs of public and private keys: $(K_1, T_1), (K_2, T_2), \ldots,$\\$(K_r, T_r)$ where $r$ denotes the number of rounds and serves as a security parameter. Each $K_i$ is a public key associated with round $i$, and $T_i$ is the corresponding private (trapdoor) key. A trapdoor function $F$ is employed (a one-way function that is easy to evaluate but computationally difficult to invert without knowledge of the trapdoor). Additionally, Bob selects $r$ random bits $x_1, x_2, \ldots, x_r$, and sends the following to Alice:  
(i) the sequence of public keys $K_1, K_2, \ldots, K_r$, and  
(ii) the encrypted random bits $F_{K_1}(x_1), F_{K_2}(x_2), \ldots, F_{K_r}(x_r)$, where $F_{K_i}(x_i)$ denotes the application of $F$ using public key $K_i$ to encrypt bit $x_i$. For each round $i = 1, 2, \ldots, r$, the protocol proceeds as follows:  
Alice selects a random bit $y_i \in \{0,1\}$ and sends it to Bob. In response, Bob reveals the corresponding private key $T_i$. Alice then verifies that $T_i$ is indeed the correct trapdoor for $K_i$. If the verification fails, Alice replaces the unrevealed bits $x_i, x_{i+1}, \ldots, x_r$ with random values and proceeds with the rest of the protocol. Once all $r$ rounds are completed, both parties compute the XORs $x_1 \oplus y_1, x_2 \oplus y_2, \ldots, x_r \oplus y_r$, and determine the majority value among them. This majority bit serves as the final output of the protocol—a jointly generated random bit agreed upon by both parties. 

\parbreak As Cleve’s analysis shows, a malicious Alice has only limited power to influence the outcome, 
since she does not have access to Bob’s random bits $x_i$ during the execution of the protocol. 
Because of the one-wayness of the trapdoor function~$F$, and her inability to invert 
$F_{K_i}(x_i)$ without the corresponding trapdoor~$T_i$, she cannot predict or bias the XOR 
outputs before the trapdoors are revealed. 

In contrast, a malicious Bob has greater influence: he may abort the protocol prematurely by 
withholding the private key~$T_i$ at some round~$i$. In such a case, Alice replaces the unknown 
values $x_i, x_{i+1}, \ldots, x_r$ with random bits, thereby introducing uncertainty into the 
final majority computation. Nevertheless, Cleve shows that even in the worst case, such 
selective aborting can alter the majority outcome by only a small amount; specifically, the 
resulting bias is upper bounded by $O\!\left(\frac{1}{\sqrt{r}}\right)$. He also establishes a 
lower bound of $\frac{1}{2r}$, demonstrating that a nonzero bias is unavoidable in any 
finite-round protocol. 

Thus, perfect fairness—understood as zero bias—is unattainable for two-party coin-tossing 
protocols with a finite number of rounds. As the number of rounds $r$ increases, the bias 
decreases but never completely disappears.

\subsubsection{From OT to FCT}

To illustrate the connection between OT and FCT, we present a simple protocol that uses OT to generate a shared random bit $b \in \{0,1\}$ in a manner that limits the ability of either party to bias the outcome, even in the presence of malicious behavior. The protocol proceeds as follows: Alice begins by choosing two independent random bits $m_0, m_1 \in \{0,1\}$. Bob selects a random bit $z \in \{0,1\}$, which serves as his choice in the OT protocol. Through the execution of the OT, Bob learns only the bit $m_z$, while gaining no information about $m_{1-z}$. Simultaneously, Alice remains unaware of Bob’s selection bit $z$. After the OT phase, both parties locally compute values that contribute to the shared random bit. Alice calculates $b_A = m_0 \oplus m_1$, and Bob computes $b_B = m_z \oplus z$. The final output is defined as $b = b_A \oplus b_B$, which both parties can agree upon after exchanging their local values. Specifically, they verify the consistency of their computations by checking whether $b_A \oplus b_B = 0$. If the check fails, the protocol is aborted. This construction ensures that neither Alice nor Bob can significantly influence the outcome. Alice cannot bias the bit $b$ since she has no knowledge of Bob’s choice $z$, while Bob cannot bias $b$ because he learns only one of the two bits selected by Alice, and remains ignorant of the other. Nevertheless, the protocol does not achieve perfect fairness. As established by Cleve’s impossibility result, any finite-round protocol can be susceptible to slight bias: a malicious party may abort at a strategic point or manipulate their input to influence the final result. To mitigate this, the protocol can be repeated multiple times, with the final bit determined by the majority outcome. While this repetition reduces the overall bias, it cannot eliminate it entirely.

\subsection{OT from the perspective of information theory}
\subsubsection{In praise of noise}

As widely recognized in information theory, \emph{noise} refers to any random or unpredictable interference that disrupts signal transmission. This interference may take the form of physical disturbances (such as static on a radio) or conceptual ambiguities (such as misinterpretation in language). Fundamentally, noise embodies uncertainty or entropy—the unpredictable component of a system that contributes to its complexity. It represents a manifestation of disorder or randomness within communication systems, highlighting the inherent imperfections of real-world transmission processes.

\parbreak Claude Shannon’s Noisy Channel Coding Theorem~\cite{Shannon1} demonstrates that reliable communication is achievable even in the presence of noise, provided the information transmission rate remains below a critical limit known as \emph{channel capacity}. This foundational result reveals that noise, while disruptive, is not entirely detrimental; rather, it can be mitigated through appropriate encoding schemes and the strategic use of \emph{redundancy}. Shannon’s insight challenges the notion that randomness and disorder are purely negative phenomena, illustrating that they can be harnessed within structured frameworks to preserve the integrity of information.

\parbreak Within the framework of Shannon’s theory, noise is intimately linked to \emph{entropy}, which quantifies the uncertainty or unpredictability in a system. A system with high entropy exhibits greater disorder and reduced predictability. Noise elevates entropy by introducing ambiguity into the transmitted message, thereby increasing the uncertainty faced by the receiver.

\parbreak In a broader conceptual sense, entropy serves as a metaphor for disorder or chaos in the universe, with far-reaching implications for how we understand complexity and the emergence of order. Information theory’s treatment of noise thus intersects with deeper philosophical considerations about the relationship between chaos and structure. Michel Serres, in particular, explores this theme through the metaphor of the \emph{parasite} \cite{Serres}, suggesting that noise is not merely disruptive but can serve as a generative force. According to Serres, noise introduces feedback and transformation within communicative systems, acting as a catalyst for innovation and the emergence of new meanings. Importantly, one of the most profound insights drawn from noise is its role in enabling the possibility of \emph{secure communication}. In the sense of OT, now we know why OT can not be obtained from scratch and the existence of noisy resources is a crucial condition for OT, as well as other cryptographic primitives. When no noisy resource is available, the system lacks the necessary randomness to obfuscate the client's actions from the servers. This absence of noise leads to a situation of complete determinism, where all parties can eventually deduce the hidden information, violating the principles of OT.

\subsubsection{OT capacity}
As we stated before, a few research works consider the problem of OT from the information theory perspective. Nascimento and Winter \cite{Winter1, Winter3} simplified the problem of a general noisy correlation (a point-to-point channel) by reducing it to a Slightly Unfair Noisy Symmetric Basic Correlation (SU-SBC). They demonstrated that any non-perfect noisy point-to-point channel or correlation can be transformed into a Slightly Unfair Noisy Channel/Correlation (SUNC/SUCO), and that any SUNC/SUCO can be used to implement a specific SU-SBC. Ultimately, they showed that in this reduced framework, it is possible to achieve $\binom{m}{1}\text{-OT}^k$ ($1$-out-of-$m$ OT with strings' length equal $k$) at a positive rate, assuming that the sender behaves in an honest-but-curious manner. These papers are significant for several reasons:

\begin{enumerate}
    \item They introduce the concept of the oblivious transfer capacity of a DMC, defined as the supremum of all achievable rates $R$ such that $\frac{k}{n} \geq R - \gamma$, where $\gamma > 0$ and $n$ is the number of channel uses. A positive number $R$ is an achievable OT rate for a given DMC if for $n \to \infty$ there exist $(n, k)$ protocols with $\frac{k}{n}\to R$ such that protocols are correct and secure. 
    \item They also address a malicious model where a malicious player can deviate arbitrarily from the channel statistics in up to $\delta n$ instances. In such cases, the deviating player will be detected by the other party with a certain probability.
\end{enumerate}

Ahlswede and Csisz\'ar analyzed the problem under the honest-but-curious model \cite{Rudolf1}. They derived a general upper bound for the OT capacity of a point-to-point DMC, which aligns with the lower bound established by Nascimento and Winter \cite{Winter1}. Furthermore, when reduced to a specific erasure channel, they demonstrated a lower bound on the OT capacity. Thus, it remains uncertain what the exact OT capacity of a noisy DMC is in general. In practice, for general channels, a potential way is to first convert the channel into a Generalized Erasure Channel (GEC) via alphabet extension and erasure emulation, followed by the application of a general construction for GEC. As a starting point, we begin with a protocol by Ahlswede and Csisz\'ar originated from the general two-party protocol of Section \ref{Sec: Per}:

\parbreak
\parbreak
\textit{Two-Party OT Protocol \cite{Rudolf1}:} Consider the following two-party secure computation in the sense of OT: Alice has two strings $M_0$ and $M_1$ and aims to send them over the noisy point-to-point channel $\mathcal{W}: \mathcal{X}\rightarrow \mathcal{Y}$ to Bob. Bob has to choose one of them by inputting a bit of $Z\in\{0,1\}$ to the channel. Alice should be unaware of the unselected string, while Bob has only one string at the end of the protocol $(M_{Z})$. Suppose the main channel is an erasure channel assisted by a noiseless channel with unlimited capacity. The OT capacity in this setup is given by $\min(p, 1 - p)$, where $p$ is the erasure probability \cite{Rudolf1, Imai}. Let $r < \min(p, 1 - p)$. The protocol by Ahlswede and Csiszár \cite{Rudolf1}, based on a technique originally introduced for a BSC in \cite[Sec. 6.4]{Crepeau2}, proceeds as follows: Alice starts by transmitting a sequence $ \mathbf{X} = X^n \sim \text{Bernoulli}(\frac12)$ of i.i.d. bits over the channel. Bob observes the channel output $\mathbf{Y}$. Let $E$ denote the set of indices at which $\mathbf{Y}$ is erased, and let $\overline{E}$ represent the set of indices where $\mathbf{Y}$ is not erased. If $|E| < nr$ or $|\overline{E}| < nr$, Bob aborts the protocol, as there are not enough erased or unerased bits to complete the protocol. From $\overline{E}$, Bob randomly selects a subset $S_Z$ of cardinality $nr$. From $E$, Bob randomly picks a subset $S_{\overline{Z}}$, also of size $nr$. Bob then shares the sets $S_0$ and $S_1$ with Alice via the public channel, where $S_0$ and $S_1$ are either $S_Z$ and $S_{\overline{Z}}$, respectively or vice versa. Alice cannot determine which of $S_0$ and $S_1$ corresponds to $E$ (erased positions) and which to $\overline{E}$ (non-erased positions) due to the independent nature of channel erasures. Using $S_0$ and $S_1$, Alice computes keys $\mathbf{X}|_{S_0}$ and $\mathbf{X}|_{S_1}$ and employs these keys to encrypt her strings, which she sends to Bob over the public channel as $M_0 \oplus \mathbf{X}|_{S_0}$ and $M_1 \oplus \mathbf{X}|_{S_1}$. Bob, who only knows the sequence corresponding to $\overline{E}$, can decrypt only one of these encrypted strings, depending on whether $S_Z = S_0$ or $S_Z = S_1$. This enables Bob to retrieve one of the two keys, $M_Z$, while he learns nothing about the other key, $M_{\overline{Z}}$. If $X$ is not uniformly distributed over $\{0,1\}$, the strings $\mathbf{X}|_{S_j}, j\in\{0,1\}$ are not directly suitable as encryption keys. They need to be transformed into binary strings of length $k < nr$ with a distribution approximately uniform over $\{0,1\}^k$. It is well-known that for any $\delta > 0$, when $n$ is large, there exists a mapping $\kappa : \{0,1\}^{nr} \rightarrow \{0,1\}^k$ with $k = n(H(X) - \delta)$ such that $k - H(\kappa(X^n))$ is exponentially small. 

\parbreak In \cite{chou}, the author extends the above protocol to pairwise oblivious transfer over a noiseless binary adder channel involving two senders and one receiver, assuming they are non-colluding and honest-but-curious. Each sender has two strings, and Bob has to choose one string from each sender while the unselected strings are hidden from his view. In this system, the output is defined as the sum of the inputs, $Y=X_1+X_2$, commonly referred to as the Binary Erasure Multiple Access Channel (BE-MAC). This channel uniquely determines the inputs except when they differ, in which case one can not identify the inputs with certainty, effectively resulting in an erasure. Specifically, erasures occur in two out of four possible input scenarios. The OT capacity of this channel is shown to be $R_1+R_2\leq \max_{P_{X_1}P_{X_2}} H(X_1,X_2|Y)=\frac12$.

\section{System Model}\label{Sec: model}

We assume that the availability of noise is provided in two main forms. Also, we consider the main OT channel with two senders and one receiver:

\begin{enumerate}
    \item \textit{Discrete Memoryless MAC:} A two-user MAC $W : \mathcal{X}_1 \times \mathcal{X}_2 \to \mathcal{Y}: (\mathcal{X}_1\times \mathcal{X}_2, p(y|x_1,x_2), \mathcal{Y}) $, connecting three parties, Alice-1, Alice-2 and Bob, which can be used $n $ times. For an input sequence $x_i^n = x_{i, 1} x_{i,2} \ldots x_{i,n} $, the output distribution over $\mathcal{Y}^n $ is given by:
   \[
   W_{x_1^n x_2^n}^{n} = W_{x_{1,1}x_{2,1}} \otimes W_{x_{1,2}x_{2,2}} \otimes \ldots \otimes W_{x_{1,n}x_{2,n}}.
   \]

    \item \textit{i.i.d. Realizations:} A tuple of random variables $(X_1, X_2, Y)$, where Alice-$i$ sends $X_i$ and Bob receives $Y $. The distribution of these variables is given by $P_{X_1X_2Y} $, defined over the finite sets $\mathcal{X}_i $ and $\mathcal{Y} $.
\end{enumerate}

\parbreak In both cases, the alphabets $\mathcal{X}_1, \mathcal{X}_2  $, and $\mathcal{Y} $ are finite.

\parbreak A key concept when analyzing noisy channels is the idea of \textit{redundant symbols}~\cite{Winter-BC}. We have the following definition for DM-MACs, presented in \cite{Winter1} for the point-to-point channel.

\begin{definition}\label{redundant: MAC-channel}
    A \textit{two-sender DM-MAC} $W(y|x_1, x_2) $, characterized by its conditional probability distribution $W(y|x_1, x_2) $ of the output $y $ given inputs $x_1 $ from Alice-1 and $x_2 $ from Alice-2, is said to be \textit{nonredundant} if none of its output distributions $W_{x_1, x_2}(y) $ (induced by fixed inputs $(x_1, x_2) $) can be expressed as a convex combination of the other output distributions. Formally, this means:
\[
\forall \,i\in\mathcal{T}\setminus\{(x_1,x_2)\}, \forall P(x_1, x_2) \text{ such that } P\{i\in\mathcal{T}\setminus\{(x_1,x_2)\}\} = 0, \quad W_{i\in\mathcal{T}\setminus\{(x_1,x_2)\}} \neq \sum_{x_1, x_2} P(x_1, x_2) W_{x_1, x_2},
\]
for any possible distinct input pairs $\mathcal{T}=\{(x_1, x_2),(x_1', x_2),(x_1, x_2'), (x_1', x_2')\} \in \mathcal{X}_1 \times \mathcal{X}_2$.
\begin{itemize}
    \item \textit{Geometric Interpretation:}  
    In geometric terms, each output distribution $W_{x_1, x_2} $ is a distinct extremal point of the polytope $\mathcal{W} = \text{conv}\{ W_{x_1, x_2} : (x_1, x_2) \in \mathcal{X}_1 \times \mathcal{X}_2 \} $, where $\mathcal{X}_1 $ and $\mathcal{X}_2 $ are the input alphabets of Senders 1 and 2, respectively. The polytope $\mathcal{W} $ represents the convex hull of all output distributions over the probability simplex on the output alphabet $\mathcal{Y} $.
    \item \textit{Constructing a Nonredundant MAC:}  
    To construct a nonredundant version of the MAC, $\overline{W}(y|x_1, x_2) $, we can remove all input pairs $(x_1, x_2) $ for which the output distribution $W_{x_1, x_2} $ is not extremal. This results in a reduced set of input pairs for which $W_{x_1, x_2} $ forms the set of extremal points of $\mathcal{W} $. The original MAC can still be simulated using the reduced MAC by reconstructing the removed distributions $W_{x_1, x_2} $ as convex combinations of the extremal distributions from $\overline{W} $.
    
    \parbreak This process ensures that the MAC retains its original operational capacity while simplifying its representation by eliminating redundancy in its input space.
    \end{itemize}
\end{definition}
A more intuitive definition based on the correlations is presented below. 
\begin{definition}\label{redundant: MAC-correlation}
      
    Consider a two-user DM-MAC characterized by random variables $X_1 $, $X_2 $ (inputs from the two senders) and $Y $ (output), with joint distribution/correlation $P(X_1, X_2, Y) $. The correlation is said to be \textit{nonredundant} if:
\begin{itemize}
    \item For any possible distinct input pairs $\mathcal{T}=\{(x_1, x_2),(x_1', x_2),(x_1, x_2'), (x_1', x_2')\} \in \mathcal{X}_1 \times \mathcal{X}_2$:
    \[\Pr\big\{Y|(X_1,X_2) = i\in\mathcal{T}\big\}\neq\Pr\big\{Y|(X_1,X_2)=j\in\mathcal{T}\setminus\{i\}\big\}.
    \]
    \item Symmetrically, the above condition also applies to redundancy in $X_1 $ (for fixed $X_2 $) or $X_2 $ (for fixed $X_1 $), similarly to $Y $.
\end{itemize}
   
    \textit{Resolving Redundancy:}  
    If there is redundancy, the MAC can be made nonredundant by collapsing indistinguishable input pairs $(x_1, x_2) $ that fail the above inequality into a single equivalent pair. Similarly, redundant output symbols $y_1, y_2 $ can be merged into one.

    \parbreak
    \textit{Geometric Interpretation:}  
    In the MAC context, redundancy occurs when the joint distribution $P(Y|X_1, X_2) $ does not map injectively over distinct input combinations $(x_1, x_2) $. This can be resolved by projecting to the set of unique conditional distributions $\Pr(Y | X_1, X_2) $, thereby defining an equivalent nonredundant MAC.

\end{definition}
\begin{definition}\label{perfectness: MAC}
    For a DM-MAC with input random variables $X_1$ and $X_2$, and output $Y$, we define \textit{perfect correlation} as follows:
    
    \parbreak The MAC is perfectly correlated if, given $Y$, both $X_1$ and $X_2$ can be determined with certainty.
    This implies:
    \[
    H(X_1, X_2 | Y) = 0.
    \]

    \parbreak Similarly, a MAC can be called \textit{perfect} if its joint output distributions (conditioned on input pairs) have mutually disjoint support. Specifically:
    
    \parbreak Given the output $Y$, the pair $(X_1, X_2)$ is uniquely determined.
    Formally, this means that for all $y \in \mathcal{Y}$, there exists at most one pair $(x_1, x_2)$ such that $P_{Y | X_1, X_2}(y | x_1, x_2) > 0$.
\end{definition}
As proved in \cite{Winter1} for point-to-point channels, a perfect DM-MAC (even after removing the redundancy) cannot be used for oblivious transfer, even against passive adversaries. This relates to the concept of noise and the emergence of noisy resources for cryptographic intents. We know that the noise produces uncertainty or entropy. It is the unpredictable aspect of a system that adds complexity. Conceptually, it can be seen as the manifestation of disorder or randomness in communication systems, emphasizing the non-perfect nature of the real world. Noise plays a crucial role in securing communication by hindering an eavesdropper's ability to extract meaningful information from the transmitted message. In a noisy channel, the inherent noise limits the amount of information an unauthorized party can access, regardless of their computational power. This concept is central to Wyner's wiretap channel model \cite{Wyner}, which demonstrates how noise can be leveraged to ensure that the legitimate receiver decodes the message accurately. In contrast, experiencing additional noise, the eavesdropper cannot gather sufficient information to reconstruct the message. As is clear, a perfect channel can be simulated by a noiseless channel where the input(s) can be obtained with certainty from the channel output(s). Obviously, such channels cannot be used for cryptographic intents with unconditional security (information-theoretic secrecy).

\parbreak As is proved in \cite{Imai,Rudolf1}, the OT capacity of the point-to-point erasure channel (BEC) with erasure probability $\frac12$ is equal $\frac12$. Here, we want to investigate whether the OT is possible over a special BE-MAC. We introduce the channel as a correlation between the senders (Alice-1, Alice-2) and the receiver (Bob). Before that, we delve deeper into the \emph{unfairness} in the channel/correlation model. Damg\aa rd \emph{et. al.} \cite{Damgard}, introduced unfairness so that an unfair player could change the communication channel parameters within a certain range. In \cite{Winter1}, this concept is limited so that an unfair player who deviates from the channel statistics in $\delta n$ positions will be caught by the other party with probability $\geq 1-C_1 \exp (-C_2 \delta^2 n)$, where $C_1$ and $C_2$ are two small positive numbers. There are two senders in our channel model. The concept of unfairness can also be extended to the MAC model. To control the fairness of other players in the last $n$ rounds, all players have access to a test unit. When both senders send a pair symbol to Bob over the channel, Bob will ask for the input symbols with probability $\frac12$. He will tell both senders his output when he has received the senders' response. If Alice-$i$ is unfair, she tells her input wrong; if Bob is unfair, he tells his output wrong. The test between players is to check out the samples after $n$ uses of the channel for joint typicality relative to $P_{X_1 X_2 Y}$.

\begin{definition}\label{BE-MAC-special case}
    Consider a DM-MAC characterized by random variables $X_1 $, $X_2 $ uniformly distributed over $\{0,1\}$ (inputs from the two senders) and $Y $ (output), with joint distribution $P(X_1, X_2, Y) $. Let $p = \frac12$ and $\mathcal{Y}$ of $Y$ be partitioned into two disjoint sets: $\mathcal{Y}=\mathcal{E}_{10}\cup\mathcal{E}_{01}$ of non-zero probability under the distribution of $Y$. The channel/correlation has the following properties: 
    \begin{itemize}
        \item For all $y_{10}\in\mathcal{E}_{10}, y_{01}\in\mathcal{E}_{01}$ and $x_i\in\{0,1\}, i\in\{1,2\}$,
        \[
        \Pr \{Y = y_{10}|X_1=x_1,X_2=x_2\} = \Pr \{Y = y_{01}|X_1=x_1,X_2=x_2\} = \frac{1}{2}.
        \]
        \item $\mathcal{E}_{10}$ is the set of received pairs $(X_1=x_1, X_2= e)$, and $\mathcal{E}_{01}$ is the set of received pairs $(X_1=e, X_2= x_2)$, where $e$ is an erased bit. 
    \end{itemize}
\end{definition}

\parbreak Now, we present a more general SU-SBC over DM-MAC. 

\begin{definition}\label{SBC in MAC}[$\text{SU-SBC}_{p, W, W'}$]
    Consider a DM-MAC characterized by random variables $X_1 $, $X_2 $ uniformly distributed over $\{0,1\}$ (inputs from the two senders) and $Y $ (output), with joint distribution $P(X_1, X_2, Y) $. Let $0<p<1$ and $\mathcal{Y}$ of $Y$ be partitioned into five disjoint sets: $\mathcal{Y}=\mathcal{U}_{11}\cup\mathcal{U}_{10}\cup\mathcal{E}\cup\mathcal{U}_{01}\cup\mathcal{U}_{00}$ of non-zero probability under the distribution of $Y$. Note that $\mathcal{E} = \mathcal{E}_{00}\cup\mathcal{E}_{10}\cup\mathcal{E}_{01}$. A \textit{symmetric basic correlation (SBC)} over this channel can be defined as follows: 
    \begin{itemize}
        \item For all $y\in \mathcal{E}_{00}$, $\Pr\{Y=y|X_1=1, X_2=1\} = \Pr\{Y=y|X_1=1, X_2=0\} =\Pr\{Y=y|X_1=0, X_2=1\} =\Pr\{Y=y|X_1=0, X_2=0\}=(1-p)^2$.
        \item For all $y\in \mathcal{E}_{10}\cup\mathcal{E}_{01}$, $\Pr\{Y=y|X_1=1, X_2=1\} = \Pr\{Y=y|X_1=1, X_2=0\} =\Pr\{Y=y|X_1=0, X_2=1\} =\Pr\{Y=y|X_1=0, X_2=0\}=2p(1-p)$.
        \item (Symmetry) For all $y_{11}\in \mathcal{U}_{11}, y_{10}\in \mathcal{U}_{10}, y_{01}\in \mathcal{U}_{01}, y_{00}\in \mathcal{U}_{00}$, and $x_i\in\{0,1\}, i\in\{1,2\}$

        \[
        \Pr \big\{Y=y_{i\in \mathcal{T}'}|(X_1X_2)=j\in\mathcal{T}'\} = \Pr \{Y=y_{i'\in\mathcal{T}'\setminus\{i\}}|(X_1X_2)=j'\in\mathcal{T}'\setminus\{j\}\big\},
        \]
        for $\mathcal{T}' = \{00,10,01,11\}$.
        \item (Non-redundancy) For all $y_{11}\in \mathcal{U}_{11}, y_{10}\in \mathcal{U}_{10}, y_{01}\in \mathcal{U}_{01}, y_{00}\in \mathcal{U}_{00}$, and $x_i\in\{0,1\}, i\in\{1,2\}$
        
        \[
        \Pr\big\{Y=y_{i\in\mathcal{T}'}|(X_1X_2) = i\}>\Pr\{Y=y_{i}|(X_1X_2)=j\in\mathcal{T}'\setminus\{i\}\big\},
        \]
        for $\mathcal{T}' = \{00,10,01,11\}$.
        \item $\Pr\{Y\in\mathcal{E}\} = 1-p^2$.
    \end{itemize}
\end{definition}

\parbreak From the senders' point of view, it looks like uniform inputs to a DM-MAC, while for Bob, it looks like the output of a distinguishable mixture of three channels: a complete erasure MAC $W''(y|x_1,x_2): \{0,1\}\times\{0,1\}\to\mathcal{E}_{00}$, a partial erasure channel $W'(y|x_1,x_2): \{0,1\}\times\{0,1\}\to\mathcal{E}_{01}\cup \mathcal{E}_{10}$ (Definition \ref{BE-MAC-special case}) which erases either $x_1$ or $x_2$, and a channel $W:\{0,1\}\times\{0,1\}\to\mathcal{U}_{11}\cup\mathcal{U}_{10}\cup\mathcal{U}_{01}\cup\mathcal{U}_{11}$, with conditional probabilities $W(y|x_1,x_2) = \frac{1}{p^2}\Pr\{Y=y|X_1=x_1, X_2=x_2\}$. If Bob finds $y \in \mathcal{E}_{00}$, he has no information at all about the inputs. If Bob finds $y \in \mathcal{E}_{01}$, he has no information at all about Alice-1's input. Similarly, if he finds $y \in \mathcal{E}_{10}$, he has no information at all about Alice-2's input. For $y \in \mathcal{U}_{i}$, where $i \in \mathcal{T}' = \{00, 01, 10, 11\}$, he has a (more or less weak) indication that $x_1x_2 = i$, as the likelihood for $x_1x_2 = j \in \mathcal{T}' \setminus \{i\}$ is smaller.

\parbreak The correlation is clearly fully characterized by $p, W$, and $W'$. Hence, we denote this distribution as $\text{SBC}_{p,W,W'}$. If it is used slightly unfairly, we denote it as $\text{SU-SBC}_{p,W,W'}$. We reduce the $\text{SU-SBC}_{p,W,W'}$ to the case in which both inputs are erased or both of them are decoded. The new SU-SBC is demonstrated by $\text{SU-SBC}_{p,W}$ since the sub-channel $W'$ defined in Definition \ref{BE-MAC-special case} is removed. 

\begin{definition}\label{SBC in MAC-reduced}[$\text{SU-SBC}_{p, W}$]
    Consider a DM-MAC characterized by random variables $X_1 $, $X_2 $ uniformly distributed over $\{0,1\}$ (inputs from the two senders) and $Y $ (output), with joint distribution $P(X_1, X_2, Y) $. Let $0<p<1$ and $\mathcal{Y}$ of $Y$ be partitioned into five disjoint sets: $\mathcal{Y}=\mathcal{U}_{11}\cup\mathcal{U}_{10}\cup\mathcal{E}\cup\mathcal{U}_{01}\cup\mathcal{U}_{00}$ of non-zero probability under the distribution of $Y$. A \textit{symmetric basic correlation (SBC)} over this channel can be defined as follows: 
    \begin{itemize}
        \item For all $y\in \mathcal{E}$, $\Pr\{Y=y|X_1=1, X_2=1\} = \Pr\{Y=y|X_1=1, X_2=0\} =\Pr\{Y=y|X_1=0, X_2=1\} =\Pr\{Y=y|X_1=0, X_2=0\}=1-p$.
        \item (Symmetry) For all $y_{11}\in \mathcal{U}_{11}, y_{10}\in \mathcal{U}_{10}, y_{01}\in \mathcal{U}_{01}, y_{00}\in \mathcal{U}_{00}$, and $x_i\in\{0,1\}, i\in\{1,2\}$,

        \[
        \Pr \big\{Y=y_{i\in \mathcal{T}'}|X_1X_2=j\in\mathcal{T}'\big\} = \Pr \big\{Y=y_{i'\in\mathcal{T}'\setminus\{i\}}|X_1X_2=j'\in\mathcal{T}'\setminus\{j\}\big\},
        \]
        for $\mathcal{T}' = \{00,10,01,11\}$.
        \item (Non-redundancy) For all $y_{11}\in \mathcal{U}_{11}, y_{10}\in \mathcal{U}_{10}, y_{01}\in \mathcal{U}_{01}, y_{00}\in \mathcal{U}_{00}$, and $x_i\in\{0,1\}, i\in\{1,2\}$,
        
        \[
        \Pr\big\{Y=y_{i\in\mathcal{T}'}|X_1X_2 = i\big\}>\Pr\big\{Y=y_{i}|X_1X_2=j\in\mathcal{T}'\setminus\{i\}\big\},
        \]
        for $\mathcal{T}' = \{00,10,01,11\}$.
        \item $\Pr\{Y\in\mathcal{E}\} = 1-p$.
    \end{itemize}
\end{definition}

From now on, we work only with $\text{SU-SBC}_{p,W}$. We demonstrate how to reduce the general case of non-perfect MACs to $\text{SU-SBC}_{p,W,W'}$. Since redundant channels or correlations can always be transformed into nonredundant ones, we will henceforth assume that all noisy resources under consideration are nonredundant.

\begin{proposition}
    Given a non-perfect noisy DM-MAC $W: \mathcal{X}_1 \times \mathcal{X}_2 \to \mathcal{Y}$, it can be utilized to obtain a certain SUCO. Similarly, a noisy correlation shared among Alice-1, Alice-2, and Bob can be used to implement an SUCO.
\end{proposition}
\begin{proof}
        For any given distribution $P_{X_i}$ possessed by Alice-$i$, $i\in \{1,2\}$, the channel generates a joint distribution $P_{X_1X_2Y}$. Alice-$i$ sends independent realizations of $X_i$ according to her probability distribution over the channel. Note that the input joint distribution is $P_{X_1X_2}(x_1,x_2) = P_{X_1}(x_1)P_{X_2}(x_2)$ due to independence of senders' probability distributions. For each received pair message, Bob will ask, with probability $\frac12$, for the input symbols, and after receiving, he will tell the senders his received message. If one or both senders are cheaters, then they will give wrong information to Bob, and if Bob is a cheater, then he will give wrong information to the senders. Lemma \ref{lemm: test unit} shows that the probability of cheating (in more than $\delta n$ positions) tends to zero as $n\to\infty$. In other words, a cheater can not deviate from the channel statistics in $\delta n$ positions without being detected because the joint typicality test fails. This shows that any non-perfect noisy MAC $W:\mathcal{X}_1\times\mathcal{X}_2\to\mathcal{Y}$ can be used to obtain a certain SUCO/SUNC.    
    \end{proof}
\begin{proposition}
    Given a non-perfect SUCO $P_{X_1X_2Y}$, one can use it to implement a certain $\text{SU-SBC}_{p,W}$.
\end{proposition}
\begin{proof}
    Given that the slightly unfair correlation is non-perfect, it follows that after removing all redundancy, $H(Y | X_1, X_2) > 0$. Therefore, there exist at least four distinct pairs of symbols $\big((a,a) \neq (b,b) \neq (a,b) \neq (b,a)\big) \in \mathcal{X}_1 \times \mathcal{X}_2$ such that:
    \begin{align*}
        \big\{ y : &\Pr(Y = y | (X_1,X_2) = (a,a)) > 0 \text{ and }\\
        &\Pr(Y = y | (X_1,X_2) = (a,b)) > 0 \text{ and }\\
        & \Pr(Y = y | (X_1,X_2) = (b,a)) > 0 \text{ and }\\
        &\Pr(Y = y | (X_1,X_2) = (b,b)) > 0 \big\} \neq \emptyset,
    \end{align*}
    since otherwise, the distributions' equivocation (equivocation about the inputs given the output) would be zero.

    Our protocol to achieve $\text{SU-SBC}_{p,W}$ closely follows the approach outlined in \cite{Winter1, Crepeau1}. It uses the source coding protocol $P_{X_1X_2Y}$ twice for two independent realizations. Let:
    \begin{align*}
        \Pr\big((X_1,X_2)(X_1,X'_2)\cup(X_1,X_2)(X'_1,X_2)\cup(X_1,X_2)(X'_1,X'_2)\in \mathcal{L} \big) = \frac{\alpha}{3},
    \end{align*}
    where $\mathcal{L} = \big\lbrace (a,a)(b,b), (a,a)(a,b), (a,a)(b,a),(b,b)(a,a), (b,b)(a,b), (b,b)(b,a), (a,b)(b,b), (a,b)(a,a), (a,b)(b,a),$ $(b,a)\linebreak(a,a), (b,a)(a,b), (b,a)(b,b)\big\rbrace $. As the protocol progresses, if  $\{(X_1,X_2)(X_1,X'_2),(X_1,X_2)(X'_1,X_2),(X_1,X_2)(X'_1,X'_2)\}\linebreak\notin\mathcal{L}$, the senders inform Bob of the value and discard the pair. By the law of large numbers and the i.i.d. distribution of $\mathcal{L}$, we know that the probability of $i \in \mathcal{L}$ is equal $\frac{\alpha}{12}$, then fewer than $\frac{12}{\alpha} (n + \epsilon n)$ (for some $\epsilon > 0$) realizations of $P_{X_1X_2Y}$ are necessary to achieve $n$ realizations of $\text{SBC}_{p,W}$. This constitutes a symmetric basic correlation because the probability of occurrence is equal for each $i\in \mathcal{L}$. Furthermore, Bob’s alphabet $\mathcal{Y} \times \mathcal{Y}$ is partitioned into $\mathcal{E}, \mathcal{U}_{00}, \mathcal{U}_{10}, \mathcal{U}_{01}, \mathcal{U}_{11}$, satisfying the definition:

\begin{itemize}
    \item $\mathcal{E}$ includes all pairs $(z,z') (z,z')$, where $(z,z') (z,z')$ has a positive probability.
    \item $\mathcal{U}_{00}$, $\mathcal{U}_{10}$, $\mathcal{U}_{01}$ and $\mathcal{U}_{11}$ consist of transposes of one another (i.e., swapping the entries) and cannot be empty. If they were, at least one of the members in the above set would be redundant.
\end{itemize}

Bob cannot behave unfairly beyond what is provided by the SUCO $P_{X_1X_2Y}$. However, the senders can introduce bias by attempting actions such as repeating unfair pairs $(x_1,x_2)(x_1,x'_2), (x_1,x_2)(x'_1,x_2)$ or $(x_1,x_2)(x'_1,x'_2)$ for a good pair $i\in\mathcal{L}$. Suppose the senders persist with such strategies for $\delta n$ times (where $\delta > 0$). In that case, Bob can detect their bias by approximating a typicality test, using Lemma \ref{lemm: test unit} and an extended version of \cite[Lemma 6]{Winter1}. This aligns with the definition of $\text{SBC}_{p,W}$.

\end{proof}
    
\section{Main Results and Proofs}\label{Sec: OT-MAC}
\subsection{OT limits}
\subsubsection{Information-Theoretic Formulation of OT}

In the foundational works \cite{Blundo, Nikov}, the information-theoretic security definition for the \emph{receiver's privacy} requires that the sender's final view, denoted by $U = (U', R_A, \mathbf{C})$, be statistically independent of the receiver's choice bit $z \in \{0,1\}$. Here, $U'$ represents the sender's initial state (including inputs), $R_A$ denotes the sender's private randomness, and $\mathbf{C}$ captures the complete public communication between the sender and receiver. However, as highlighted by Crépeau in \cite{Crepeau}, such a strong condition is generally unachievable in many practical settings. In particular, when there is a known dependency between the parties' inputs, the sender's view will unavoidably correlate with the receiver’s input. This is because the sender's input $X$, which is inherently part of her view, may be statistically dependent on the receiver’s input $Z$. To address this, a more appropriate formulation requires that the sender's view be conditionally independent of the receiver’s input, given the sender’s input. Formally, this is expressed via the Markov chain $(U', R_A, \mathbf{C}) \rightarrow X \rightarrow Z$, which implies the conditional mutual information satisfies $I(U', R_A, \mathbf{C}; Z \mid X) = 0$. Nevertheless, Ahlswede and Csiszár \cite{Rudolf1} adopt a slightly relaxed criterion, aligning with the interpretation in \cite{Blundo, Nikov}, where the privacy condition is written as $I(M_0, M_1, X^n, R_A, \mathbf{C}; Z) \rightarrow 0$. This formulation is weaker, as it allows for some residual dependence under the assumption that it diminishes asymptotically.

\parbreak
Formulating the security definition for the \emph{sender} presents additional challenges beyond those encountered in defining receiver security \cite{Blundo, Nikov}. Several issues commonly arise in this context:

\parbreak
Notably, in \cite{Brassard1, Darco}, the model constrains a malicious receiver to alter their input in a purely deterministic manner. That is, the effective input $Z'$—representing the bit ultimately learned by the receiver—is required to be a deterministic function of their original input $Z$. This stands in contrast to the ideal functionality of oblivious transfer, where a malicious receiver is permitted to choose $Z'$ probabilistically, potentially based on arbitrary strategies. Such a deterministic restriction limits the adversary's capabilities in the simulation, potentially weakening the security guarantee for the sender in realistic adversarial settings.

\parbreak
In scenarios where a malicious receiver may alter their effective input from $z$ to $z'$ (with $z, z' \in \{0,1\}$), \cite{Brassard2} allows $Z'$ to depend on the sender's input $X$. Such dependency is not permitted in the ideal model, where parties' inputs and outputs are generated independently of each other through a trusted functionality. Furthermore, the security definition requires that the view of the malicious receiver—denoted by $(V', R_B, Y^n, \mathbf{C})$—be conditionally independent of the sender’s input $X$, given the receiver's original input $Z$ and their output $M_{Z^{'}}$. This corresponds to the Markov chain: $
X \rightarrow (Z, M_{Z^{'}}) \rightarrow (V', R_B, Y^n, \mathbf{C}),
$
which implies: $I(V', R_B, Y^n, \mathbf{C}; X \mid Z, M_{Z^{'}}) = 0$.

\parbreak
However, a more accurate formulation would account for the potential dependence of the receiver's strategy on $Z'$, leading to the stronger Markov condition: $X \rightarrow (Z, Z', M_{Z^{'}}) \rightarrow (V', R_B, Y^n, \mathbf{C})$.
Ahlswede and Csiszár \cite{Rudolf1} propose an alternative criterion under the honest-but-curious model—where Alice and Bob follow the protocol faithfully—that requires:
\[
I(Z, R_B, Y^n, \mathbf{C}; X) = I(Z, R_B, Y^n, \mathbf{C}; M_{\overline{Z}}),
\]
and study its asymptotic behavior in the i.i.d. setting. Here, Alice’s effective input is the message pair $(M_0, M_1)$, while Bob’s input is $Z$. Although this formulation simplifies analysis under passive adversaries, it is overly permissive and can allow certain protocols that are insecure under stronger, malicious models.

\parbreak
In the context of OT, a malicious receiver may alter their effective input from $z$ to $z'$ for $z, z' \in \{0,1\}$. In \cite{Brassard2}, the variable $Z'$—representing the receiver’s effective choice—may depend on the honest sender’s input $X$, which deviates from the ideal functionality where such a dependency is disallowed. Moreover, the receiver's view, denoted by $(V', R_B, Y^n, \mathbf{C})$, is required to be conditionally independent of the sender’s input $X$, given the receiver’s original input $Z$ and the corresponding output $M_{Z^{'}}$. This condition corresponds to the Markov chain: $X \rightarrow (Z, M_{Z^{'}}) \rightarrow (V', R_B, Y^n, \mathbf{C})$, or equivalently, $I(V', R_B, Y^n, \mathbf{C}; X \mid Z, M_{Z^{'}}) = 0$. However, a more accurate and robust formulation accounts for the dependency of the adversary's strategy on the full pair $(Z, Z')$, implying the stronger Markov chain:
$X \rightarrow (Z, Z', M_{Z^{'}}) \rightarrow (V', R_B, Y^n, \mathbf{C})$. Ahlswede and Csiszár \cite{Rudolf1} propose a different formulation in the honest-but-curious setting, where both parties follow the protocol faithfully. In this model, Alice's input is the message pair $(M_0, M_1)$, and Bob's input is $Z$. Their security criterion requires that: $I(Z, R_B, Y^n, \mathbf{C}; X) \triangleq I(Z, R_B, Y^n, \mathbf{C}; M_{\overline{Z}})$, and they study its asymptotic behavior under the i.i.d. assumption. While this approach simplifies the analysis for semi-honest participants, it is too permissive under malicious adversaries and may inadvertently allow protocols with security vulnerabilities.

\subsubsection{OT over a point-to-point noisy channel}
\begin{definition}\cite{Crepeau}\label{thm1'}
    A protocol $\Pi$ securely computes $\binom 21-\text{OT}^{k}$ perfectly if and only if for every pair of algorithms $\Bar{A}= (\Bar{A}_1,\Bar{A}_2)$ that is admissible for protocol $\Pi$ and for all inputs $(X(m_0,m_1), Z)$ and auxiliary input $\mathbf{C}$ (representing the total information transmitted over the noiseless public channel), $\Bar{A}$ produces outputs $(U,V)$ such that the following conditions are satisfied: 
    \begin{itemize}
        \item (Correctness) If both players are honest, then $(U,V) = (\triangle, M_Z)$.
        \item (Security for Alice) If Alice is honest, then we have $U = \triangle$, and there exists a random variable $Z'$, such that
        \[
        I(U';Z'|\mathbf{C},Z) = 0, \,\,\, \text{and}\,\,\, I(U';V|\mathbf{C},Z,Z',M_{Z^{'}}) = 0.
        \]
        \item (Security for Bob) If Bob is honest, then we have
        \[
        I(U;Z|\mathbf{C},U') = 0,
        \]
    \end{itemize}
    where $X, U, V, \mathbf{C}$ are random variables and $\mathbf{C}$ is an additional auxiliary input (total public transmission) available to both players but assumed to be ignored by honest players. Note that $U$ (Alice's final view after completing OT protocol) and $V$ (Bob's final view after completing OT protocol) are $(U', R_A, X^n,  \mathbf{C})$ and $(V', R_B, Y^n, \mathbf{C})$, respectively, where $U'$ and $V'$ in the latter form are Alice's and Bob's initial view, respectively.
\end{definition}

\parbreak
\begin{theorem}\label{thm: impossibility-pp}
Consider a noisy DMC between Alice and Bob. OT with perfect secrecy (unconditional security) is impossible over a DM-MAC if one of the players is an unbounded cheater.
\end{theorem}
\begin{proof}
    In Appendix \ref{App: proof-thm: impossibility-pp}.
\end{proof}

\subsubsection{OT over a DM-MAC}
\parbreak
\begin{definition}\label{thm2'}
    A protocol $\Pi$ securely computes $\binom 21-\text{OT}^{k_1,k_2}$ between non-colluding players perfectly if and only if for every algorithm $\Bar{A}= (\Bar{A}_1,\Bar{A}_2,\Bar{A}_3)$ that is admissible for protocol $\Pi$ and for all inputs $(X_i(M_{i0},M_{i1}),Z_i)$ and auxiliary input $\mathbf{C}\triangleq (\mathbf{C}_1,\mathbf{C}_2)$ (total public transmission), $\Bar{A}$ produces outputs $(U_1,U_2,V)$ such that the following conditions are satisfied:
    \begin{itemize}
        \item (Correctness) If all players are honest, then $(U_1,U_2,V) = (\triangle,\triangle, (M_{1Z_1}, M_{2Z_2}))$.
        \item (Security for Alice-$i$) If both senders are honest; then we have $U_i = \triangle, i\in\{1,2\}$, and there exist two random variables (malicious Bob's effective inputs) $Z'_i,i\in\{1,2\}$, such that,
        \begin{align*}
        I(X_i;Z'_i|\mathbf{C},Z_i)_{i\in\{1,2\}} &= 0, \\
        I(X_i;V|\mathbf{C},Z_i,Z'_i,M_{iZ^{'}_i})_{i\in\{1,2\}} &= 0.
        \end{align*}
        
        \item (Security for Bob) If Bob is honest, then we have:
        \[
        I(U_i;Z_i|\mathbf{C},X_i)_{i\in\{1,2\}} = 0,
        \]
    \end{itemize}
    where Alice-$i$'s final view and Bob's final view are $U_i = (U'_i, R_{A_i}, \mathbf{C}_i)$ and $V = (V', R_B, Y^n, \mathbf{C}_i)$, respectively, and $U'_i$ and $V'$ in the latter form are Alice-$i$'s and Bob's initial views, respectively. 
\end{definition}

\parbreak
\begin{remark}[Structural reduction]\label{rem}
 We have assumed two reductions from the general OT over MAC: 1- Both senders act independently, and there is no dependency between their chosen messages. 2- Both senders act honestly or maliciously.\footnote{If we assume that Alice-$i$ is honest while Alice-$\overline{i}$ is malicious, then OT with perfect secrecy is possible between Alice-$i$ and honest Bob due to Cleve's impossibility.}   
\end{remark}

Now, we aim to demonstrate that the reduced version cannot be realized from an information-theoretical perspective.
\parbreak
\begin{theorem}\label{thm: impossibility-MAC}
    Consider the OT setting over the DM-MAC described in Definition \ref{thm2'} and Remark \ref{rem}. OT with perfect secrecy (unconditional security) is impossible over the DM-MAC if Bob (Or both senders) is an unbounded cheater.  
\end{theorem}
\begin{proof}
    In Appendix \ref{App: proof-thm: impossibility-MAC}.
\end{proof}
\subsection{OT over the DM-MAC}
\parbreak Consider Alice-1, Alice-2, and Bob connected by a DM-MAC. Alice-$i$, $i\in\{1,2\}$ possesses two strings of bits $m_{i0}, m_{i1}$, each of which with length $k_i$. They could have OT between themselves (Alice-1$\to$Bob, Alice-2$\to$Bob) as follows: They send their strings over a noisy DM-MAC to Bob, and Bob has to choose one string from the Alice-1's strings ($m_{1Z_1}$) and one string from Alice-2's strings ($m_{2Z_2}$) by inputting two bits to the channel $Z_1,Z_2$. The unselected messages should be kept hidden from Bob's view, while the selected messages should be kept hidden from Alice's view. After completing a protocol, Bob gets $m_{1Z_1},m_{2Z_2}$ based on his choices, while Alice-$i$ gets nothing $\Delta$.

\parbreak Now we consider a reduction from the general non-perfect MAC (Figure \ref{fig: both-MAC}-$(a)$) to a MAC made up of $\text{SBC}_{p,W}$ with non-independently distributed noise\footnote{The concept of noise could be different in OT over multi-user channels compared to a point-to-point channel. Noise can act independently over the users' links to the receiver in a MAC so that the described $\text{SU-SBC}_{p,W}$ can be defined as $\text{SU-SBC}_{p_1, p_2,W}$, where the channel $W$ in the second one has the conditional probabilities $W(y|x_1,x_2) = \frac{1}{p_1p_2}\Pr\{Y=y|X_1=x_1, X_2=x_2\}$. Also, the sets $\mathcal{E}_{01}$ and $\mathcal{E}_{10}$ could not be unified since their probabilities are not the same. ($1-p_i$ is the erasure probability for Alice-$i$ messages)}.

\begin{figure}[t]
    \centering
    \begin{minipage}{0.45\textwidth}
        \centering
        \includegraphics[scale=0.9,trim={1cm 20.4cm 10cm 1.5cm},clip]{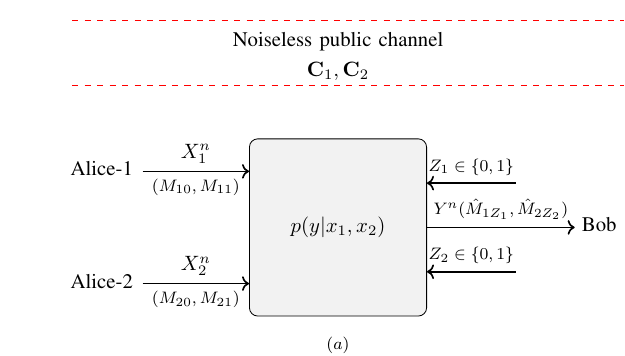}
    
    \end{minipage}
    \hfill
    \begin{minipage}{0.45\textwidth}
        \centering
        \includegraphics[scale=0.9,trim={1cm 20cm 10cm 1.5cm},clip]{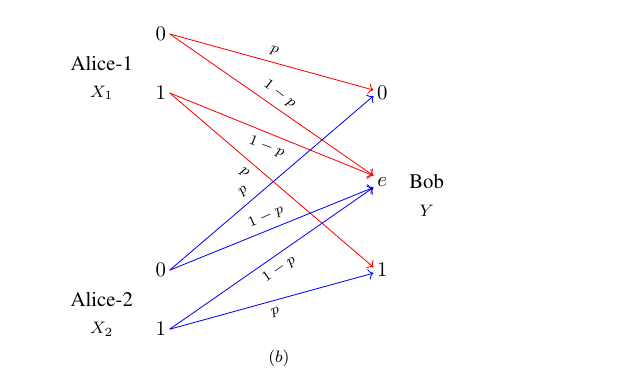}
        
    \end{minipage}
    \caption{$(a)$-OT over a general non-perfect MAC. $(b)$-The non-perfect MAC model reduced to the $\text{SU-SBC}_{p,W,W'}$ described in Definition \ref{SBC in MAC}.}
    \label{fig: both-MAC}
\end{figure}

\begin{definition}\label{def: OT-MAC}
    Let $n, k_1, k_2 \in \mathbb{N} $. An $(n, k_1, k_2)$ protocol involves interaction among Alice-1, Alice-2, and Bob via the setup illustrated in Figure~\ref{fig: both-MAC}-$(b)$. At each time step $l = 1, 2, \ldots, n $, Alice-$i$ transmits a bit $X_{i,l} $ through the MAC. Users alternately exchange messages over a noiseless public channel in multiple rounds, both prior to each transmission and after the final transmission at $l = n $. While the number of rounds may vary, it is finite. Each user's transmission is determined by a function of their input, private randomness, and all prior messages, channel inputs, or channel outputs observed. A positive rate pair $(R_1, R_2)$ is said to be an achievable OT rate pair for the DM-MAC if for $n\to\infty$ there exist $(n, k_1, k_2)$ protocols satisfying $\frac{k_i}{n} \to R_i$ such that for \emph{non-colluding} parties, the asymptotic conditions \eqref{goals: MAC-nColl-1}--\eqref{goals: MAC-nColl-3} hold:
\begin{align}\label{goals: MAC-nColl-1}
    \lim_{n\rightarrow \infty}\text{Pr}\,\left[(\hat{M}_{1Z_1}, \hat{M}_{2Z_2})\neq (M_{1Z_1}, M_{2Z_2})\right] = 0,\\
    \label{goals: MAC-nColl-2}
    \lim_{n\rightarrow \infty}I (M_{1\overline{Z}_1},M_{2\overline{Z}_2} ; V) = 0,\\
    \label{goals: MAC-nColl-3}
    \lim_{n\rightarrow \infty}I (Z_i ; U_i)_{i \in \{1,2\}} = 0,
\end{align}
where the final view of Alice-$i$ and Bob are $U_i = (M_{i0},M_{i1},R_{A_i},X_i^n, \textbf{C})$ and $V = (Z_1, Z_2, R_B, Y^n, \textbf{C})$, respectively, and $\mathbf{C}\triangleq (\textbf{C}_1, \textbf{C}_2)$. The closure of all achievable OT$^{k_1,k_2}$ rate pairs is called the OT capacity region of the MAC. Condition \eqref{goals: MAC-nColl-1} says that Bob correctly learns both $M_{iZ_i}, i\in \{1,2\}$ with negligible probability of error. Condition \eqref{goals: MAC-nColl-2} says that Bob gains negligible information about the unselected messages, and conditions \eqref{goals: MAC-nColl-3} says that Alice-$i$ gains negligible information about Bob's choices $Z_i$. 
\end{definition}
\parbreak In this section, we present a protocol of symmetric basic correlations (SU-SBC) in the case of honest-but-curious players. For a malicious receiver, the protocol achieves an achievable rate region.

\parbreak As previously discussed, from Alice's perspective the SU--SBC is equivalent to supplying a uniform input to a binary MAC. 
In contrast, from Bob's viewpoint the resulting observation can be represented as a distinguishable mixture of three channels: 
(i) a complete erasure channel with probability $(1-p)^2$, 
(ii) a partial erasure channel with probability $2p(1-p)$, and 
(iii) a DM--MAC with probability $p^2$.

\parbreak Now, we present an OT protocol over noisy DM-MAC reduced to the noisy correlation (Figure \ref{fig: both-MAC}-$(b)$) with non-colluding honest-but-curious parties. We assume that $p<\frac12$ and that Alice-1, Alice-2 and Bob are in possession of $n$ realizations $\text{SU-SBC}_{p,W}$. Alice-$i$'s data and Bob's data are denoted by $X_{i,1},\cdots,X_{i,n}$ for $i\in\{1,2\}$ and $Y_1,\cdots,Y_n$, respectively. The output of the channel $W$ is denoted by $Z_1,\cdots,Z_n$ on inputs $X_{i,1},\cdots,X_{i,n}$ for $i\in\{1,2\}$. The output of the partial erasure channel $W':X_1X_2\to\mathcal{E}_{10}\cup\mathcal{E}_{01}$ is denoted by $Z'_1,\cdots,Z'_n$ on inputs $X_{i,1},\cdots,X_{i,n}$ for $i\in\{1,2\}$. 

\parbreak Let $s_i$ and $r_i$, for $i\in\{1,2\}$ be four parameters and $h_{ij} : \mathcal{R}_{ij}\times\mathcal{X}_i^n\to\{0,1\}^{s_in}$, and $\kappa_{ij} : \mathcal{T}_{ij}\times\mathcal{X}_i^n\to\{0,1\}^{r_in}$, $i\in\{1,2\}, j\in\{0,1\}$ be two-universal hash functions. The $\binom{2}{1}$-OT$^{k_1,k_2}$ rate pair is defined as $(r_1 = \frac{k_1}{n}, r_2 = \frac{k_2}{n})$. Note that the strings are ciphered by $\kappa_{ij} : \mathcal{T}_{ij}\times\mathcal{X}_i^n\to\{0,1\}^{r_in}$, $i\in\{1,2\}, j\in\{0,1\}$ and $h_{ij} : \mathcal{R}_{ij}\times\mathcal{X}_i^n\to\{0,1\}^{s_in}$ are used as privacy amplification.

\parbreak We now prove a lemma that says the above protocol without the hashes $h_{ij}: \mathcal{R}_{ij}\times\mathcal{X}_i^n\to\{0,1\}^{s_in}$ is still correct and secure.
\begin{lemma}\label{lemm: conditions}
    Protocol 1 satisfies the conditions \eqref{goals: MAC-nColl-1}-\eqref{goals: MAC-nColl-3} even without privacy amplification.
\end{lemma}

\begin{proof}
    In Appendix \ref{App: proof-lemm: conditions}.
\end{proof}

\parbreak As a central block of proving the upper bound on OT rate, we want to consider the problem of pre-generated secret key agreement over MAC by public discussion. We use the results in the proof of Theorem \ref{thm: OT capacity}.

\begin{protocol}{OT over noisy DM-MAC (Definition \ref{SBC in MAC-reduced})}
\label{protocol: Main}
    \textit{Parameters:} \begin{itemize}
        \item $p<\frac12$
        \item $\gamma>0, \eta>0$. 
        \item The rate of the protocol is $r_i-\gamma = \frac{k_i}{n}$, and $k_i$ is the length of Alice-$i$'s strings. 
    \end{itemize}
    \sbline
    \textit{Goal.} Alice-$i$ sends two strings $m_{i0}$ and $m_{i1}, i \in \{1,2\}$ to Bob. At the end of the protocol, Bob gets $m_{iZ_i}, Z_i\in \{0,1\}, i\in\{1,2\}$ based on his choices while Alice-$i$ gets nothing $\Delta$.
    \sbline
    \textit{The protocol:}
      \begin{enumerate}
        \item Alice-$i$ transmits an $n$-tuple $\mathbf{X}_i= X_i^n$ of i.i.d. Bernoulli($\frac12$) bits over the reduced noisy DM-MAC. 
    
        \item Bob receives the $n$-tuple $\mathbf{Y}= Y^n$. Bob forms the sets 
        \[
        \overline{E}_i:=\left\{j\in\{1,2,\cdots,n\}:Y_j=(\hat{x}_i\neq e, \hat{x}_{\overline{i}})\right\}
        \]
        \[
        E_i:=\left\{j\in\{1,2,\cdots,n\}:Y_j=(\hat{x}_i=e, \hat{x}_{\overline{i}})\right\}
        \]
        If $\lvert \overline{E}_i\rvert<r_in$ or $\lvert E_i \rvert<r_in$, Bob aborts the protocol.
        \item Bob creates the following sets:
        \[
        S_{iZ_i}\sim \text{Unif}\left\{A\subset\overline{E}_i: \lvert A \rvert=(p-\eta)n\right\}
        \]
        \[
        S_{i\overline{Z}_i}\sim \text{Unif}\left\{A\subset E_i: \lvert A \rvert=(p-\eta)n\right\}
        \]
        Bob reveals $S_{i0}$ and $S_{i1}$ to Alice-$i$ over the noiseless public channel (only the description of the sets). 
        
        \item Alice-$i$ randomly and independently chooses functions $\kappa_{i0}$, $\kappa_{i1}$, $h_{ij} : \mathcal{R}_{ij}\times\mathcal{X}_i^n\to\{0,1\}^{s_in}, i\in\{1,2\}, j\in\{0,1\}$ from a family $\mathcal{K}$ of two-universal hash functions:
        \[
        \kappa_{i0}, \kappa_{i1}: \{0,1\}^{r_in}\rightarrow \{0,1\}^{k_i}
        \]
        Alice-$i$ finally sends the following information to Bob over the noiseless public channel: 
        \[
        h_{i0}, h_{i1}, \kappa_{i0}, \kappa_{i1}, M_{i0}\oplus\kappa_{i0}(\mathbf{X}_i|_{S_{i0}}), M_{i1}\oplus\kappa_{i1}(\mathbf{X}_i|_{S_{i1}})
        \]
        and the total randomness in her possession.
        \item Bob knows $\kappa_{iZ_i}$, $\mathbf{X}_i|_{S_{iZ_i}}$ one can decode $M_{iZ_i}$. At first, he computes $\hat{\mathbf{X}}_i|_{S_{iZ_i}}$ so that: 
        \begin{itemize}
            \item $(\hat{\mathbf{X}}_1|_{S_{1Z_1}}, \hat{\mathbf{X}}_2|_{S_{2Z_2}})$ and $\mathbf{Y}|_{S_{1Z_1}, S_{2Z_2}}$ are $\epsilon$-conditional typical according to $W$, that is, $\mathbf{Y}|_{S_{1Z_1}, S_{2Z_2}}\in \mathcal{T}^n_{W, \epsilon}(\hat{\mathbf{X}}_1|_{S_{1Z_1}}, \hat{\mathbf{X}}_2|_{S_{2Z_2}})$;
            \item $h_{i}(R_{iZ_i}, \mathbf{X}_i|_{S_{iZ_i}}) = h_{i}(R_{iZ_i}, \hat{\mathbf{\mathbf{X}}}_i|_{S_{iZ_i}}), i\in\{1,2\}$;
            
            If there is more than one such $(\hat{\mathbf{X}}_1|_{S_{1Z_1}}, \hat{\mathbf{X}}_2|_{S_{2Z_2}})$ or none, Bob outputs an error.
        \end{itemize}
        
      \end{enumerate}
    
\end{protocol} 

\subsection{Secret Key Agreement over MAC by Public Discussion} 
Secret key agreement over a simple point-to-point channel is initially addressed in \cite{Maurer} and \cite{Rudolf2}. In \cite{Maurer}, the author considers that Alice aims to share a secret key with Bob in the presence of a passive wiretapper. A noiseless one-way channel assists the main noisy channel with unlimited capacity. Also, it is assumed that Eve can receive all messages sent over the public channel without error, but the wiretapper can not change the messages without being detected. 

\parbreak A generalization of the above setting can be found in \cite{Gohari1, Gohari2}. Also, the problem of secret key agreement over MAC is studied in \cite{Aref}, in which the authors assumed that each transmitter plays the role of wiretapper for the other transmitter. This is the MAC with confidential message \cite{Poor}. Consider the following settings: Alice-1 and Alice-2 aim to share secret keys with Bob in the presence of a passive wiretapper. Alice-$i$, Bob, and Eve govern the input $X_i$, the output $Y$, and $E$, respectively. Furthermore, a public channel with unlimited capacity is available for one-way communication from the senders to Bob. Both resources (MAC and the public channel) are available for communication, but not at the same time. First, the senders are allowed to use the MAC $n$ times, and then they can use the public channel only once.

\parbreak Alice-$i$ and Bob use a protocol in which, at each step, Alice-$i$ sends a message to Bob depending on $X_i$ and all the messages previously received by Bob and Bob sends a message to Alice-$i$ depending on $Y$ and all the messages previously received by Alice-$i$. Let $C_{i,k}, i\in\{1,2\}, k \in \{\text{odd}\}$ the use of public channel from Alice-$i$ to Bob, and $C_{i,k}, i\in\{1,2\}, k \in \{\text{even}\}$ the use of public channel from Bob to the senders. Also, all legitimate parties can benefit from randomness statistically independent of $X_i, Y$, and $E$.

\parbreak At the end of round-$l$, Alice-$i$ computes a key $S_i$ as a function of $X_i$ and $C_i^l \triangleq [C_{i,1}, \cdots, C_{i,l}], l\, \text{is even}$, and Bob computes keys $S'_i$ as a function of $Y$ and $C_i^l$. The goal is to maximize $H(S_1,S_2)$ while $S'_1$ and $S'_2$ agree with very high probability and Eve has negligible information about $S_i$ and $S'_i, i \in \{1,2\}$:

\begin{align}\label{SKA-MAC-1}
    H(C_{i,k}|C_i^{k-1},X_1,X_2) & = 0, \qquad\text{For odd}\, k\\\label{SKA-MAC-2}
    H(C_{i,k}|C_i^{k-1},Y) & = 0, \qquad\text{For even}\,k\\\label{SKA-MAC-3}
    H(S_i|C_i^l,X_i) & = 0,\\\label{SKA-MAC-4}
    H(S_1,S_2|C_1^l,C_2^l ,X_1,X_2) & = 0,\\\label{SKA-MAC-6}
    H(S'_1,S'_2|C_1^l,C_2^l, Y) & = 0,\\\label{SKA-MAC-7}
    \text{Pr}\lbrace (S_1,S_2) \neq (S'_1,S'_2) \rbrace & \leq \epsilon,\\\label{SKA-MAC-8}
    I(S_i;C_i^l, E)& \leq \delta_i,\\
    \label{SKA-MAC-9} I(S_1,S_2;C_1^l,C_2^l, E)& \leq \delta^{\prime},
\end{align}
for $i \in \{1,2\}$, where $\epsilon, \delta_i$, and $\delta^{\prime}$ are specified small numbers. 

\begin{lemma}\label{lemm1}

    \begin{equation}
        H(S_1,S_2|S'_1,S'_2) \leq h(\epsilon_1) +h(\epsilon_2) + \log_2(|S_1|-1) + \log_2(|S_2|-1).
    \end{equation}
        
\end{lemma}
\begin{proof}

    \begin{equation*}
        H(S_1,S_2|S'_1,S'_2) = H(S_1|S'_1,S'_2) + H(S_2|S_1,S'_1,S'_2) \leq H(S_1|S'_1) + H(S_2|S'_2),
    \end{equation*}
    where the last inequality is due to the fact that $S_i$ is independent of $S_{\overline{i}}$ and $S'_{\overline{i}}$. According to Fano's lemma, conditions (\ref{SKA-MAC-7}), (\ref{SKA-MAC-8}) implies that $H(S_i|S'_i)\leq h(\epsilon_i) + \log_2(|S_i|-1)$. $|S_i|$ is the number of distinct values that $S_i$ takes on with non-zero probability. Note that $H(S_1,S_2|S'_1,S'_2)\rightarrow 0$ as $\epsilon_i, \epsilon' \rightarrow 0$.
\end{proof}

\begin{theorem}\label{thm: key agreement}
    For every key agreement protocol satisfying \eqref{SKA-MAC-1}-\eqref{SKA-MAC-9}, we have:
    \begin{align*}
        H(S_1) & \leq I(X_1;Y|C_1^l, X_2, E) + H(S_1|S'_1) + I(S_1;C_1^l, E),\\
        H(S_2) & \leq I(X_2;Y|C_2^l, X_1, E) + H(S_2|S'_2) + I(S_2;C_2^l, E),\\
        H(S_1,S_2) & \leq I(X_1,X_2;Y|C_1^l,C_2^l, E) + H(S_1,S_2|S'_1,S'_2) + I(S_1,S_2;C_1^l,C_2^l, E),
    \end{align*}
    and for the case with constant $E$: 
    \begin{align*}
        H(S_1) & \leq I(X_1;Y|C_1^l, X_2) + H(S_1|S'_1) + I(S_1;C_1^l),\\
        H(S_2) & \leq I(X_2;Y|C_2^l, X_1) + H(S_2|S'_2) + I(S_2;C_2^l),\\
        H(S_1,S_2) & \leq I(X_1,X_2;Y|C_1^l,C_2^l) + H(S_1,S_2|S'_1,S'_2) + I(S_1,S_2;C_1^l,C_2^l).
    \end{align*}
\end{theorem}
\begin{proof}
    In Appendix \ref{App: proof-thm: key agreement}.
\end{proof}

\begin{corollary}\label{cor1}
        The secret key rate region $\mathcal{R}_{\text{SK}}$ of $X_1, X_2$, and $Y$ with respect to the constant random variable $E$ is upper bounded as:
        \begin{align}%
\mathcal{R}_{\text{SK}}&\subseteq
\left\{ \begin{array}{rl}
  (R_1,R_2) \,:\;
	R_1 &\leq  \max_{P_{X_1}P_{X_2}} I(X_1;Y|X_2)\\
  R_2 & \leq \max_{P_{X_1}P_{X_2}} I(X_2;Y|X_1)\\
	R_1+R_2 & \leq \max_{P_{X_1}P_{X_2}} I(X_1,X_2;Y)
	\end{array}
\right\} \,,
\label{eq:inRetx_In1}
\end{align}
        
        for some distribution $p(x_1)p(x_2)$ on $\mathcal{X}_1\times\mathcal{X}_2$.
\end{corollary}

\begin{proof}
    In Appendix \ref{App: proof-cor: cor1}.
\end{proof}

Now, we present the main theorems of the paper. The OT achievable rate pairs and OT capacity are defined in Definition \ref{def: OT-MAC}. The OT capacity region is denoted by $\mathcal{C}_{\text{OT}}$. 
\begin{theorem}[honest-but-curious players]\label{thm: OT capacity}
    The OT capacity region of DM-MAC for honest-but-curious parties, is such that (upper bound):
    \begin{align}%
\mathcal{C}_{\text{OT}}&\subseteq
\left\{ \begin{array}{rl}
  (R_1,R_2) \,:\;
	R_1 &\leq  \max_{P_{X_1}P_{X_2}}\min\{ I(X_1;Y|X_2), H(X_1|Y)\}\\
  R_2 & \leq \max_{P_{X_1}P_{X_2}}\min\{ I(X_2;Y|X_1),H(X_2|Y)\}\\
	R_1+R_2 & \leq \max_{P_{X_1}P_{X_2}}\min\{I(X_1,X_2;Y),H(X_1,X_2|Y)\}
	\end{array}
\right\} \,,
\label{eq:inRetx_In2}
\end{align}
for some distribution $p(x_1)p(x_2)$ on $\mathcal{X}_1\times\mathcal{X}_2$.
\end{theorem}

\begin{proof}
    In Appendix \ref{App: proof-thm: OT Capacity-upper}.
\end{proof}
\begin{theorem}[honest-but-curious players]\label{thm: OT capacity2}
    The OT capacity region of DM-MAC reduced to the noisy $\text{SU-SBC}_{p,W}$ for honest-but-curious parties, is:
    \begin{align}%
\mathcal{C}_{\text{OT}}& =
\left\{ \begin{array}{rl}
  (R_1,R_2) \,:\;
	R_1 &\leq  \max_{P_{X_1}P_{X_2}} I(X_1;Y|X_2)\\
  R_2 & \leq \max_{P_{X_1}P_{X_2}} I(X_2;Y|X_1)\\
	R_1+R_2 & \leq \max_{P_{X_1}P_{X_2}} I(X_1,X_2;Y)
	\end{array}
\right\} ,
\end{align}
        for some distribution $p(x_1)p(x_2)$ on $\mathcal{X}_1\times\mathcal{X}_2$.
\end{theorem}
\begin{proof}
    In Appendix \ref{App: proof-thm: OT capacity}.
\end{proof}
A rate pair $(R_1,R_2)$ is achievable if there exists
a sequence of $\binom{2}{1}$-OT$^{k_1,k_2}$ protocols satisfying \eqref{goals: MAC-nColl-1}-\eqref{goals: MAC-nColl-3}
such that\linebreak $\lim_{n\to\infty}(\frac{k_1}{n},\frac{k_2}{n})= (R_1, R_2)$. 
\begin{theorem}[malicious Bob]\label{thm: OT capacity-malicious}
    An achievable rate region for OT $(\mathcal{R}_{\text{OT}})$ over the DM-MAC reduced to the noisy $\text{SU-SBC}_{p,W}$ with malicious Bob is:
    \begin{align}%
\mathcal{R}_{\text{OT}}(P_{Y|X_1,X_2})& =\bigcup_{p_{X_1} p_{X_2} }
\left\{ \begin{array}{rl}
  (R_1,R_2) \,:\;
	R_1 & <  \frac12\max_{P_{X_1}P_{X_2}}\{ I(X_1;Y|X_2)+I(X_1;X_2|Y)\}\\
  R_2 & <  \frac12\max_{P_{X_1}P_{X_2}}\{ I(X_2;Y|X_1)+I(X_1;X_2|Y)\}\\
	 R_1+R_2 & <  \frac12\max_{P_{X_1}P_{X_2}} I(X_1, X_2;Y)
	\end{array}
\right\} ,
\end{align}
        for some distribution $p(x_1)p(x_2)$ on $\mathcal{X}_1\times\mathcal{X}_2$.
\end{theorem}
If $X_1\to Y\to X_2$ is a Markov chain, then the above achievable rate region is simplified to: 

 \begin{align}%
\mathcal{R}_{\text{OT}}(P_{Y|X_1,X_2})& =\bigcup_{p_{X_1} p_{X_2} }
\left\{ \begin{array}{rl}
  (R_1,R_2) \,:\;
	R_1 & <  \frac12\max_{P_{X_1}P_{X_2}}I(X_1;Y|X_2)\\
   R_2 & <  \frac12\max_{P_{X_1}P_{X_2}} I(X_2;Y|X_1)\\
	 R_1+R_2 & <  \frac12\max_{P_{X_1}P_{X_2}} I(X_1, X_2;Y)
	\end{array}
\right\} ,
\end{align}
\begin{proof}
    In Appendix \ref{App: proof-thm: OT capacity-malicious}.
\end{proof}

\section{Examples}\label{Sec: examples}
\subsection{Noiseless Binary
Adder Channel (BAC)}\label{ex1}
Consider the following OT setting over the noiseless BAC~\cite{chou}: two senders, Alice-1 and Alice-2, each hold two binary strings. Bob must select one string from each sender, with the unchosen strings remaining hidden. Communication occurs via the BAC, where the channel output is given by the sum of the inputs, $ Y = X_1 + X_2 $. This channel deterministically reveals the inputs when they are equal ($0$ or $1$), but produces an erasure symbol ($e$) when they differ, making the output ambiguous. Erasures thus occur in two of the four possible input pairs.  

For the non-colluding, honest-but-curious case, the OT capacity satisfies  
\[
R_1 + R_2 \le \max_{P_{X_1}P_{X_2}} H(X_1, X_2 \mid Y) = \frac{1}{2}.
\]  
If Bob cheats, for example, by altering non-erased outputs into the erasure set $\mathcal{E}$ to extract information about the unselected strings, Alice-1 and Alice-2 can detect such manipulation via a typicality test and abort the protocol. However, any such biasing capability would violate perfect security, as per Cleve’s impossibility result (FCT).

\subsection{A more general BE-MAC} \label{ex2}
Consider a general noisy DM-MAC with two senders ($\mathcal{X}_1\times\mathcal{X}_2\to\mathcal{Y}$, $X_1, X_2\sim \text{Bern}(\frac12)$) as described in Definition \ref{thm2'}. Now, suppose that this channel is reducible to the $\text{SU-SBC}_{p,W}$. Let the probability of erasure $1-p = 0.6$. Then, the OT capacity region is $\{R_1\leq 0.4, R_2\leq 0.4, R_1+R_2\leq 0.8\}$ if the players are honest but curious. The OT capacity for the case of bounded malicious Bob is unknown, but an inner bound on the OT capacity is $\{R_1\leq 0.2, R_2\leq 0.2, R_1+R_2\leq 0.4\}$. Since the cheating capability of Bob is bounded within a certain range, imperfect OT is still possible, but the region is strictly smaller than the case of honest Bob.

\parbreak 
As is clear, the sum rate of the OT capacity of $\text{SU-SBC}_{p,W}$ and $\text{SU-SBC}_{p,W,W'}$ is redundant, independent of the probability of erasure. Now consider the following example.

\subsection{Noisy BAC}\label{ex3}
Let $X_1, X_2 \sim \mathrm{Bern}\!\left(\frac{1}{2}\right)$ be independent.
The channel output $Y$ is given by
\[
Y =
\begin{cases}
X_1 + X_2, & \text{with probability } p \\
e, & \text{with probability } 1-p
\end{cases},
\]
where $e$ denotes an erasure symbol. The OT capacity region for the honest-but-curious case is as follows:
\[
\mathcal{C}_{\text{OT}} = \left\{
(R_1,R_2) \ \middle| \
\begin{aligned}
& R_1 \le p \\
& R_2 \le p \\
& R_1 + R_2 \le 3\,p/2
\end{aligned}
\right\},
\]
and for the case of malicious Bob, the lower bound on OT capacity is:
\[
\mathcal{R}_{\text{OT}} = \left\{
(R_1,R_2) \ \middle| \
\begin{aligned}
& R_1 < 3\,p/4 \\
& R_2 < 3\,p/4 
\end{aligned}
\right\}.
\]

The key distinction between this example and example~\ref{ex1} is that independent OT is not feasible over the noiseless BAC, as the channel lacks sufficient intrinsic uncertainty to support independent OT. In contrast, example~\ref{ex3} presents a different scenario in which the channel exhibits two distinct sources of uncertainty:  
(i) the first arises from the noisy correlation $\text{SU-SBC}_{p,W}$, and  
(ii) the second stems from the inherent ambiguity induced by the channel structure itself.

\section{Conclusion}\label{Sec: conc}

We investigated bounds on the oblivious transfer (OT) capacity of noisy multiple-access channels (MACs), focusing on secure multiparty computation with two non-colluding senders and a single receiver. We proposed a protocol that is correct and secure against honest-but-curious parties and showed that it remains correct and secure even when the receiver is malicious, although the exact OT capacity in this adversarial setting remains open. For honest-but-curious players, we established the OT capacity region for a reduced MAC model, achieving optimal rates. For a malicious receiver, we characterized a feasible achievable rate region. 

This work provides a comprehensive treatment of noisy MACs, introducing reductions to symmetric basic correlations (SBCs) as defined in~\cite{Winter1} and extending OT protocols to multi-sender configurations. Key contributions include a precise characterization of OT capacity in the honest-but-curious regime, the use of information-theoretic bounds to analyze performance, and significant progress toward understanding protocols in the presence of adversarial behavior.

\parbreak

An intriguing direction for future research is the study of OT over MACs with colluding parties. If one or both senders may collude with the receiver—a scenario more reflective of practical systems—the secrecy requirements become substantially more stringent. Moreover, in~\cite{Winter1}, the authors present a protocol that remains correct and secure against a malicious sender under slightly unfair correlations, albeit without achieving positive rates. In the MAC setting, if both senders cheat, it is plausible that no positive OT rate can be achieved. However, the case in which only one sender is malicious remains largely unexplored. Investigating these challenging scenarios forms the foundation of our ongoing research efforts.

\section*{Acknowledgment}

The authors gratefully acknowledge valuable discussions with 
Holger Boche, Rémi A. Chou, Pin-Hsun Lin, and Andreas Winter during the course of this research.

CD and HA acknowledge financial support from the Federal Ministry of Research, Technology and Space (BMFTR) within the program \emph{Souver\"an. Digital. Vernetzt.} as part of the joint project \textbf{6G-life$^2$} (project ID~16KIS2415). They also gratefully acknowledge funding from the BMBF quantum programs \textbf{QUIET} (grant~16KISQ0170) and \textbf{QD-CamNetz} (grant~16KISQ169). Their research is further supported by the German Research Foundation (DFG) through the project \emph{Post-Shannon Theory and Implementation} (grant~DE~1915/2-1).

{\appendices
\section{Proof of Lemma \ref{lemm: min-entropy vs smooth} \label{App: proof-lemm: min-entropy vs smooth}}

\subsection{Upper bound}
    By definition, $H_\infty^\epsilon(X|Y)$ considers the min-entropy over all distributions $P_{X'Y'}$ that are $\epsilon$-close to the original distribution $P_{XY}$. This means that:
    \[
    H_\infty^\epsilon(X|Y) \geq H_\infty(P_{X'|Y'}) \quad \text{for any such } P_{X'Y'}.
    \]
    
    Since the original distribution $P_{XY}$ satisfies $\|P_{XY} - P_{X'Y'}\|_1 = 0 \leq \epsilon$, it is included in the set of distributions over which the supremum is taken. Therefore:
    \[
    H_\infty^\epsilon(X|Y) \geq H_\infty(P_{X|Y}),
    \]
    where $P_{X|Y}$ is the conditional distribution derived from $P_{XY}$. From the definition of $H_\infty(X|Y)$:
    \[
    H_\infty(P_{X|Y}) = H_\infty(X|Y).
    \]
    
    Thus, combining the inequalities:
    \[
    H_\infty^\epsilon(X|Y) \geq H_\infty(X|Y).
    \]
\subsection{Lower bound}
    The relaxation ensures that the min-entropy $H_\infty^\epsilon(X|Y)$ can only increase compared to $H_\infty(X|Y)$, as it maximizes over a larger set of distributions. For any $P_{X'Y'}$, it holds that:
    \[
    H_\infty^\epsilon(X|Y) = -\log \max_{P_{X'Y'}} \max_y \sum_x P_{X'|Y'=y}(x).
    \]
    This is an optimized version of the standard $H_\infty(X|Y)$. However, we consider the $\epsilon$-closeness in terms of probabilities to connect these values. Using standard smoothing arguments, the worst-case difference in probabilities for any outcome is bounded by $\epsilon$. Specifically:
    \[
     P_{X|Y=y}(x).\frac{1}{1-\epsilon} \leq P_{X'|Y'=y}(x),  \]
    taking logs:
    \[
    -\log P_{X|Y=y}(x) +\log (1-\epsilon)\geq -\log \left(  P_{X'|Y'=y}(x)  \right),
    \]
    maximizing over all $P_{X'Y'}$:
    \[
    H_\infty(X|Y) \geq H_\infty^\epsilon(X|Y) - \log(1 / \epsilon).
    \]
    This accounts for the worst-case adjustment introduced by smoothing. \qed
\section{Proof of Lemma \ref{lemm: Holenstein} \label{App: proof-lemm: Holstein}}

Let $P_{XY}$ be a probability distribution over $\mathcal{X} \times \mathcal{Y}$, $P_{X^n Y^n} := (P_{XY})^n$ the $n$-wise direct product. Then, for any $\delta \in [0, \log(|\mathcal{X}|)]$ and $(x^n, y^n)$ chosen according to $P_{X^n Y^n}$, let $(X_i, Y_i) \sim P_{XY}$ i.i.d. for $i = 0, \ldots, n-1$. Define:
\[
Z_i := \log \left( \frac{1}{P_{X|Y}(X_i | Y_i)} \right)
\quad \Rightarrow \quad
\sum_{i=0}^{n-1} Z_i = \log \left( \frac{1}{P_{X^n|Y^n}(X^n | Y^n)} \right).
\]
We know that  $\mathbb{E}[Z_i] = H(X|Y)$, so the total expectation is:
\[
\mathbb{E}\left[ \sum_{i=0}^{n-1} Z_i \right] = n H(X|Y).
\]
Each $Z_i \in [0, \log(|\mathcal{X}|)]$ because: (i) For fixed $y$, $P_{X|Y}(x|y) \geq 0$, and (ii) $\frac{1}{P_{X|Y}(x|y)} \leq |\mathcal{X}|$. So:
    \[
    0 \leq \log \left( \frac{1}{P_{X|Y}(x|y)} \right) \leq \log(|\mathcal{X}|).
    \]

\medskip

Let $Z := \sum_{i=0}^{n-1} Z_i$. Then $Z$ is a sum of $n$ i.i.d. bounded random variables in $[0, \log(|\mathcal{X}|)]$. Hoeffding’s inequality \cite{Hoeffding} states that for independent $Z_i \in [a, b]$:
\[
\Pr\left[ \frac{1}{n} \sum_{i=0}^{n-1} Z_i - \mathbb{E}[Z_i] \geq \delta \right]
\leq \exp \left( - \frac{2n^2 \delta^2}{n (b - a)^2} \right)
= \exp \left( - \frac{2n \delta^2}{(b - a)^2} \right).
\]

Here, $a = 0$, $b = \log(|\mathcal{X}|)$, so:

\begin{align}\label{looser baound}
    \Pr\left[ Z \geq n(H(X|Y) + \delta) \right]&\leq \exp\left( - \frac{2n \delta^2}{\log^2(|\mathcal{X}|)} \right)\notag\\
    & = 2^{ - \frac{2n \delta^2}{\log^2(|\mathcal{X}|)} \cdot \log_2 e }\notag\\
    & \leq 2^{ - \frac{n \delta^2}{16 \log^2(|\mathcal{X}|)} }.
\end{align}
In the last step, we used only a slightly \emph{looser} bound.

Now, set $\delta := 4 \sqrt{ \frac{\log(1/\epsilon)}{n} } \log(|\mathcal{X}|)$.
Then, from \eqref{looser baound}:
\[
\Pr_{(U, V) \sim P_{X^n Y^n}} \left[ \log \left( \frac{1}{P_{X^n|Y^n}(U|V)} \right) \geq n(H(X|Y) + \delta) \right] \leq \epsilon.
\]

Thus, with probability at least $1 - \epsilon$, we have:
\[
P_{X^n|Y^n}(x^n | y^n) \geq 2^{-nH(X|Y) - \delta}.
\]

Therefore, the $\epsilon$-smooth conditional min-entropy satisfies:
\[
H^\epsilon_\infty(X^n | Y^n) \geq nH(X|Y) - \delta = nH(X|Y) - 4 \sqrt{n \log(1/\epsilon)} \log(|\mathcal{X}|).
\]

This completes the proof. \qed
\section{Proof of Lemma \ref{lemm: test unit}\label{App: proof-lemm: test unit}}

This lemma is an extended version of \cite[Lemma 5]{Winter1}. It evaluates the likelihood of passing a "typicality test" when input sequences are fed into a noisy MAC and output sequences are generated. The lemma bounds the probability that the channel's output remains consistent with a typical input-output relationship, even under deviations in the input strings (e.g., due to cheating). In simpler terms, this lemma demonstrates that if someone attempts to manipulate the inputs or outputs, the probability of such manipulations going undetected is exponentially small under specific conditions.

\parbreak Divide $X_1^n$ and $X_2^n$ into smaller "blocks" where manipulations occurred. Each block contains at least one incorrect symbol. Use the law of large numbers: Since $\delta n$ positions were changed, the overall likelihood of $(X_1^n, X_2^n, Z)$ remaining jointly typical decreases exponentially:

\parbreak Define the sets $\mathcal{I}_{x_1}$ and $\mathcal{I}_{x_2}$: 

\[
\pi_{x_1}:= \pi(x_1|x_1^n) = \lvert \mathcal{I}_{x_1} \rvert \qquad\text{and}\qquad \pi_{x_2}:= \pi(x_2|x_2^n) = \lvert \mathcal{I}_{x_2} \rvert 
.\]

These sets identify positions where the input symbols $x_1$ and $x_2$ appear in $X_1^n$ and $X_2^n$, respectively. Assume that the sequences $x_1^{\mathcal{I}_{x_1}}$ and $x_2^{\mathcal{I}_{x_2}}$ are manipulated. By Hamming distance properties, the cardinalates satisfy: 

\[
\pi_{x_1} \geq\frac{1}{\lvert \mathcal{X}_1 \rvert}\delta n, \qquad \pi_{x_2} \geq\frac{1}{\lvert \mathcal{X}_2 \rvert}\delta n.
\]

For the MAC output, we analyze the empirical distributions $W_{(\tilde{x}_1,\tilde{x}_2)|_{k}}$ over positions $\mathcal{I}_{x_1}\cap\mathcal{I}_{x_2}$. For the joint input-output behavior at positions $k$: 

\[
\norm{\frac{1}{\pi_{x_1}\pi_{x_2}}\sum_{k\in \mathcal{I}_{x_1}\cap\mathcal{I}_{x_2}} W_{(\tilde{x}_1,\tilde{x}_2)|_{k}} - W_{x_1,x_2}}\geq \frac{1}{\lvert\mathcal{X}_1\rvert\lvert\mathcal{X}_2\rvert}\delta^2\eta.
\]

Here, $W_{(\tilde{x}_1,\tilde{x}_2)|_{k}}$ is the empirical output distribution, and $W_{x_1,x_2}$ is the expected output distribution given inputs $(x_1,x_2)$. The deviation occurs because the joint input behavior $(x_1,x_2)$ does not match the expected channel behavior. By the \emph{pigeonhole principle}, there exists at least one output symbol $z\in\mathcal{Z}$ such that:

\[
\left\lvert\frac{1}{\pi_{x_1}\pi_{x_2}}\sum_{k\in \mathcal{I}_{x_1}\cap\mathcal{I}_{x_2}} W_{(\tilde{x}_1,\tilde{x}_2)|_{k}}(z) - W_{x_1,x_2}(z)\right\rvert\geq \frac{1}{\lvert\mathcal{X}_1\rvert\lvert\mathcal{X}_2\rvert\lvert\mathcal{Z}\rvert}\delta^2\eta.
\]
$W_{(\tilde{x}_1,\tilde{x}_2)|_{k}}$  is the probability of output $z$ at position $k$. This inequality ensures at least one output symbol $z$ where the deviation is significant.

\parbreak Now, consider the number of positions where $z$ occurs in the output sequence $Z^n$. Define:
\[
\pi(z|z^{\mathcal{I}_{x_1},\mathcal{I}_{x_2}}):= \text{count of $z$ in $Z^n$ over positions $\mathcal{I}_{x_1}\cap\mathcal{I}_{x_2}$.}
\]
Then:

\[
\left\lvert \pi(z|z^{\mathcal{I}_{x_1},\mathcal{I}_{x_2}}) - \sum_{k\in \mathcal{I}_{x_1}\cap\mathcal{I}_{x_2}} W_{(\tilde{x}_1,\tilde{x}_2)|_{k}} \right\rvert \geq  \frac{1}{2\lvert\mathcal{X}_1\rvert\lvert\mathcal{X}_2\rvert\lvert\mathcal{Z}\rvert}\delta^2\eta \pi_{x_1}\pi_{x_2}.
\]
This bounds the deviation in counts of $z$ compared to the expected behavior.

\parbreak Refine the bound further by introducing sets:
\[
\mathcal{J}_{x_1x_2y} = \left\{ 
k\in \mathcal{I}_{x_1}\cap\mathcal{I}_{x_2}: y_k = y\right\}, \qquad \pi_{x_1x_2y}:=\lvert\mathcal{J}_{x_1x_2y}\rvert.
\]
Then:
\[
\pi_{x_1x_2y}\geq \frac{1}{4\lvert\mathcal{X}_1\rvert^2\lvert\mathcal{X}_2\rvert^2\lvert\mathcal{Z}\rvert}\delta^2\eta^2 \pi_{x_1}\pi_{x_2}\geq \frac{1}{4\lvert\mathcal{X}_1\rvert^3\lvert\mathcal{X}_2\rvert^3\lvert\mathcal{Z}\rvert}\delta^4\eta^2 n.
\]

Using the Chernoff bound \cite[Lemma 4]{Winter1} completes the proof. \qed
\section{Proof of Theorem \ref{thm: impossibility-pp}}\label{App: proof-thm: impossibility-pp}

Consider the general case in Definition \ref{thm1'}, then we have (Omitting the random variable $\mathbf{C}$, since it is also considered in the final views and is assumed to be ignored by honest players): 

\parbreak\noindent\emph{Bob's privacy:} Alice should not learn the receiver's effective choice $Z'$:
\begin{equation}\label{eq.7}
     I(X;Z'|Z) = 0,
\end{equation}
\begin{equation}\label{eq.8}
    I(X;V|Z,Z',M_{Z^{'}}) = 0.
\end{equation}
This means that:
\begin{equation}\label{eq.9}
    H(M_{\overline{Z^{'}}}|V,Z,Z',M_{Z^{'}}) = H(M_{\overline{Z^{'}}}).
\end{equation}
This implies that given the information $V$, the choice $Z$, the effective choice $Z'$ and $M_{Z^{'}}$, Bob has no additional information about $M_{\overline{Z^{'}}}$. Consider the mutual information between Alice's bits and the information received by Bob $I(X;V) \geq I(M_0,M_1;V)$. Since Bob should learn $M_{Z^{'}}$ and nothing about $M_{\overline{Z^{'}}}$, we have:
\begin{align}\label{eq.10}
    I(M_{Z^{'}},M_{\overline{Z^{'}}} ; V|Z,Z',M_{Z^{'}}) =& I(M_{Z^{'}};V|Z,Z',M_{Z^{'}}) + I(M_{\overline{Z^{'}}};V|Z,Z',M_{Z^{'}}).
\end{align}
The first term should not be equal to zero because $V$ must also convey the value of  $M_{Z^{'}}$ to Bob. The second term should be zero due to the perfect secrecy of OT. So consider the second term:
\begin{align}\label{eq.11}
    I(M_{\overline{Z^{'}}};V|Z,Z',M_{Z^{'}}) & = I(M_{\overline{Z^{'}}};V|Z,Z')\notag\\
    & = H(M_{\overline{Z^{'}}}| Z,Z') - H(M_{\overline{Z^{'}}}|V, Z,Z')\notag\\
    & = H(M_{\overline{Z^{'}}}| Z,Z') - H(M_{\overline{Z^{'}}})\notag\\
    & = - I(M_{\overline{Z^{'}}}; Z,Z'),
\end{align}
where the first equality is happened due the fact that  $M_{\overline{Z^{'}}}$ and $M_{Z^{'}}$ are independent, and the third equality follows from \eqref{eq.9}. 

\parbreak As $I(M_{\overline{Z^{'}}}; Z,Z')$ is a non-negative quantity, then we proved that $I(M_{\overline{Z^{'}}};Y|M_{Z^{'}},Z,Z')$ is not a positive quantity. It is impossible unless $I(M_{\overline{Z^{'}}};Z,Z') = 0$, which implies that all inputs $(X, Z)$ are independent (We already clarified that this is often unattainable and in the most general case, we assume there is a known dependency between the inputs). Combining \eqref{eq.10} and \eqref{eq.11}, we have:
\begin{align}\label{eq.12}
    I(X ; V|Z,Z',M_{Z^{'}})&\stackrel{(a)}{\geq} I(M_{Z^{'}},M_{\overline{Z^{'}}} ; V|Z,Z',M_{Z^{'}})\notag\\
    &\stackrel{(b)}{=} I(M_{Z^{'}};V|Z,Z',M_{Z^{'}})\notag\\
    & \,= H(M_{Z^{'}}|Z,Z',M_{Z^{'}})\notag -H(M_{Z^{'}}|V,Z,Z',M_{Z^{'}})\notag\\
    & \,= 0,
\end{align}
where $(a)$ is due to the Markov chain $(M_0,M_1)\to X \to V$ and data processing inequality (DPI), and $(b)$ is since given $V$ and $Z'$, the uncertainty about $M_{Z^{'}}$ is zero (correctness criterion).
Also assume the usual OT correctness: given $V$ and the (effective) index $Z'$, Bob recovers $M_{Z^{'}}$ with certainty, so
\begin{align}\label{corr}
    H(M_{Z^{'}} \mid V, Z, Z') = 0.
\end{align}
Finally assume the messages are independent a priori and non-degenerate:
\[
M_0 \perp M_1, \qquad H(M_0) > 0, \quad H(M_1) > 0.
\]
From the causal structure of the protocol we have the Markov chain $(M_0, M_1) \longrightarrow X \longrightarrow V$. Now consider \eqref{eq.7}:
\[
I(X; Z' \mid Z) = 0.
\]
Again use the Markov chain $(M_0, M_1) \to X \to Z'$ (the variable $Z'$ depends only on the transcript / view which in turn depends on $X$). By conditional DPI we get
\[
I(M_0, M_1; Z' \mid Z) \le I(X; Z' \mid Z) = 0,
\]
hence
\begin{align}\label{mid1}
    I(M_0, M_1; Z' \mid Z) = 0.
\end{align}
Chain-rule on \eqref{mid1} gives
\[
0 = I(M_{Z^{'}}; Z' \mid Z) + I(M_{\overline{Z^{'}}}; Z' \mid Z, M_{Z^{'}}).
\]
Both terms are nonnegative, so each must be zero:
\[
I(M_{Z^{'}}; Z' \mid Z)  = 0, \qquad I(M_{\overline{Z^{'}}}; Z' \mid Z, M_{Z^{'}}) = 0. 
\]
From the first equality we deduce
\[
H(M_{Z^{'}} \mid Z, Z') = H(M_{Z^{'}} \mid Z).
\tag{6}
\]
(Conditional independence of $M_Z$ and $Z'$ given $Z$ means conditioning on $Z'$ does not further reduce  
entropy beyond conditioning on $Z$ alone.)
By correctness \eqref{corr} we have  
\[
H(M_{Z^{'}} \mid V, Z, Z') = 0.
\]
Combine this with \eqref{eq.10}
\begin{align}\label{mid2}
    I(M_0, M_1 ; V \mid Z, Z', M_{Z^{'}}) = 0.
\end{align}
The latter says that modulo $(Z, Z', M_{Z^{'}})$, $V$ is independent of $(M_0, M_1)$ — equivalently, all dependence of  
$V$ on the message space is already captured by the conditioned $M_{Z^{'}}$.  
Concretely, we can rewrite \eqref{eq.10} as
\[
H(V \mid Z, Z', M_{Z^{'}})
=
H(V \mid Z, Z', M_{Z^{'}}, M_0, M_1).
\]

But given $M_0, M_1$ and the protocol, $V$ is determined stochastically through the protocol; the important  
observation is that \eqref{mid2} implies that, after conditioning on $(Z, Z', M_{Z^{'}})$, $V$ carries no additional information  
about the messages. Together with \eqref{corr} this implies that the uncertainty of $M_{Z^{'}}$ given $Z, Z'$ must be no  
larger than its uncertainty after seeing $V$:
\begin{align}\label{mid3}
    H(M_{Z^{'}} \mid Z, Z') = H(M_{Z^{'}} \mid V, Z, Z') = 0,
\end{align}
where the last equality is due to \eqref{corr}.

This means that the message $M_{Z^{'}}$ must be determined (with probability 1) by $(Z, Z')$. Under the usual OT input model (two independent non-degenerate messages), this is
impossible: $M_{Z^{'}}$ has positive entropy conditioned on any small amount of side information unless the
messages are degenerate (deterministic functions of the choices). Thus we get a contradiction to the
assumption that $H(M_0) > 0$ or $H(M_1) > 0$. Therefore the perfect-information equalities \eqref{eq.7} and \eqref{eq.8}
cannot hold simultaneously for nontrivial messages.\qed

\section{Proof of Theorem \ref{thm: impossibility-MAC}}\label{App: proof-thm: impossibility-MAC}
 
The proof follows from the extended version presented in \cite{Crepeau}. 
Consider the intuitive security criterion for Bob: if Alice-$i$ acts maliciously, she sends 
$X'_i$ instead of $X_i$. This means that Alice-$i$ may replace any subset of bits by 
deterministic or random bits of her choice, so that the Hamming distance between the original 
and modified sequences satisfies $0 \leq d_H \leq n$. 

In such a case, the security constraint requires that the mutual information between Bob and 
Alice-$i$'s effective input, conditioned on $(X_i,\mathbf{C})$, is zero 
$(X'_i \rightarrow (\mathbf{C},X_i) \rightarrow Y)$. 
Likewise, the mutual information between Bob's final view $(V,Y)$ and Alice-$i$'s final view, 
conditioned on $(\mathbf{C}, X_i, X'_i)$, must also be zero 
$(U_i \rightarrow (\mathbf{C},X_i,X'_i) \rightarrow (V,Y))$:
\[
    I(Y;X'_i \mid \mathbf{C},X_i) = 0, 
    \qquad\text{and}\qquad 
    I(V,Y;U_i \mid \mathbf{C},X_i,X'_i) = 0,
\]
where $Y = Z_i$ and $V = M'_{iZ_i}$. Combining these conditions, we obtain
\begin{align*}
   I(Z_i;X'_i \mid &\mathbf{C},X_i) 
   + I(M'_{iZ_i},Z_i;U_i \mid \mathbf{C},X_i,X'_i)\\
   &\,= I(Z_i;X'_i \mid \mathbf{C},X_i)
      + I(Z_i;U_i \mid \mathbf{C},X_i,X'_i)
      + I(M'_{iZ_i};U_i \mid \mathbf{C},Z_i,X_i,X'_i)\\
   &\stackrel{(a)}{=} 
      I(Z_i;X'_i \mid \mathbf{C},X_i)
      + I(Z_i;U_i \mid \mathbf{C},X_i,X'_i)\\
   &\,= I(Z_i;X'_i,U_i \mid \mathbf{C},X_i)\\
   &\,= I(Z_i;U_i \mid \mathbf{C},X_i)
      + I(Z_i;X'_i \mid \mathbf{C},X_i,U_i)\\
   &\stackrel{(b)}{=} 
      I(Z_i;U_i \mid \mathbf{C},X_i)\\
   &\,= 0,
\end{align*}
where $(a)$ follows from $I(M'_{iZ_i};U_i \mid \mathbf{C},Z_i,X_i,X'_i) = 0$, and 
$(b)$ is justified as follows. Let $M'_{iZ_i} = (M'_{i,0}, M'_{i,1})$ where, for each 
$l \in \{0,1\}$, the value $M'_{i,l}$ is chosen according to the distribution 
$P_{V \mid \mathbf{C},X_i,U_i,Z_i=l}$, except that $M'_{iZ_i} = V$, which corresponds to 
Bob's final view. Since 
\[
    P_{V \mid \mathbf{C},X_1,X_2,U_1,U_2,Z_1,Z_2}
    = P_{V \mid \mathbf{C},X_1,U_1,Z_1}\,
      P_{V \mid \mathbf{C},X_2,U_2,Z_2}
\]
due to independence of the transmitters, both $M'_{i,l}$ for $l \in \{0,1\}$ follow the 
distribution $P_{V \mid \mathbf{C},X_i,U_i,Z_i=l}$. This construction ensures the Markov chain 
\[
    X'_i \rightarrow (\mathbf{C},X_i,U_i) \rightarrow Z_i,
\]
which implies 
\[
    I(X'_i;Z_i \mid \mathbf{C},X_i,U_i) = 0,
\]
for $i \in \{1,2\}$.

\parbreak Without considering the random variable $\mathbf{C}$, in the case of malicious Bob, we have: 

\parbreak\noindent\emph{Bob's privacy:} Alice-$i$ should not learn the receiver's effective choice $Z'_i$:
\begin{equation}\label{eq.13}
     I(X_i;Z'_i|Z_i) = 0,
\end{equation}
\begin{equation}\label{eq.14}
    I(X_i;V|Z_i,Z'_i,M_{iZ^{'}_i}) = 0.
\end{equation}

This means that:
\begin{equation}\label{eq.15}
    H(M_{i\overline{Z^{'}_i}}|V,Z_i,Z'_i,M_{iZ^{'}_i}) = H(M_{i\overline{Z^{'}_i}}).
\end{equation}
This implies that given the information $V$, the choice $Z_i$, the effective choice $Z'_i$ and $M_{iZ^{'}_i}$, Bob has no additional information about $M_{i\overline{Z^{'}_i}}$. Consider the mutual information between Alice's bits and the information received by Bob $I(X_i;V) = I(M_{i0},M_{i1};V)$. Since Bob should learn $M_{iZ^{'}_i}$ and nothing about $M_{i\overline{Z^{'}_i}}$, we have:
\begin{align}\label{eq.16}
    I(M_{iZ^{'}_i}&, M_{i\overline{Z^{'}_i}} ; V|Z_i,Z'_i,M_{iZ^{'}_i}) = I(M_{iZ^{'}_i};V|Z_i,Z'_i,M_{iZ^{'}_i}) + I(M_{i\overline{Z^{'}_i}};V|Z_i,Z'_i,M_{iZ^{'}_i}).
\end{align}
The first term should not equal zero because $V$ must also convey the value of $M_{iZ^{'}_i}$ to Bob. The second term should be zero due to the perfect secrecy of OT. So consider the second term:
\begin{align}\label{eq.17}
    I(M_{i\overline{Z^{'}_i}};V|Z_i,Z'_i,M_{iZ^{'}_i}) & = I(M_{i\overline{Z^{'}_i}};V|Z_i,Z'_i)\notag\\
    & = H(M_{i\overline{Z^{'}_i}}| Z_i,Z'_i) - H(M_{i\overline{Z^{'}_i}}|V, Z_i,Z'_i)\notag\\
    & = H(M_{i\overline{Z^{'}_i}}| Z_i,Z'_i) - H(M_{i\overline{Z^{'}_i}})\notag\\
    & = - I(M_{i\overline{Z^{'}_i}}; Z_i,Z'_i),
\end{align}
where the first equality is due the fact that  $M_{i\overline{Z^{'}_i}}$ and $M_{iZ^{'}_i}$ are independent, and the third equality follows from \eqref{eq.15}. 

\parbreak As $I(M_{i\overline{Z^{'}_i}}; Z_i,Z'_i)$ is a non-negative quantity, then we proved that $I(M_{i\overline{Z^{'}_i}};Y|M_{iZ^{'}_i},Z_i,Z'_i)$ is not a positive quantity, unless $I(M_{i\overline{Z^{'}_i}};Z_i,Z'_i) = 0$, which implies that all inputs $(X_i, Z_i)$ are independent. Combining \eqref{eq.16} and \eqref{eq.17}, we have:
\begin{align}\label{eq.18}
    I(X_i; V|Z_i,Z'_i,M_{iZ^{'}_i})&\geq I(M_{iZ^{'}_i},M_{i\overline{Z^{'}_i}} ; V|Z_i,Z'_i,M_{iZ^{'}_i})\notag\\
    & = I(M_{iZ^{'}_i};V|Z_i,Z'_i)\notag\\
    & = 0,
\end{align}
where the first equality is because given $V$ and $Z'_i$, the uncertainty about $M_{iZ^{'}_i}$ is zero. \qed

\section{Proof of Lemma \ref{lemm: conditions}}\label{App: proof-lemm: conditions}

\begin{itemize}
        \item Due to the Chernoff bound, we know that the probability of aborting the protocol by Bob in step 2 tends to zero as $n\rightarrow \infty$. When $|\overline{E}_i|<r_i n$ and $|E_i|<r_i n$, then Bob knows $\mathbf{X}_i|_{S_{iZ_i}}$. Since Bob also knows $\kappa_{iZ_i}$, Bob can compute the key $\kappa_{iZ_i}(\mathbf{X}_i|_{S_{iZ_i}})$. Then Bob can recover $M_{iZ_i}$ from $M_{iZ_i}\oplus\kappa_{iZ_i}(\mathbf{X}|_{S_{iZ_i}})$ sent by Alice-$i$. Then, 
        \[
        \lim_{n\rightarrow \infty}\text{Pr}\,\left[(\hat{M}_{1Z_1}, \hat{M}_{2Z_2})\neq (M_{1Z_1}, M_{2Z_2})\right] = 0.\\
        \]
        \item 
        \begin{align}
            I (M_{1\overline{Z}_1},M_{2\overline{Z}_2} ; V) & \,= I (M_{1\overline{Z}_1},M_{2\overline{Z}_2} ; Z_1, Z_2, Y^n, \mathbf{C})\notag\\
            & \,= I (M_{1\overline{Z}_1},M_{2\overline{Z}_2} ; Z_1, Z_2, Y^n, \mathbf{C}_1, \mathbf{C}_2)\notag\\
            & \,= I (M_{1\overline{Z}_1} ; Z_1, Z_2, Y^n, \mathbf{C}_1, \mathbf{C}_2)\notag + I (M_{2\overline{Z}_2} ; Z_1, Z_2, Y^n, \mathbf{C}_1, \mathbf{C}_2|M_{1\overline{Z}_1})\notag\\
            & \overset{(a)}{=} I (M_{1\overline{Z}_1} ; Z_1, Y^n, \mathbf{C}_1)\notag + I (M_{2\overline{Z}_2} ; Z_2, Y^n, \mathbf{C}_2|M_{1\overline{Z}_1})\notag\\
            & \overset{(b)}{=} I (M_{1\overline{Z}_1} ; Z_1, Y^n, \mathbf{C}_1)\notag + I (M_{2\overline{Z}_2} ; Z_2, Y^n, \mathbf{C}_2, M_{1\overline{Z}_1})\notag\\
            & \overset{(c)}{=} I (M_{1\overline{Z}_1} ; Z_1, Y^n, \mathbf{C}_1) + I (M_{2\overline{Z}_2} ; Z_2, Y^n, \mathbf{C}_2),\label{MAC-nColl-middle}
        \end{align}
        where $(a)$ follows from the fact that $M_{iZ_i}- (Z_i, Y^n, \mathbf{C}_i)- (Z_{\overline{i}}, \mathbf{C}_{\overline{i}})$ is a Markov chain, $(b)$ is due to the independency of $M_{iZ_i}$ from $M_{\overline{i}Z_{\overline{i}}}$, and $(c)$ is due the Markov chain $M_{2\overline{Z}_2}-(Z_2,Y^n, \mathbf{C}_2)-M_{1\overline{Z}_1}$. Now, it suffices to show that $I (M_{i\overline{Z}_i} ; Z_i, Y^n, \mathbf{C}_i)_{i\in\{1,2\}}$ tends to zero as $n\rightarrow\infty$. 
        
        \begin{align}
          I (  M_{i\overline{Z}_i};Z_i, Y^n, \mathbf{C}_i)_{i\in\{1,2\}} & \,= I (  M_{i\overline{Z}_i};Z_i, Y^n, S_{i0}, S_{i1}, \kappa_{i0}, \kappa_{i1}, M_{i0}\oplus\kappa_{i0}(\mathbf{X}_i|_{S_{i0}}), M_{i1}\oplus\kappa_{i1}(\mathbf{X}_i|_{S_{i1}}))\notag\\
          & \,= I (  M_{i\overline{Z}_i};Z_i, Y^n, S_{iZ_i}, S_{i\overline{Z}_i},\kappa_{iZ_i}, \kappa_{i\overline{Z}_i}, M_{iZ_i}\oplus\kappa_{iZ_i}(\mathbf{X}_i|_{S_{iZ_i}}), M_{i\overline{Z}_i}\oplus\kappa_{i\overline{Z}_i}(\mathbf{X}_i|_{S_{i\overline{Z}_i}}))\notag\\
          & \,= I (  M_{i\overline{Z}_i};M_{i\overline{Z}_i}\oplus\kappa_{i\overline{Z}_i}(\mathbf{X}_i|_{S_{i\overline{Z}_i}})|Z_i,Y^n, S_{iZ_i}, S_{i\overline{Z}_i}, \kappa_{iZ_i}, \kappa_{i\overline{Z}_i}, M_{iZ_i}\oplus\kappa_{iZ_i}(\mathbf{X}_i|_{S_{iZ_i}}))\notag\\
          & \quad + I (  M_{i\overline{Z}_i};Z_i, Y^n, S_{iZ_i}, S_{i\overline{Z}_i}, \kappa_{iZ_i}, \kappa_{i\overline{Z}_i}, M_{iZ_i}\oplus\kappa_{Z_i}(\mathbf{X}_i|_{S_{iZ_i}}))\notag\\
          & \overset{(a)}{=} I (  M_{i\overline{Z}_i};M_{i\overline{Z}_i}\oplus\kappa_{i\overline{Z}_i}(\mathbf{X}_i|_{S_{i\overline{Z}_i}})|Z_i, Y^n, S_{iZ_i}, S_{i\overline{Z}_i}, \kappa_{iZ_i}, \kappa_{i\overline{Z}_i}, M_{iZ_i}\oplus\kappa_{iZ_i}(\mathbf{X}_i|_{S_{iZ_i}}))\notag\\
          & \,= H (M_{i\overline{Z}_i}\oplus\kappa_{i\overline{Z}_i}(\mathbf{X}_i|_{S_{i\overline{Z}_i}})|Z_i, Y^n, S_{iZ_i}, S_{i\overline{Z}_i}, \kappa_{iZ_i}, \kappa_{i\overline{Z}_i}, M_{iZ_i}\oplus\kappa_{iZ_i}(\mathbf{X}_i|_{S_{iZ_i}}))\notag\\
          & \quad - H (M_{i\overline{Z}_i}\oplus\kappa_{i\overline{Z}_i}(\mathbf{X}_i|_{S_{i\overline{Z}_i}})| M_{i\overline{Z}_i}, Z_i, Y^n, S_{iZ_i}, S_{i\overline{Z}_i}, \kappa_{iZ_i}, \kappa_{i\overline{Z}_i}, M_{iZ_i}\oplus\kappa_{iZ_i}(\mathbf{X}_i|_{S_{iZ_i}}))\notag\\
          & \overset{(b)}{\leq} n (r_i-\lambda')\notag\\
          & \quad - H (\kappa_{i\overline{Z}_i}(\mathbf{X}_i|_{S_{i\overline{Z}_i}})| M_{i\overline{Z}_i}, Z_i, (Y^n|_{S_{i\overline{Z}_i}},Y^n|_{S_{iZ_i}}) , S_{iZ_i}, S_{i\overline{Z}_i}, \kappa_{iZ_i}, \kappa_{i\overline{Z}_i}, M_{iZ_i}\oplus\kappa_{iZ_i}(\mathbf{X}_i|_{S_{iZ_i}}))\notag\\\label{final2: MAC-nColl}
          & \overset{(c)}{=} n (r_i-\lambda')- H (\kappa_{i\overline{Z}_i}(\mathbf{X}_i|_{S_{i\overline{Z}_i}})| Y^n|_{S_{i\overline{Z}_i}},\kappa_{i\overline{Z}_i}),
        \end{align}
where $(a)$ follows from the independency of $M_{i\overline{Z}_i}$ from $(Z_i, Y^n, S_{iZ_i}, S_{i\overline{Z}_i}, \kappa_{iZ_i}, \kappa_{i\overline{Z}_i}, M_{iZ_i}\oplus\kappa_{iZ_i}(\mathbf{X}_i|_{S_{iZ_i}})), i \in \{1,2\}$, $(b)$ follows since $\kappa_{i\overline{Z}_i}(\mathbf{X}_i|_{S_{i\overline{Z}_i}})$ is $n (r_i-\lambda'), i\in \{1,2\}$ bits long and $(c)$ follows since $\kappa_{i\overline{Z}_i}(\mathbf{X}_i|_{S_{i\overline{Z}_i}})- (Y^n|_{S_{i\overline{Z}_i}},\kappa_{i\overline{Z}_i})-M_{i\overline{Z}_i}, Z_i, Y^n , S_{i\overline{Z}_i}, S_{i{Z_i}}, \kappa_{iZ_i}, M_{iZ_i}\oplus\kappa_{iZ_i}(\mathbf{X}_i|_{S_{iZ_i}})),{i\in\{1,2\}}$ is a Markov chain.
We know that,
\[
H_2(\mathbf{X}_i|_{S_{i\overline{Z}_i}}|Y^n|_{S_{i\overline{Z}_i}}=y^n|_{s_{i\overline{z}_i}}) = \Delta(y^n|_{s_{i\overline{z}_i}}) \geq n r_i, i \in \{1,2\},\]

since the size of the set $S_{i\overline{Z}_i}$ is at least $n r_i$. Then, from Lemma \ref{entropy hash} we have:
\begin{align*}
    H(\kappa(\mathbf{X}_i|_{S_{i\overline{Z}_i}})|\kappa, Y^n|_{S_{i\overline{Z}_i}}=y^n|_{s_{i\overline{z}_i}})
    & \geq n(r_i-\lambda') - \frac{2^{n(r_i-\lambda')-nr_i}}{\ln 2}\\
    & = n(r_i-\lambda') - \frac{2^{-n\lambda'}}{\ln 2}.
\end{align*}
Then, \eqref{final2: MAC-nColl} tends to: 
\begin{align*}
    \lim_{n\rightarrow\infty} I (  M_{i\overline{Z}_i};V)_{i\in\{1,2\}}& \leq \lim_{n\rightarrow\infty}\left[ n(r_i-\lambda') - n(r_i-\lambda') + \frac{2^{-n\lambda'}}{\ln 2}\right]\\
    & = \lim_{n\rightarrow\infty} \frac{2^{-n\lambda'}}{\ln 2}\\
    & = 0 .
\end{align*}
Then, from \eqref{MAC-nColl-middle}, we have $\lim_{n\rightarrow\infty} I (M_{1\overline{Z}_1},M_{2\overline{Z}_2} ; V) = 0$. 
    \item \begin{align*}I(Z_i ; U_i) & \,= I(Z_i ; M_{i0}, M_{i1}, X_i^n, R_{A_i}, \mathbf{C}_i)\\
    & \,= I(Z_i ; M_{i0}, M_{i1}, X_i^n, S_{i0}, S_{i1}, \kappa_{i0}, \kappa_{i1}, M_{i0}\oplus\kappa_{i0}(\mathbf{X}_i|_{S_{i0}}), M_{i1}\oplus\kappa_{i1}(\mathbf{X}_i|_{S_{i1}}), R_{A_i}) \\
    & \,= I(Z_i ; M_{i0}, M_{i1}, X_i^n, S_{i0}, S_{i1}, \kappa_{i0}, \kappa_{i1}, \kappa_0(\mathbf{X}_i|_{S_{i0}}),\kappa_1(\mathbf{X}_i|_{S_{i1}}), R_{A_i})\\
    & \,= I(Z_i ; M_{i0}, M_{i1}, X_i^n, S_{i0}, S_{i1}, \kappa_{i0}, \kappa_{i1}, R_{A_i}) \\
    & \overset{(a)}{=} I(Z_i ; X_i^n, S_{i0}, S_{i1}) \\
    &\overset{(b)}{=} I(Z_i ; S_{i0}, S_{i1}) \\
    &\overset{(c)}{=} 0,
\end{align*}
where $R_{A_i}=(R^{(i)}=(R_{i0}, R_{i1}), T^{(i)}=(T_{i0}, T_{i1}))$, $(a)$ follows since $M_{i0}, M_{i1}, \kappa_{i0}, \kappa_{i1} \perp (Z_i, X_i^n, S_{i0}, S_{i1}, R_{A_i})$, $(b)$ follows since $X_i^n \perp (Z_i, S_{i0}, S_{i1})$, and $(c)$ follows since the channel acts independently on each input bit and $|S_{i0}| = |S_{i1}|$.
\end{itemize}
\qed
\section{Proof of Theorem \ref{thm: key agreement}\label{App: proof-thm: key agreement}}

\begin{align*}
        H(S_i) & = H(S_i|X_{\overline{i}})\\
        & = I(S_i; C_i^l, E|X_{\overline{i}}) + H(S_i|C_i^l ,E, X_{\overline{i}}).
    \end{align*}
    Consider the last term of the above expression,
    \begin{align*}
        H(S_i|C_i^l, E, X_{\overline{i}}) & \,= H(S_i,X_i|C_i^l, E, X_{\overline{i}}) - H(X_i|S_i, C_i^l, E, X_{\overline{i}})\\
        & \,= H(X_i|C_i^l, E, X_{\overline{i}}) + H(S_i|C_i^l, E, X_i,X_{\overline{i}}) - H(X_i|S_i,C_i^l, E, X_{\overline{i}})\\
        & \stackrel{(a)}{=} H(X_i|C_i^l, E X_{\overline{i}}) - H(X_i|S_i, C_i^l, E, X_{\overline{i}})\\
        & \stackrel{(b)}{\leq} H(X_i|C_i^l, E, X_{\overline{i}}) - H(X_i|S_i, C_i^l, E, X_{\overline{i}} Y)\\
        & \,= H(X_i|C_i^l, E, X_{\overline{i}}) - H(X_i,S_i|C_i^l, E, Y, X_{\overline{i}}) + H(S_i|C_i^l, E, Y, X_{\overline{i}})\\
        & \stackrel{(c)}{=} H(X_i|C_i^l, E, X_{\overline{i}}) - H(X_i|C_i^l, E, Y, X_{\overline{i}}) + H(S_i|C_i^l, E, Y, X_{\overline{i}})\\
        & \stackrel{(d)}{\leq} I(X_i;Y|X_{\overline{i}}, C_i^l, E) + H(S_i|C_i^l, E, Y)\\
        & \stackrel{(e)}{\leq} I(X_i;Y|X_{\overline{i}}, C_i^l, E) + H(S_i|S'_i),\\
    \end{align*}
    for $i \in \{1,2\}$, where $(a)$ and $(c)$ follow from the condition \eqref{SKA-MAC-3}, $(b)$ is due to the fact that conditioning does not increase the entropy, $(d)$ is due to the fact that $S_i\bot X_{\overline{i}}$, and $(e)$ follows from:
    \begin{align*}
        H(S_i,S'_i|C_i^l, Y, E) & = H(S_i|C_i^l, Y, E) + H(S'_i|S_i, C_i^l, Y, E)\\
        & = H(S'_i|C_i^l, Y, E) + H(S_i|S'_i, C_i^l, Y, E).
    \end{align*}
    Considering the condition \eqref{SKA-MAC-6}, implies that:
    \begin{equation*}
        H(S_i|C_i^l, Y, E) = H(S_i|S'_i, C_i^l, Y, E) \leq H(S_i|S'_i).
    \end{equation*}
    Putting everything together completes the proof for the individual rates.

    \parbreak For the joint entropy, by the same reasoning as above, we have:
    \begin{align}\label{joint entropy: SKA1}
       H(S_1,S_2)= I(S_1,S_2; C_1^l, C_2^l, E) + H(S_1,S_2|C_1^l, C_2^l, E).
    \end{align}
    Consider the last term of the above expression,
    \begin{align}
        H(S_1,S_2|C_1^l, C_2^l, E) & \,= H(S_1, S_2, X_1, X_2|C_1^l, C_2^l, E) - H(X_1,X_2|S_1,S_2, C_1^l, C_2^l, E)\notag\\
        & \,= H(X_1,X_2|C_1^l, C_2^l, E) + H(S_1,S_2|C_1^l,C_2^l, E, X_1, X_2) - H(X_1,X_2|S_1, S_2, C_1^l, C_2^l, E)\notag\\
        & \,= H(X_1,X_2|C_1^l, C_2^l, E) - H(X_1,X_2|S_1, S_2, C_1^l, C_2^l, E)\notag\\
        & \,\leq H(X_1,X_2|C_1^l, C_2^l, E) - H(X_1,X_2|S_1,S_2, C_1^l, C_2^l, E, Y)\notag\\
        & \,= H(X_1,X_2|C_1^l, C_2^l, E) - H(X_1,X_2, S_1, S_2|C_1^l, C_2^l, E, Y) + H(S_1,S_2|C_1^l, C_2^l, E, Y)\notag\\
        & \,= H(X_1,X_2|C_1^l, C_2^l, E) - H(X_1,X_2|C_1^l, C_2^l, E, Y) + H(S_1,S_2|C_1^l, C_2^l, E, Y)\notag\\
        & \stackrel{(a)}{\leq} I(X_1,X_2;Y|C_1^l, C_2^l, E) + H(S_1,S_2|C_1^l, C_2^l, E, Y)\notag\\
        & \,\leq I(X_1,X_2;Y|C_1^l, C_2^l, E) + H(S_1,S_2|S'_1, S'_2),\label{joint entropy: SKA2}
    \end{align}
    where $(a)$ follows from \eqref{SKA-MAC-4}. \qed

\section{Proof of Corollary \ref{cor1}\label{App: proof-cor: cor1}}

Alice-1 and Alice-2 each independently generate uniformly distributed keys $S_1$ and $S_2$, respectively. They then produce channel inputs as stochastic functions of these keys, resulting in $X_1^n = f'_1(S_1)$ and $X_2^n = f'_2(S_2)$. These inputs are sent over the DM-MAC. The outputs $Y^n$ and $E^n$ are subsequently received by Bob and Eve, respectively. Following this, Alice-$i$ generates $\mathbf{C}_i = f_i(S_i,E^n)$. These functions $(\mathbf{C}_i, i\in\{1,2\})$ are then transmitted over the public channel so that the receiver can reconstruct the keys. It is important to note that all the functions mentioned above are stochastic. According to Fano's inequality, for any arbitrarily small $\epsilon\geq 0$, we have:
    \begin{align*}
        H(S_1,S_2|Y^n, \textbf{C}_1,\textbf{C}_2) & \leq h(\epsilon) + \epsilon (H(S_1)+H(S_2))\\
        & \leq h(\epsilon) + \epsilon (nR_1 + nR_2 + 2n\epsilon)\\
        & \leq n (\frac{h(\epsilon)}{n} + \epsilon (R_1 + R_2 + 2n\epsilon)) \triangleq n\epsilon'.
    \end{align*}
    It is clear that $\epsilon'\rightarrow 0$ if $\epsilon\rightarrow 0$. Also, two security criteria should be fulfilled for arbitrarily small $\epsilon\geq 0$: $I(S_1; E^n \mathbf{C}_1)\leq n\epsilon$, $I(S_2; E^n \mathbf{C}_2)\leq n\epsilon$ (Condition \eqref{SKA-MAC-8}):
    \begin{align}\label{eq: upper bound: Salimi}
        R_i & \,\leq \frac{1}{n} H(S_i) + \epsilon\notag\\
        & \stackrel{(a)}{\leq} \frac{1}{n} H(S_i|E^n, \mathbf{C}_i) + 2\epsilon\notag\\
        & \, = \frac{1}{n} H(S_i|X_{\overline{i}}^n, E^n, \mathbf{C}_i) + 2\epsilon\notag\\
        & \stackrel{(b)}{\leq} \frac{1}{n} H(S_i|X_{\overline{i}}^n, E^n, \mathbf{C}_i) - \frac{1}{n} H(S_i|Y^n, \mathbf{C}_i, \mathbf{C}_{\overline{i}}) + 2\epsilon + \epsilon'\notag\\
        & \, \leq \frac{1}{n} H(S_i|X_{\overline{i}}^n, E^n, \mathbf{C}_i) - \frac{1}{n} H(S_i|Y^n, X_{\overline{i}}^n, E^n, \mathbf{C}_i, \mathbf{C}_{\overline{i}}) + 2\epsilon + \epsilon'\notag\\
        & \stackrel{(c)}{=} \frac{1}{n} H(S_i|X_{\overline{i}}^n, E^n, \mathbf{C}_i) - \frac{1}{n} H(S_i|Y^n, X_{\overline{i}}^n, E^n, \mathbf{C}_i) + 2\epsilon + \epsilon'\notag\\
        & \, = \frac{1}{n}I(S_i;Y^n|X_{\overline{i}}^n, E^n, \mathbf{C}_i) + 2\epsilon + \epsilon'\notag\\
        & \, = \frac{1}{n} H(Y^n|X_{\overline{i}}^n, E^n, \mathbf{C}_i) - \frac{1}{n} H(Y^n|X_{\overline{i}}^n, \mathbf{C}_i, E^n, S_i) + 2\epsilon + \epsilon'\notag\\
        & \, \leq \frac{1}{n} H(Y^n|X_{\overline{i}}^n, E^n) - \frac{1}{n} H(Y^n|X_i^n, X_{\overline{i}}^n, \mathbf{C}_i, E^n, S_i) + 2\epsilon + \epsilon'\notag\\
        & \stackrel{(d)}{=} \frac{1}{n} H(Y^n|X_{\overline{i}}^n, E^n) - \frac{1}{n} H(Y^n|X_i^n, X_{\overline{i}}^n, E^n, S_i) + 2\epsilon + \epsilon'\notag\\
        & \stackrel{(e)}{=} \frac{1}{n} H(Y^n|X_{\overline{i}}^n, E^n) - \frac{1}{n} H(Y^n|X_i^n, X_{\overline{i}}^n, E^n) + 2\epsilon + \epsilon'\notag\\
        & \stackrel{(f)}{\leq} \max_{P_{X_1}P_{X_2}}I(X_i;Y|X_{\overline{i}}, E) + 2\epsilon + \epsilon',
    \end{align}
    for $i \in \{1,2\}$, where $(a)$ follows from the security criteria, $(b)$ follows from Fano's lemma, $(c)$ and $(d)$ follow from the fact that $C_i = f_i(S_i,E^n)$, $(e)$ follows from the Markov chain $Y^n-(X_1^n,X_2^n)-(S_1,S_2)$, and $(f)$ follows from the memoryless property of the channel and an argument similar to \cite[Theorem 4]{Maurer}. As is mentioned before and proved in \cite{Maurer} for the case of constant random variable $E$, we can remove the impact of $E$ from the above mutual information quantity. The whole above process can similarly be repeated for the joint entropy $H(S_1,S_2)$. This completes the proof. \qed

\section{Proof of Theorem \ref{thm: OT capacity}}\label{App: proof-thm: OT Capacity-upper}

\subsection{The first upper bound on OT capacity}
Consider the system model illustrated in Figure \ref{fig: both-MAC}-$(a)$. To prove the upper bound for $\binom 21-\text{OT}^{k_1,k_2}$ capacity, consider that $(n, k_1, k_2)$ protocols fulfilling conditions \eqref{goals: MAC-nColl-1}-\eqref{goals: MAC-nColl-3}. Condition \eqref{goals: MAC-nColl-2} can be rewritten as:
    \begin{align}
        \lim_{n\rightarrow \infty}I (M_{1\overline{Z}_1},M_{2\overline{Z}_2} ; V) &= \lim_{n\rightarrow \infty}I (M_{1\overline{Z}_1} ; V) +  \lim_{n\rightarrow \infty}I (M_{2\overline{Z}_2} ; V|M_{1\overline{Z}_1})\label{first step: general}\\
        & = \lim_{n\rightarrow \infty}I (M_{1\overline{Z}_1} ; V) +  \lim_{n\rightarrow \infty}I (M_{2\overline{Z}_2} ; V)\\
        & = 0, \notag
    \end{align} 
    This means that $\lim_{n\rightarrow \infty}I (M_{i\overline{Z}_i} ; V)_{i \in \{1,2\}} = 0$ and its relaxed version: $I (M_{i\overline{Z}_i} ; V)_{i \in \{1,2\}} = o(n)$:
    \begin{align}\label{smoothed condition}
        I (M_{i\overline{Z}_i} ; V)_{i \in \{1,2\}} = I (M_{i\overline{Z}_i} ; Z_i, Z_{\overline{i}}, R_B, Y^n, \textbf{C})_{i \in \{1,2\}} = o(n).
    \end{align}
    $I (Z_i, Z_{\overline{i}}, R_B, Y^n, \textbf{C} ; M_{i\overline{Z}_i})_{i \in \{1,2\}} = I (Z_i, Z_{\overline{i}} ; M_{i\overline{Z}_i})_{i \in \{1,2\}} + I (R_B, Y^n, \textbf{C} ; M_{i\overline{Z}_i}| Z_i, Z_{\overline{i}})_{i \in \{1,2\}} \linebreak\stackrel{(a)}{=} I (R_B, Y^n, \textbf{C}_i ; M_{i\overline{Z}_i}| Z_i, Z_{\overline{i}})_{i \in \{1,2\}}$, where $(a)$ follows from the fact that $I (Z_i, Z_{\overline{i}} ; M_{i\overline{Z}_i})_{i \in \{1,2\}} = 0$. So, Condition \eqref{smoothed condition} implies that $I (R_B, Y^n, \textbf{C} ; M_{i\overline{Z}_i}| Z_i, Z_{\overline{i}})_{i \in \{1,2\}}\rightarrow 0$. Instead of using Condition \eqref{smoothed condition}, we have: 
    \begin{equation}\label{goal-MAC-ncoll- 3: smoothed final}
        I (R_B, Y^n, \textbf{C}_i ; M_{i\overline{Z}_i}| Z_i, Z_{\overline{i}})_{i \in \{1,2\}} = o(n).
    \end{equation}
    Given a DM-MAC $\{\mathcal{W}: \mathcal{X}_1\mathcal{X}_2\rightarrow \mathcal{Y}\}$, consider $\binom 21-\text{OT}^{k_1,k_2}$ protocols that satisfy \eqref{goals: MAC-nColl-1}, \eqref{goals: MAC-nColl-3}, and \eqref{goal-MAC-ncoll- 3: smoothed final}. According to Lemma \ref{lemma: Rudolf}, Condition \eqref{goals: MAC-nColl-3} implies:
    
    \begin{align}
       H(M_{i\overline{z}_i}|X_i^n, \mathbf{C}_i, Z_i=\overline{z}_i)_{i\in\{1,2\}} -& H(M_{i\overline{z}_i}|X_i^n, \mathbf{C}_i, Z_i=z_i)_{i\in\{1,2\}}\\
       &\,= H(M_{i\overline{z}_i}|\mathbf{C}_i, Z_i=\overline{z}_i)_{i\in\{1,2\}}\notag - H(M_{i\overline{z}_i}|\mathbf{C}_i, Z_i=z_i)_{i\in\{1,2\}}\\
        & \stackrel{(a)}{=} H(M_{i\overline{z}_i}|\mathbf{C}_i, Z_i=\overline{z}_i)_{i\in\{1,2\}} - H(M_{i\overline{z}_i}|\mathbf{C}_i, Z_i=z_i)_{i\in\{1,2\}} \notag\\
        & \qquad\qquad\qquad\qquad\qquad\qquad\quad \,- H(M_{i\overline{z}_i}|Z_i=\overline{z}_i)_{i\in\{1,2\}}\notag\\
        &\qquad\qquad\qquad\qquad\qquad\qquad\quad \,+ H(M_{i\overline{z}_i}|Z_i=z_i)_{i\in\{1,2\}}\notag\\\label{o-1}
        & = I(M_{i\overline{z}_i};\mathbf{C}_i|Z_i=\overline{z}_i)_{i\in\{1,2\}}-I(M_{i\overline{z}_i};\mathbf{C}_i|Z_i=z_i)_{i\in\{1,2\}}\\\label{o}
        & = o(n),
    \end{align}
       where $(a)$ follows from the fact that $H(M_{i\overline{z}_i}|Z_i=\overline{z}_i) = H(M_{i\overline{z}_i}|Z_i=z_i) = k_i$.

       \parbreak Suppose $Z_i = z_i$ and $Z_{\overline{i}} = z_{\overline{i}}$ in \eqref{goal-MAC-ncoll- 3: smoothed final}, then we have $I (R_B, Y^n, \textbf{C}_i ; M_{i\overline{z}_i}| Z_i = z_i, Z_{\overline{i}} = z_{\overline{i}})_{i \in \{1,2\}} = o(n)$ combined with \eqref{o-1} and \eqref{o} concludes: 
       \begin{equation}\label{o+1}
       I(M_{i\overline{z}_i};\mathbf{C}_i|Z_i = z_i, Z_{\overline{i}} = z_{\overline{i}})_{i \in \{1,2\}} = o(n).
       \end{equation}
       
       Conditions \eqref{goals: MAC-nColl-1} and \eqref{o+1} without conditioning on $(Z_i = z_i, Z_{\overline{i}} = z_{\overline{i}})_{i \in \{1,2\}}$ are akin to those defining a secret key for Alice-$i$ and Bob with weak secrecy (conditions \eqref{SKA-MAC-7} and \eqref{SKA-MAC-8}, respectively), ensuring security from an eavesdropper who observes their public communication $\mathbf{C}_i$. So, $k_i$ would constitute such a secret key by definition, as demonstrated in Corollary \ref{cor1}. Thus, we get:
       \begin{equation}
           k_i = H(M_{i\overline{z}_i})\leq\sum_{l=1}^{n} I(X_{i,l};Y_l|X_{\overline{i},l}) + o(n).
       \end{equation}
       From the memoryless property of the channel we have: $\frac{k_i}{n}=R_i\leq I(X_i;Y|X_{\overline{i}})$. By a similar calculation for the joint entropy where $k_1 + k_2 = H(M_{1\overline{z}_1},M_{2\overline{z}_2})\leq\sum_{l=1}^{n} I(X_{1,l},X_{2,l};Y_l) + o(n)$ and from the memoryless property of the channel we have: $\frac{k_1 + k_2}{n}=R_1+R_2\leq I(X_1,X_2;Y)$.
       \subsection{The second upper bound on OT capacity}
       At the first step, we prove that $I(M_{iz_i}; Y^n, R_{B} | X_i^n, \mathbf{C}_i, Z_i = z_i,Z_{\overline{i}} = z_{\overline{i}}) = 0$ for $i\in\{1,2\}$. This means that there is no mutual information between Bob's received messages and his received sequence given his chosen bits, Alice's encoded strings, and the total public transmission.
       Define for $i \in \{1,2\}$, $l \in [1, n]$, $t \in [1, r_l]$, $C_{i, l, 1:t} \triangleq (C_{0, i,l}(j), C_{i, l}(j))_{j \in [1, t]}$ as the messages exchanged between Alice-$i$ and Bob between the first and the $t$-th communication rounds occurring after the $l$-th channel usage. Let $C_i^l \triangleq (C_{i, j, 1:r_j})_{j \in [1, l]}$ represent all messages exchanged by Server 1 with the client before the $l+1$-th channel use. Let $l \in [1, n]$ and $j \in [1, r_l]$. Then, we have:
       \begin{align}\label{eq 1}
    I(M_{i0}, M_{i1}, R_{A_i}; Y^l, R_{B} | X_i^l, C_i^{l-1},&C_{i,l,1:j}, Z_i,Z_{\overline{i}})\notag\\
    & \overset{(a)}{\leq} I(M_{i0}, M_{i1}, R_{A_i}; Y^l, R_{B}, C_{0,i,l}(j) | X_i^l, C_i^{l-1},C_{i,l,1:(j-1)}, C_{i,l}(j), Z_i, Z_{\overline{i}})\\ 
    & \overset{(b)}{=} I(M_{i0}, M_{i1}, R_{A_i}; Y^l, R_{B} | X_i^l, C_i^{l-1},C_{i,l,1:(j-1)}, C_{i,l}(j), Z_i, Z_{\overline{i}})\notag\\
    & \leq I(M_{i0}, M_{i1}, R_{A_i}; Y^l, R_{B}, C_{i,l}(j) | X_i^l, C_i^{l-1},C_{i,l,1:(j-1)}, Z_i, Z_{\overline{i}})\notag\\
    &\overset{(c)}{=} I(M_{i0}, M_{i1}, R_{A_i}; Y^l, R_{B} | X_i^l, C_i^{l-1},C_{i,l,1:(j-1)}, Z_i, Z_{\overline{i}})\label{eq 2} \\
    &\overset{(d)}{\leq} I(M_{i0}, M_{i1}, R_{A_i}; Y^l, R_{B} | X_i^l, C_i^{l-1}, Z_i, Z_{\overline{i}}) \notag\\
    &\overset{(e)}{=} I(M_{i0}, M_{i1}, R_{A_i}; Y^{l-1}, R_{B} | X_i^l, C_i^{l-1}, Z_i, Z_{\overline{i}}) \notag\\
    &\leq I(M_{i0}, M_{i1}, R_{A_i}, X_{i,l}; Y^{l-1}, R_{B} | X_i^{l-1}, C_i^{l-1}, Z_i, Z_{\overline{i}})\notag\\
    &\overset{(f)}{=} I(M_{i0}, M_{i1}, R_{A_i}; Y^{l-1}, R_{B} | X_i^{l-1}, C_i^{l-1}, Z_i, Z_{\overline{i}})\label{eq 3}\\
    & \overset{(g)}{=} 0,    
\end{align}
where the steps are justified as follows:
\begin{itemize}
    \item $(a)$ follows by the definition of $C_{i, l, 1:j} = (C_{i, l, 1:(j-1)}, C_{0, i,l}(j), C_{1, l}(j))$ and by the chain rule.
    \item $(b)$ follows by the chain rule and because $C_{0, i,l}(j)$ depends on $(Z_i, Z_{\overline{i}}, R_{B}, Y^l, C_{i,l}(j), C_{i, l, 1:(j-1)})$.
    \item $(c)$ follows by the chain rule and because $C_{i,l}(j)$ depends on $(M_{i0}, M_{i1}, R_{A_i}, C_{i, l, 1:(j-1)}, C_i^{l-1})$.
    \item $(d)$ follows by previous iterations $(j-1)$ of Equations \eqref{eq 1} to \eqref{eq 2}.
    \item $(e)$ follows from the Markov chain: $(M_{i0},M_{i1},R_{A_i})-(Y^{l-1}, R_{B}, X_i^l, C_i^{l-1}, Z_i, Z_{\overline{i}})-Y_{l}$.
    \item $(f)$ follows the chain rule and the fact that $X_{i,l}$ is a function of $(M_{i0}, M_{i1}, R_{A_i}, C_i^{l-1})$.
    \item $(g)$ follows from the following calculation:
    \begin{align}
        I(M_{i0}, M_{i1}, R_{A_i}; Y^{l}, R_{B} | X_i^{l}, C_i^{l}, Z_i, Z_{\overline{i}}) & \leq I(M_{i0}, M_{i1}, R_{A_i}; Y^{l-1}, R_{B} | X_i^{l-1}, C_i^{l-1}, Z_i, Z_{\overline{i}})\\
        &\leq I(M_{i0}, M_{i1}, R_{A_i};R_{B} |Z_i, Z_{\overline{i}})\\
        & = 0.
    \end{align}
    Then, for any $z_i\in\{0,1\}, i\in\{1,2\}$ we have:
    \begin{align}
        I(M_{iz_i}; Y^n, R_{B} | X_i^n, \mathbf{C}_i, Z_i = z_i,Z_{\overline{i}} = z_{\overline{i}}) &\leq I(M_{i0}, M_{i1}, R_{A_i}; Y^{n}, R_{B} | X_i^{n}, \mathbf{C}_i, Z_i = z_i,Z_{\overline{i}} = z_{\overline{i}})\notag\\
        & = 0.\label{eq: imp}
    \end{align}
    Now we can prove that for any $z_i\in \{0, 1\}$, we have $H(M_{i\overline{Z}_i}|X_i^n, \textbf{C}, Z_i, Z_{\overline{i}}) = o(n), i\in\{1,2\}$ which means that the uncertainty about the unchosen messages given the encoded inputs, Bob's inputs and the total public communication is negligible as $n\to\infty$:
    \begin{align*}
        H(M_{i\overline{z}_i}|X_i^n, \textbf{C},Z_i = z_i, Z_{\overline{i}}= z_{\overline{i}}) & \,\leq H(M_{i\overline{z}_i}|X_i^n, \textbf{C}_i,Z_i = z_i, Z_{\overline{i}}= z_{\overline{i}})\\
        & \overset{(a)}{\leq} H(M_{i\overline{z}_i}|X_i^n, \textbf{C}_i,Z_i = z_i, Z_{\overline{i}}= z_{\overline{i}}) + o(n)\\
        & \stackrel{(b)}{=} H(M_{i\overline{z}_i}|Y^n, R_{B} , X_i^n, \textbf{C}_i,Z_i = z_i, Z_{\overline{i}}= z_{\overline{i}}) + o(n)\\
        & \stackrel{(c)}{\leq} H(M_{i\overline{z}_i}|Y^n, R_{B} , \hat{M}_{i\overline{z}_i}, \textbf{C}_i, Z_i = z_i, Z_{\overline{i}}= z_{\overline{i}}) + o(n)\\
        & \stackrel{(d)}{\leq} o(n).
    \end{align*}
    Finally, 
    \begin{align}
    H(M_{i\overline{Z}_i}|X_i^n, \textbf{C}_i,Z_i = z_i, Z_{\overline{i}}= z_{\overline{i}}) & = \sum_{z_i, z_{\overline{i}}} \mathbb{P} [(Z_i, Z_{\overline{i}}) = (z_i, z_{\overline{i}})] \times H(M_{i\overline{z}_i}|X_i^n, \textbf{C}_i, Z_i = z_i, Z_{\overline{i}}= z_{\overline{i}})\notag\\
    & = o(n),\label{eq: imp2}
    \end{align}
    where $(a)$ follows from \eqref{o}, $(b)$ follows from \eqref{eq: imp}, $(c)$ holds because $\hat{M}_{i\overline{z}_i}$ is a function of $(Y_i^n, R_{B_i}, \mathbf{C}_i)$, and $(d)$ holds by Fano's inequality and \eqref{goals: MAC-nColl-1}.

    Now, we have: 
    we have:

\begin{align*}
           k_i = H(M_{i\overline{z}_i}|Z_i=z_i, Z_{\overline{i}}=z_{\overline{i}}) & \,= H(M_{i\overline{z}_i}| R_{B}, Y^n, \textbf{C}_i,Z_i=z_i, Z_{\overline{i}}=z_{\overline{i}}) + o(n)\notag\\
           & \,= H(M_{i\overline{z}_i}| R_{B}, Y^n, \textbf{C}_i,Z_i=z_i, Z_{\overline{i}}=z_{\overline{i}})\\
           & \qquad + H(X_i^n|M_{i\overline{z}_i},R_{B}, Y^n, \textbf{C}_i,Z_i=z_i, Z_{\overline{i}}=z_{\overline{i}})\notag\\
           & \qquad -H(X_i^n|M_{i\overline{z}_i},R_{B}, Y^n, \textbf{C}_i,Z_i=z_i, Z_{\overline{i}}=z_{\overline{i}}) + o(n)\notag\\
           & \,= H(M_{i\overline{z}_i},X_i^n|R_{B}, Y^n, \textbf{C}_i, Z_i=z_i, Z_{\overline{i}}=z_{\overline{i}})\\
           & \qquad - H(X_i^n|M_{i\overline{z}_i},R_{B}, Y^n, \textbf{C}_i,Z_i=z_i, Z_{\overline{i}}=z_{\overline{i}}) + o(n)\notag\\
           & \,\leq H(M_{i\overline{z}_i},X_i^n|R_{B}, Y^n, \textbf{C}_i,Z_i=z_i, Z_{\overline{i}}=z_{\overline{i}}) + o(n)\notag\\
           & \,= H(X_i^n|R_{B}, Y^n, \textbf{C}_i, Z_i=z_i, Z_{\overline{i}}=z_{\overline{i}})\\
           & \qquad + H(M_{i\overline{z}_i}|X_i^n, R_{B}, Y^n, \textbf{C}_i, Z_i=z_i, Z_{\overline{i}}=z_{\overline{i}}) + o(n)\notag\\
           & \,\leq H(X_i^n|R_{B}, Y^n, \textbf{C}_i, Z_i=z_i, Z_{\overline{i}}=z_{\overline{i}})\\
           & \qquad + H(M_{i\overline{z}_i}|X_i^n, \textbf{C}_i, Z_i=z_i, Z_{\overline{i}}=z_{\overline{i}}) + o(n)\notag\\ 
           & \stackrel{(a)}{=}  H(X_i^n|R_{B}, Y^n, \textbf{C}_i, Z_i=z_i, Z_{\overline{i}}=z_{\overline{i}}) + o(n)\notag\\
           & \stackrel{(b)}{\leq} H(X_i^n|Y^n,Z_i=z_i, Z_{\overline{i}}=z_{\overline{i}}) + o(n)\notag\\
           & \leq \sum_{l=[1:n]} H(X_{i,l}|Y_{l}, Z_i=z_i, Z_{\overline{i}}=z_{\overline{i}}) + o(n)\notag\\
           & \stackrel{(c)}{\leq} \sum_{l=1}^{n} H(X_{i,l}|Y_{l}) + o(n), 
       \end{align*}
       where $(a)$ follows from \eqref{eq: imp2}, $(b)$ is because conditioning does not increase the entropy, and $(c)$ follows from the arguments presented in \cite{Rudolf1}. Similarly, we can prove that $k_1+k_2 = H(M_{1\overline{z}_1},M_{2\overline{z}_2}|Z_1=z_1, Z_2=z_2)\leq \sum_{l=1}^{n} H(X_{1,l},X_{2,l}|Y_{l}) + o(n)$.
\end{itemize}
Then, the final upper bound presented in Theorem \ref{thm: OT capacity} is proved. \qed

\section{Proof of Theorem \ref{thm: OT capacity2}\label{App: proof-thm: OT capacity}}

        We must bound three conditional min-entropies to get three bounds on the lower bound:
        \begin{equation}\label{eq: first}
            H^\epsilon_{\infty}(\mathbf{X}_1|_{S_{1\overline{Z}_1}}|h_{10}(R_{10}, \mathbf{X}_1|_{S_{10}}), h_{11}(R_{11}, \mathbf{X}_1|_{S_{11}}), Y^n, R^{(1)}, T^{(1)}),
        \end{equation}
        \begin{equation}\label{eq: second}
            H^\epsilon_{\infty}(\mathbf{X}_2|_{S_{2\overline{Z}_2}}|h_{20}(R_{20}, \mathbf{X}_2|_{S_{20}}), h_{21}(R_{21}, \mathbf{X}_2|_{S_{21}}), Y^n, R^{(2)}, T^{(2)}),
        \end{equation}
        \begin{equation}\label{eq: third}
            H^\epsilon_{\infty}(\mathbf{X}_1|_{S_{1\overline{Z}_1}},\mathbf{X}_2|_{S_{2\overline{Z}_2}}|h_{10}(R_{10}, \mathbf{X}_1|_{S_{10}}), h_{11}(R_{11}, \mathbf{X}_1|_{S_{11}}), h_{20}(R_{20}, \mathbf{X}_2|_{S_{20}}), h_{21}(R_{21}, \mathbf{X}_2|_{S_{21}}), Y^n, R_{A_1}, R_{A_2}).
        \end{equation}
        Consider \eqref{eq: first}. Since the $\text{SU-SBC}_{p,W}$ is i.i.d., we have: 
        \begin{align*}
            H^\epsilon_{\infty}(\mathbf{X}_1|_{S_{1\overline{Z}_1}}|h_{10}(R_{10}, \mathbf{X}_1|_{S_{10}}), h_{11}(R_{11}, \mathbf{X}_1|_{S_{11}})&, Y^n, R^{(1)}, T^{(1)})\\
            & = H^\epsilon_{\infty}(\mathbf{X}_1|_{S_{1\overline{Z}_1}}|h_{1j}(R_{1j}, \mathbf{X}_1|_{S_{1\overline{Z}_1}}), \mathbf{Y}|_{S_{1\overline{Z}_1}}, R^{(1)}, T^{(1)}).
        \end{align*}
        
        Applying \eqref{eq: holenstein-1} for $\epsilon, \epsilon'>0$, we have: 
        \begin{align}\label{eq: after applying holenstein-1}
            H_\infty^{\epsilon+\epsilon'}(\mathbf{X}_1&|_{S_{1\overline{Z}_1}}|h_{1j}(R_{1j}, \mathbf{X}_1|_{S_{1\overline{Z}_1}}), \mathbf{Y}|_{S_{1\overline{Z}_1}}, R^{(1)}, T^{(1)})\notag\\
            & \geq H_\infty(\mathbf{X}_1|_{S_{1\overline{Z}_1}}|\mathbf{Y}|_{S_{1\overline{Z}_1}}, R^{(1)}, T^{(1)}) + H_\infty^\epsilon(h_{1j}(R_{1j}, \mathbf{X}_1|_{S_{1\overline{Z}_1}})|\mathbf{X}_1|_{S_{1\overline{Z}_1}},\mathbf{Y}|_{S_{1\overline{Z}_1}}, R^{(1)}, T^{(1)})\notag\\
            & \qquad\qquad\qquad\qquad\qquad\qquad\qquad\quad- H_0(h_{1j}(R_{1j}, \mathbf{X}_1|_{S_{1\overline{Z}_1}})|\mathbf{Y}|_{S_{1\overline{Z}_1}}, R^{(1)}, T^{(1)}) - \log(\frac{1}{\epsilon'}).
        \end{align}
        Note that $H_0(h_{1j}(R_{1j}, \mathbf{X}_1|_{S_{1\overline{Z}_1}})|\mathbf{Y}|_{S_{1\overline{Z}_1}}, R^{(1)}, T^{(1)})$ limits the number of distinct possible outputs, restricting the amount of information Alice-1 can gain about Bob's choice: 
        \[
        H_0(h_{1j}(R_{1j}, \mathbf{X}_1|_{S_{1\overline{Z}_1}})|\mathbf{Y}|_{S_{1\overline{Z}_1}}, R^{(1)}, T^{(1)})\leq s_1 n.
        \]
        Knowing that $H_\infty^\epsilon(h_{1j}(R_{1j}, \mathbf{X}_1|_{S_{1\overline{Z}_1}})|\mathbf{X}_1|_{S_{1\overline{Z}_1}},\mathbf{Y}|_{S_{1\overline{Z}_1}}, R^{(1)}, T^{(1)}) = 0$, \eqref{eq: after applying holenstein-1} is simplified to:
        \begin{align}\label{eq: middle}
            H_\infty^{\epsilon+\epsilon'}(\mathbf{X}_1|_{S_{1\overline{Z}_1}}|h_{1j}(R_{1j}, \mathbf{X}_1|_{S_{1\overline{Z}_1}}), \mathbf{Y}|_{S_{1\overline{Z}_1}}, R^{(1)}, T^{(1)})&\,\geq H_{\infty}(\mathbf{X}_1|_{S_{1\overline{Z}_1}}|\mathbf{Y}|_{S_{1\overline{Z}_1}}) - s_1 n-\log(\frac{1}{\epsilon'})\notag\\
            & \overset{(a)}{\geq} H^\epsilon_{\infty}(\mathbf{X}_1|_{S_{1\overline{Z}_1}}|\mathbf{Y}|_{S_{1\overline{Z}_1}}) - s_1 n-\log(\frac{1}{\epsilon})-\log(\frac{1}{\epsilon'}),
        \end{align}
        where $(a)$ follows from Lemma \ref{lemm: min-entropy vs smooth}. Let $V$ be an i.i.d. random variable so that $V = e$ (erasure) with probability $\frac12-\eta$ and $V = Z$ (the output of channel $W$ on input $(X_1,X_2)$). With negligible error probability and $\text{SU-SBC}_{p,W}$ being i.i.d., for $S_{1\overline{Z}_1}$, we have:
        \begin{align}\label{eq: middle+1}
            H_\infty^{\epsilon+\epsilon'}(\mathbf{X}_1|_{ S_{1\overline{Z}_1}}|\mathbf{Y}|_{S_{1\overline{Z}_1}})&\,\geq H_\infty^{\epsilon}(\mathbf{X}_1|_{\lvert S_{1\overline{Z}_1}\lvert}|\mathbf{V}|_{\lvert S_{1\overline{Z}_1}\lvert})\notag\\
            & \overset{(a)}{\geq} \lvert S_{1\overline{Z}_1}\rvert H(X_1|V)-4\sqrt{\lvert S_{1\overline{Z}_1}\rvert\log(1/\epsilon)}\log \lvert \mathcal{X}_1 \rvert\notag\\
            & \,\geq (p-\eta)n H(X_1|V)-4\sqrt{(p-\eta)n\log(1/\epsilon)}\notag\notag\\
            & \,\geq pn H(X_1|V)-\eta n H(X_1|V)-4\sqrt{(p-\eta)n\log(1/\epsilon)}\notag\notag\\
            & \,\geq pn H(X_1|V)-\eta n -4\sqrt{(p-\eta)n\log(1/\epsilon)}\notag\notag\\
            & \overset{(b)}{\geq} pn (1-2\eta)H(X_1)+2\eta n H(X_1|Z)-\eta n-4\sqrt{(p-\eta)n\log(1/\epsilon)}\notag\notag\\
            & \,\geq pn H(X_1)-2\eta n-4\sqrt{n\log(1/\epsilon)},
        \end{align}
        where $(a)$ follows from Lemma \ref{lemm: Holenstein}, $\lvert\mathcal{X}_1\rvert = 2$, and $(b)$ follows from this fact that honest Bob doesn't split the erasures received from Alice-1 between $S_{10}$ and $S_{11}$, with probability exponentially close to one, the total number of non-erased symbols Bob receives from each sender will not exceed $(p + \eta)n$, so the number of non-erasures in $\mathbf{Y}|_{S_{1\overline{Z}_1}}$ is at most $\lvert(p-\eta)n - (p + \eta)n\rvert = 2n\eta$.
        
        \parbreak Putting $\epsilon = \epsilon + \epsilon'$ in \eqref{eq: middle}, then putting \eqref{eq: middle+1} to \eqref{eq: middle}, we have: 
        \begin{align}
            H_\infty^{\epsilon+2\epsilon'}(\mathbf{X}_1|_{S_{1\overline{Z}_1}}|h_{1j}(R_{1j}, \mathbf{X}_1|_{S_{1\overline{Z}_1}}), \mathbf{Y}|_{S_{1\overline{Z}_1}}, R^{(1)}, T^{(1)})&\,\geq pn H(X_1)-2\eta n-4\sqrt{n\log(1/\epsilon)} - s_1 n-\log(\frac{1}{\epsilon})-\log(\frac{1}{\epsilon'})\notag\\
            & \overset{(a)}{=} pn H(X_1) - s_1 n -2\eta n-4\sqrt{n\alpha}-n (\alpha+\alpha'),
        \end{align}
        where $(a)$ follows from setting $\epsilon = 2^{-\alpha n}$ and  $\epsilon' = 2^{-\alpha' n}$ ($\epsilon$ and $\epsilon'$ are negligible in $n$). For any $\delta \geq (\alpha+\alpha' + 2\eta +4\sqrt{\alpha})> 0$, we have:
        \[
        H_\infty^{\epsilon+2\epsilon'}(\mathbf{X}_1|_{S_{1\overline{Z}_1}}|h_{1j}(R_{1j}, \mathbf{X}_1|_{S_{1\overline{Z}_1}}), \mathbf{Y}|_{S_{1\overline{Z}_1}}, R^{(1)}, T^{(1)})\geq pn H(X_1) - s_1 n -\delta n.
        \]
        From Lemma \ref{DLHL}, we know that if we set $r_1<pH(X_1)-s_1$ and appropriately choose the constant $\delta, \eta, \alpha$ and $\alpha'$, Bob can not obtain non-trivial information about the unselected string. The proof for Alice-2 is the same.

        \parbreak Now we consider \eqref{eq: third}. Since the $\text{SU-SBC}_{p,W}$ is i.i.d., \eqref{eq: third} can be written as: 
        \begin{align*}
            H^\epsilon_{\infty}(\mathbf{X}_1|_{S_{1\overline{Z}_1}}&,\mathbf{X}_2|_{S_{2\overline{Z}_2}}|h_{1j}(R_{1j}, \mathbf{X}_1|_{S_{1\overline{Z}_1}}), h_{2j}(R_{2j}, \mathbf{X}_2|_{S_{2\overline{Z}_2}}), \mathbf{Y}|_{S_{1\overline{Z}_1},S_{2\overline{Z}_2}}, R_{A_1}, R_{A_2})\notag\\
            & \overset{(a)}{\geq} H_\infty(\mathbf{X}_1|_{S_{1\overline{Z}_1}}|h_{1j}(R_{1j}, \mathbf{X}_1|_{S_{1\overline{Z}_1}}), h_{2j}(R_{2j}, \mathbf{X}_2|_{S_{2\overline{Z}_2}}), \mathbf{Y}|_{S_{1\overline{Z}_1},S_{2\overline{Z}_2}}, R_{A_1}, R_{A_2}) \notag\\
            & \quad + H_\infty^\epsilon(\mathbf{X}_2|_{S_{2\overline{Z}_2}}|\mathbf{X}_1|_{S_{1\overline{Z}_1}}, h_{1j}(R_{1j}, \mathbf{X}_1|_{S_{1\overline{Z}_1}}), h_{2j}(R_{2j}, \mathbf{X}_2|_{S_{2\overline{Z}_2}}), \mathbf{Y}|_{S_{1\overline{Z}_1},S_{2\overline{Z}_2}}, R_{A_1}, R_{A_2})\notag\\
            & \overset{(b)}{\geq} H^\epsilon_\infty(\mathbf{X}_1|_{S_{1\overline{Z}_1}}|h_{1j}(R_{1j}, \mathbf{X}_1|_{S_{1\overline{Z}_1}}), h_{2j}(R_{2j}, \mathbf{X}_2|_{S_{2\overline{Z}_2}}), \mathbf{Y}|_{S_{1\overline{Z}_1},S_{2\overline{Z}_2}}, R_{A_1}, R_{A_2}) \notag\\
            & \quad + H_\infty^\epsilon(\mathbf{X}_2|_{S_{2\overline{Z}_2}}|\mathbf{X}_1|_{S_{1\overline{Z}_1}}, h_{1j}(R_{1j}, \mathbf{X}_1|_{S_{1\overline{Z}_1}}), h_{2j}(R_{2j}, \mathbf{X}_2|_{S_{2\overline{Z}_2}}), \mathbf{Y}|_{S_{1\overline{Z}_1},S_{2\overline{Z}_2}}, R_{A_1}, R_{A_2}) - \log(\frac{1}{\epsilon}),\notag\\
        \end{align*}
        where $(a)$ is due to \eqref{eq: holenstein-1 + holenstein-1}, and $(b)$ is due to Lemma \ref{lemm: min-entropy vs smooth}. Applying \eqref{eq: holenstein-1} $\epsilon, \epsilon'>0$ for each terms, we have:
        \begin{align}\label{eq: sum rate}
            H^{\epsilon+\epsilon'}_\infty(\mathbf{X}_1|_{S_{1\overline{Z}_1}}&|h_{1j}(R_{1j}, \mathbf{X}_1|_{S_{1\overline{Z}_1}}), h_{2j}(R_{2j}, \mathbf{X}_2|_{S_{2\overline{Z}_2}}), \mathbf{Y}|_{S_{1\overline{Z}_1},S_{2\overline{Z}_2}}, R_{A_1}, R_{A_2}) \notag\\
            & \quad + H_\infty^{\epsilon+\epsilon'}(\mathbf{X}_2|_{S_{2\overline{Z}_2}}|\mathbf{X}_1|_{S_{1\overline{Z}_1}}, h_{1j}(R_{1j}, \mathbf{X}_1|_{S_{1\overline{Z}_1}}), h_{2j}(R_{2j}, \mathbf{X}_2|_{S_{2\overline{Z}_2}}), \mathbf{Y}|_{S_{1\overline{Z}_1},S_{2\overline{Z}_2}}, R_{A_1}, R_{A_2}) - \log(\frac{1}{\epsilon})\notag\\
            &\qquad\overset{(a)}{\geq} H_\infty(\mathbf{X}_1|_{S_{1\overline{Z}_1}}| \mathbf{Y}|_{S_{1\overline{Z}_1},S_{2\overline{Z}_2}}, R_{A_1}, R_{A_2})\notag\\
            & \quad\qquad+ H_\infty^\epsilon(h_{1j}(R_{1j}, \mathbf{X}_1|_{S_{1\overline{Z}_1}})|\mathbf{X}_1|_{S_{1\overline{Z}_1}},\mathbf{Y}|_{S_{1\overline{Z}_1},S_{2\overline{Z}_2}}, R_{A_1}, R_{A_2})\notag\\
            & \qquad\quad- H_0(h_{1j}(R_{1j}, \mathbf{X}_1|_{S_{1\overline{Z}_1}})|\mathbf{Y}|_{S_{1\overline{Z}_1},S_{2\overline{Z}_2}}, R_{A_1}, R_{A_2}) - \log(\frac{1}{\epsilon'})\notag\\
            &\qquad\quad + H_\infty(\mathbf{X}_2|_{S_{2\overline{Z}_2}}| \mathbf{Y}|_{S_{1\overline{Z}_1},S_{2\overline{Z}_2}}, R_{A_1}, R_{A_2})\notag\\
            & \quad\qquad+ H_\infty^\epsilon(h_{2j}(R_{2j}, \mathbf{X}_2|_{S_{2\overline{Z}_2}})|\mathbf{X}_1|_{S_{1\overline{Z}_1}},\mathbf{Y}|_{S_{1\overline{Z}_1},S_{2\overline{Z}_2}}, R_{A_1}, R_{A_2})\notag\\
            & \qquad\quad - H_0(h_{2j}(R_{2j}, \mathbf{X}_2|_{S_{2\overline{Z}_2}})|\mathbf{Y}|_{S_{1\overline{Z}_1},S_{2\overline{Z}_2}}, R_{A_1}, R_{A_2}) - \log(\frac{1}{\epsilon'})\notag\\
            & \qquad\overset{(b)}{\geq} H^\epsilon_{\infty}(\mathbf{X}_1|_{S_{1\overline{Z}_1}}|\mathbf{Y}|_{S_{1\overline{Z}_1},S_{2\overline{Z}_2}}) + H^\epsilon_{\infty}(\mathbf{X}_2|_{S_{2\overline{Z}_2}}|\mathbf{X}_1|_{S_{1\overline{Z}_1}},\mathbf{Y}|_{S_{1\overline{Z}_1},S_{2\overline{Z}_2}})\notag\\
            & \qquad\qquad\qquad\qquad\qquad\qquad\qquad\,\,\,\,
            - s_1 n -s_2 n-2\log(\frac{1}{\epsilon})-2
            \log(\frac{1}{\epsilon'}),
        \end{align}
        where $(a)$ is due to \eqref{eq: holenstein-1} and the independence of $X_i$ from $h_{\overline{i},j}$, and $(b)$ is due to the similar simplification as \eqref{eq: after applying holenstein-1}-\eqref{eq: middle}. The first term of \eqref{eq: sum rate} can be bounded similarly to \eqref{eq: middle+1}. For the second term, we have: 
        \begin{align}\label{eq: middle+1: sum rate}
            H_\infty^{\epsilon+\epsilon'}(\mathbf{X}_2|_{ S_{2\overline{Z}_2}}|\mathbf{X}_1|_{ S_{1\overline{Z}_1}}, \mathbf{Y}|_{S_{1\overline{Z}_1},S_{2\overline{Z}_2}})&\,
            \geq H_\infty^{\epsilon}(\mathbf{X}_2|_{\lvert S_{2\overline{Z}_2}\lvert}|\mathbf{X}_1|_{\lvert S_{1\overline{Z}_1}\rvert},\mathbf{V}|_{\lvert S_{1\overline{Z}_1},S_{2\overline{Z}_2}\lvert})\notag\\
            & \overset{(a)}{\geq} \lvert S_{2\overline{Z}_2}\rvert H(X_2|X_1, V)-4\sqrt{\lvert S_{2\overline{Z}_2}\rvert\log(1/\epsilon)}\log \lvert \mathcal{X}_2 \rvert\notag\\
            & \,\geq (p-\eta)n H(X_2|X_1, V)-4\sqrt{(p-\eta)n\log(1/\epsilon)}\notag\notag\\
            & \,\geq pn H(X_2|X_1, V)-\eta n H(X_2|X_1, V)-4\sqrt{(p-\eta)n\log(1/\epsilon)}\notag\notag\\
            & \,\geq pn (1-2\eta)H(X_2|X_1)+2\eta n H(X_2|X_1, Z)-\eta n-4\sqrt{(p-\eta)n\log(1/\epsilon)}\notag\notag\\
            & \,\geq pn H(X_2)-2\eta n-4\sqrt{n\log(1/\epsilon)},
        \end{align}
        where $(a)$ is due to Lemma \ref{lemm: Holenstein} and this fact that in Protocol 1, $\lvert S_{1Z_1} \rvert = \lvert S_{2Z_2} \rvert = (p-\eta)n$.
        
        Then by putting $\epsilon=\epsilon+\epsilon'$, \eqref{eq: sum rate} is as follows: \begin{align}\label{eq: sum rate-2}
        H^{\epsilon+2\epsilon'}_\infty(\mathbf{X}_1|_{S_{1\overline{Z}_1}}&|h_{1j}(R_{1j}, \mathbf{X}_1|_{S_{1\overline{Z}_1}}), h_{2j}(R_{2j}, \mathbf{X}_2|_{S_{2\overline{Z}_2}}), \mathbf{Y}|_{S_{1\overline{Z}_1},S_{2\overline{Z}_2}}, R_{A_1}, R_{A_2}) \notag\\
        & \quad + H_\infty^{\epsilon+\epsilon'}(\mathbf{X}_2|_{S_{2\overline{Z}_2}}|\mathbf{X}_1|_{S_{1\overline{Z}_1}}, h_{1j}(R_{1j}, \mathbf{X}_1|_{S_{1\overline{Z}_1}}), h_{2j}(R_{2j}, \mathbf{X}_2|_{S_{2\overline{Z}_2}}), \mathbf{Y}|_{S_{1\overline{Z}_1},S_{2\overline{Z}_2}}, R_{A_1}, R_{A_2}) - \log(\frac{1}{\epsilon})\notag\\
        & \,\qquad\geq pn \big(H(X_1)+H(X_2)\big)-4\eta n-8\sqrt{n\log(1/\epsilon)}-s_1 n-s_2 n - 2\log(\frac{1}{\epsilon})- 2\log(\frac{1}{\epsilon'})\notag\\
        & \qquad \overset{(a)}{\geq} pn \big(H(X_1)+H(X_2)\big)-s_1 n-s_2 n - 2\delta n,
        \end{align}
        for any $\delta \geq (\alpha+\alpha' + 2\eta +4\sqrt{\alpha})> 0$,  $(a)$ follows from setting $\epsilon = 2^{-\alpha n}$ and  $\epsilon' = 2^{-\alpha' n}$ ($\epsilon$ and $\epsilon'$ are negligible in $n$).

        From Lemma \ref{DLHL}, we know that if we set $r_1+r_2<p(H(X_1)+H(X_2))-s_1-s_2$ and appropriately choose the constant $\delta, \eta, \alpha$ and $\alpha'$, Bob can not obtain non-trivial information about the unselected strings. 
        
        \parbreak Now, we find the appropriate $s_1$ and $s_2$ under which the protocol remains correct and secure. Due to the Chernoff bound, we know that the probability of aborting the protocol by Bob in step $(2)$ tends to zero as $n\rightarrow \infty$. The protocol fails in step $(5)$, if there is more than one pair, such as $(\hat{\mathbf{x}}_1|_{S_{1Z_1}}, \hat{\mathbf{x}}_2|_{S_{2Z_2}})$ where $h_i(\mathbf{X}_i|_{S_{iZ_i}}) = h_i(\hat{\mathbf{X}}_i|_{S_{iZ_i}})$. We know that if all players are honest, then $\mathbf{Z}|_{S_{1Z_1},S_{2Z_2}} = \mathbf{Y}|_{S_{1Z_1},S_{2Z_2}}$ with probability exponentially close to one and the number of paired sequences $(\hat{\mathbf{x}}_1|_{S_{1Z_1}}, \hat{\mathbf{x}}_2|_{S_{2Z_2}})$ jointly typical with $\mathbf{z}|_{S_{1Z_1}, S_{2Z_2}}$ can be upper bounded as follows:
        
        \begin{itemize}
            \item If one of the sequences $\hat{\mathbf{x}}_i|_{S_{iZ_i}}$ is not typical with $(\mathbf{x}_{\overline{i}}|_{S_{\overline{i}Z_{\overline{i}}}}, \mathbf{z}|_{S_{1Z_1}, S_{2Z_2}})$: $2^{\lvert S_{iZ_i} \rvert(H(X_1,X_2,Z)- H(X_{\overline{i}},Z)+ \delta'}) = \linebreak 2^{np(H(X_i|X_{\overline{i}}, Z)+\delta')}$, $\delta'>0$. From \eqref{eq: hashs}, we know that $p\leq 2^{-s_in}2^{pn(H(X_i|X_{\overline{i}},Z)+\delta')}$, then $s_i > p H(X_i|X_{\overline{i}},Z)$.
            \item If both of the sequences $(\hat{\mathbf{x}}_1|_{S_{1Z_1}},\hat{\mathbf{x}}_2|_{S_{2Z_2}})$ are not typical with $ \mathbf{z}|_{S_{1Z_1}, S_{2Z_2}}$: $2^{\lvert S_{iZ_i} \rvert(H(X_1,X_2,Z)- H(Z)+ \delta'}) = \linebreak 2^{np(H(X_1,X_2,Z)- H(Z)+ \delta')}$, $\delta'>0$. From \eqref{eq: hashs-double}, we know that $p\leq 2^{-(s_1+s_2)n}2^{pn(H(X_1,X_2,Z)- H(Z)+ \delta')}$, then $s_1+s_2 > p (H(X_1,X_2,Z)- H(Z)+ \delta')$.
        \end{itemize}

        The final inner bound is as follows:
        \begin{align*}
            r_1 &< p \max_{P_{X_1}P_{X_2}} \big(H(X_1)-H(X_1|X_2, Z)\big) = \max_{P_{X_1}P_{X_2}} I(X_1;Y|X_2),\\
            r_2 &< p \max_{P_{X_1}P_{X_2}} \big(H(X_2)-H(X_2|X_1, Z)\big) = \max_{P_{X_1}P_{X_2}} I(X_2;Y|X_1),\\
            r_1+r_2&<p\max_{P_{X_1}P_{X_2}}\big(H(X_1)+H(X_2)+H(Z)-H(X_1, X_2,Z)\big) = \max_{P_{X_1}P_{X_2}} I(X_1, X_2;Y) .
        \end{align*}

        \parbreak

       \parbreak The lower and upper bounds coincide, then the capacity is proved. \qed
\section{Proof of Theorem \ref{thm: OT capacity-malicious}\label{App: proof-thm: OT capacity-malicious}}\label{{App: proof-thm: OT capacity-malicious}}}
    
The overall structure of the proof is the same as Theorem \ref{thm: OT capacity} wherein all parties are honest. Malicious Bob can benefit from the unfairness of the channel and deviate from the channel statistics in $\delta n$ positions without being detected. He tries to find the unselected strings from both senders. Thus, he could compute sets $S_{i0}$ and $S_{i1}$ so that non-erasures are distributed in both sets. With probability exponentially close to one, the total number of non-erasures Bob receives from each sender will be no larger than $(p+\eta)n$. Thus, for any strategy Bob distributes these non-erasures between two sets $S_{i0}$ and $S_{i1}$, the number of erasures in $S_{i\overline{Z}_i}$ is no less than $(p-\eta)n - \frac{p + \eta}{2}n = \frac{1}{2}(p-3\eta) n$.

\parbreak Again, we must bound \eqref{eq: first}-\eqref{eq: third}, to get three bounds on the lower bound. For \eqref{eq: first} and \eqref{eq: second}, all steps are the same until \eqref{eq: middle}. Let $V$ be an i.i.d. random variable so that $V = e$ (erasure) with probability $\frac12-\eta$ and $V = Z$ (the output of channel $W$ on input $(X_1,X_2)$). As the number of erasures in $S_{i\overline{Z}_i}$ is no less than $\frac{1}{2}(p-3\eta) n$ with negligible error probability and $\text{SU-SBC}_{p,W}$ being i.i.d., by taking the same steps as \eqref{eq: middle+1} for $S_{1\overline{Z}_1}$, we have:
    \begin{align}\label{eq: middle+1-malicious}
            H_\infty^{\epsilon+\epsilon'}(\mathbf{X}_1|_{ S_{1\overline{Z}_1}}|\mathbf{Y}|_{S_{1\overline{Z}_1}})&\geq H_\infty^{\epsilon}(\mathbf{X}_1|_{\lvert S_{1\overline{Z}_1}\lvert}|\mathbf{V}|_{\lvert S_{1\overline{Z}_1}\lvert})\notag\\
            & \geq \frac{p}{2}n\big( H(X_1)+H(X_1|Z)\big)-2\eta n-4\sqrt{n\log(1/\epsilon)}.
        \end{align}

        \parbreak Putting $\epsilon = \epsilon + \epsilon'$ in \eqref{eq: middle}, then putting \eqref{eq: middle+1-malicious} to \eqref{eq: middle}, we have: 
        \begin{align}
            H_\infty^{\epsilon+2\epsilon'}(\mathbf{X}_1|_{S_{1\overline{Z}_1}}|h_{1j}(R_{1j}, \mathbf{X}_1|_{S_{1\overline{Z}_1}}), \mathbf{Y}|_{S_{1\overline{Z}_1}}, R^{(1)}, T^{(1)})&\,\geq \frac{p}{2}n \big(H(X_1)+H(X_1|Z)\big)-2\eta n-4\sqrt{n\log(1/\epsilon)}\notag\\
            & \qquad\qquad\qquad\qquad\qquad\quad\,\,\,\,\,- s_1 n-\log(\frac{1}{\epsilon})-\log(\frac{1}{\epsilon'})\notag\\
            & \overset{(a)}{=} \frac{p}{2}n\big( H(X_1)+H(X_1|Z)\big)-2\eta n-4\sqrt{n\log(1/\epsilon)}\notag\\
            & \qquad\qquad\qquad\qquad\qquad\qquad - s_1 n -2\eta n-4\sqrt{n\alpha}-n (\alpha+\alpha')),
        \end{align}
        where $(a)$ follows from setting $\epsilon = 2^{-\alpha n}$ and  $\epsilon' = 2^{-\alpha' n}$ ($\epsilon$ and $\epsilon'$ are negligible in $n$). For any $\delta \geq (\alpha+\alpha' + 2\eta +4\sqrt{\alpha})> 0$, we have:
        \[
        H_\infty^{\epsilon+2\epsilon'}(\mathbf{X}_1|_{S_{1\overline{Z}_1}}|h_{1j}(R_{1j}, \mathbf{X}_1|_{S_{1\overline{Z}_1}}), \mathbf{Y}|_{S_{1\overline{Z}_1}}, R^{(1)}, T^{(1)})\geq\frac{p}{2}n \big(H(X_1)+H(X_1|Z)\big) - s_1 n -\delta n.
        \]
        From Lemma \ref{DLHL}, we know that if we set $r_1<\frac{p}{2}n \big(H(X_1)+H(X_1|Z)\big)-s_1$ and appropriately choose the constant $\delta, \eta, \alpha$ and $\alpha'$, Bob can not obtain non-trivial information about the unselected string. The proof for Alice-2 is the same.

        \parbreak Now consider \eqref{eq: sum rate} for the sum-rate. Similarly, it can be bounded from below:  
        \begin{align*}
            H^\epsilon_{\infty}(\mathbf{X}_1|_{S_{1\overline{Z}_1}}|\mathbf{Y}|_{S_{1\overline{Z}_1},S_{2\overline{Z}_2}}) & + H^\epsilon_{\infty}(\mathbf{X}_2|_{S_{2\overline{Z}_2}}|\mathbf{X}_1|_{S_{1\overline{Z}_1}},\mathbf{Y}|_{S_{1\overline{Z}_1},S_{2\overline{Z}_2}})\notag - s_1 n -s_2 n-2\log(\frac{1}{\epsilon})-2
            \log(\frac{1}{\epsilon'})\\
            &\geq \frac{p}{2}n \big(H(X_1)+H(X_2)+H(X_1|Z)+H(X_2|X_1, Z)\big) -4\eta n-8\sqrt{n\log(1/\epsilon)}-s_1 n-s_2 n\\
            & \qquad\qquad\qquad\qquad\qquad\qquad\qquad\qquad\qquad\qquad\qquad\qquad\quad\,\,\,\,- 2\log(\frac{1}{\epsilon})- 2\log(\frac{1}{\epsilon'})\notag\\
            & \overset{(a)}{\geq} \frac{p}{2}n \big(H(X_1)+H(X_2)+H(X_1|Z)+H(X_2|X_1, Z)\big)-s_1 n-s_2 n - 2\delta n,
        \end{align*}
        for any $\delta \geq (\alpha+\alpha' + 2\eta +4\sqrt{\alpha})> 0$,  $(a)$ follows from setting $\epsilon = 2^{-\alpha n}$ and  $\epsilon' = 2^{-\alpha' n}$ ($\epsilon$ and $\epsilon'$ are negligible in $n$).

        \parbreak Up to now, we proved that if $r_i>\frac{p}{2} n\big(H(X_i) + H(X_i|Z)\big)-s
        _i n -\delta n$ and $r_1 + r_2>\frac{p}{2} n\big(H(X_1) + H(X_1|Z)+H(X_2) + H(X_2|X_1, Z)\big)-s_1
        n - s_2 n -2\delta n$, then Protocol 1 is private against malicious Bob. Since we proved for $s_i > p H(X_i|X_{\overline{i}},Z)$ and $s_1+s_2 > p (H(X_1,X_2,Z)- H(Z)+ \delta^{\prime})$, Protocol 1 is correct for honest players, Then the above region can be written as: 
        \begin{align*}
            R_1 &<  \frac12\max_{P_{X_1}P_{X_2}}\big\{ I(X_1;Y|X_2)+I(X_1;X_2|Y)\big\},\\
            R_2 &<  \frac12\max_{P_{X_1}P_{X_2}}\big\{ I(X_2;Y|X_1)+I(X_1;X_2|Y)\big\},\\
            R_1+R_2 &<  \frac12\max_{P_{X_1}P_{X_2}} I(X_1, X_2;Y),
        \end{align*}

        for some distribution $p(x_1)p(x_2)$ on $\mathcal{X}_1\times\mathcal{X}_2$. This completes the proof.

\bibliography{references}      

\end{document}